\documentclass{amsart}
\usepackage[latin1]{inputenc}
\usepackage{latexsym}
\usepackage{amsmath}
\usepackage{amssymb}
\usepackage{amsfonts}
\usepackage{mathrsfs}               
\usepackage{bbm}                    
\usepackage{bbold}                  
\usepackage{textcomp}               
\usepackage{amstext}
\usepackage{enumerate}
\usepackage{units}
\usepackage{multicol}
 \newtheorem{thm}{Theorem}[section]
 \newtheorem{cor}[thm]{Corollary}
 \newtheorem{lem}[thm]{Lemma}
 \newtheorem{prop}[thm]{Proposition}
 \theoremstyle{definition}
 
 \theoremstyle{remark}
 \newtheorem{rem}[thm]{Remark}
 \newtheorem{assumption}[thm]{Assumption}
 \numberwithin{equation}{section}
\newcommand{\CC}{\mathbb{C}}

\newcommand{\NN}{\mathbb{N}}

\newcommand{\RR}{\mathbb{R}}

\newcommand{\ZZ}{\mathbb{Z}}
\newcommand{\supp}{\mathrm{supp}}

\newcommand{\Ran}{\mathrm{Ran}}
\newcommand{\tr}{\mathrm{Tr}}

\newcommand{\sgn}{\mathrm{sgn}}

\newcommand{\id}{\mathbbm{1}}
\newcommand{\ve}{\varepsilon}
\newcommand{\vp}{\varphi}
\newcommand{\vo}{\varpi}
\newcommand{\vk}{\varkappa}

\newcommand{\vt}{\vartheta}

\newcommand{\wt}[1]{\widetilde{#1}}
\newcommand{\SL}{\langle}                          
\newcommand{\SR}{\rangle}    
\newcommand{\SPn}[2]{\langle \,#1\,|\,#2\, \rangle} 
\newcommand{\SPb}[2]{\big\langle \,#1\,\big|\,#2\, \big\rangle} 
\newcommand{\SPB}[2]{\Big\langle \,#1\,\Big|\,#2\, \Big\rangle}
\newcommand{\ol}[1]{\overline{#1}} 
\newcommand{\mr}[1]{\mathring{#1}}
\newcommand{\wh}[1]{\widehat{#1}}
\newcommand{\nf}[2]{\nicefrac{#1}{#2}}
\newcommand{\bigO}{\mathcal{O}}
\newcommand{\Jy}{\langle y\rangle}
\newcommand{\cA}{\mathcal{A}}\newcommand{\cN}{\mathcal{N}}
\newcommand{\cB}{\mathcal{B}} 
\newcommand{\cC}{\mathcal{C}}
\newcommand{\cD}{\mathcal{D}} 
\newcommand{\cQ}{\mathcal{Q}}
\newcommand{\cR}{\mathcal{R}}
\newcommand{\cS}{\mathcal{S}}
\newcommand{\cH}{\mathcal{H}}\newcommand{\cT}{\mathcal{T}}

\newcommand{\cV}{\mathcal{V}}

\newcommand{\sA}{\mathscr{A}}
 
\newcommand{\sC}{\mathscr{C}}

\newcommand{\sF}{\mathscr{F}}
\newcommand{\sG}{\mathscr{G}}
\newcommand{\sH}{\mathscr{H}}

\newcommand{\sK}{\mathscr{K}}
\newcommand{\sL}{\mathscr{L}}\newcommand{\sX}{\mathscr{X}}

\newcommand{\fA}{\mathfrak{A}}

\newcommand{\fW}{\mathfrak{W}}

\renewcommand{\le}{\leqslant} 
\renewcommand{\ge}{\geqslant}
\renewcommand{\imath}{i}
\newcommand{\V}[1]{\mathbf{#1}}
\newcommand{\vsigma}{\boldsymbol{\sigma}}
\newcommand{\valpha}{\boldsymbol{\alpha}}

\newcommand{\veps}{\boldsymbol{\varepsilon}}

\newcommand{\gap}{\mathrm{gap}}
\newcommand{\const}{\mathfrak{c}}
\newcommand{\ren}{\mathrm{ren}}
\newcommand{\Spec}{\mathrm{Spec}}
\newcommand{\Pf}{\mathbf{p}_{\mathrm{f}}}
\newcommand{\Ham}{\cH}
\newcommand{\Hf}{H_{\mathrm{f}}}
\newcommand{\chHf}{\check{H}_{\mathrm{f}}}
\newcommand{\ee}{\mathfrak{e}}
\newcommand{\dom}{\cD}
\newcommand{\fdom}{\cQ}
\newcommand{\HR}{\mathscr{H}}
\newcommand{\HP}{\mathscr{K}}
\newcommand{\Fock}[2]{\mathscr{F}_{#1}^{#2}}
\newcommand{\LO}{\mathscr{L}}
\newcommand{\ad}{a^\dagger}
\newcommand{\Tr}{\mathrm{Tr}}
\newcommand{\HS}{\mathrm{HS}}
\newcommand{\pmax}{\mathfrak{p}}
\newcommand{\qmax}{\mathfrak{q}}
\renewcommand{\Im}{\mathrm{Im}\,}
\renewcommand{\Re}{\mathrm{Re}\,}
\newcommand{\D}{D}
\newcommand{\R}{R}
\newcommand{\Sgn}{{S}}
\newcommand{\HAM}{{H}}
\newcommand{\RES}{{\mathcal{R}}}
\begin{document}

\title[The mass shell in the semi-relativistic Pauli-Fierz model]
 {The mass shell in the semi-relativistic\\ Pauli-Fierz model}

\author[Martin K\"onenberg {\protect \and} Oliver Matte]{Martin K\"onenberg {\protect
\and} Oliver Matte}

\begin{abstract}
We consider the semi-relativistic
Pauli-Fierz model for a single free electron interacting
with the quantized radiation field.
Employing a variant of Pizzo's iterative analytic perturbation theory
we construct a sequence of ground state eigenprojections of
infra-red cutoff, dressing transformed fiber Hamiltonians
and prove its convergence, as the cutoff goes to zero. 
Its limit is the ground state eigenprojection of
a certain Hamiltonian unitarily equivalent
to a renormalized fiber Hamiltonian acting in a
coherent state representation space.
The ground state energy is
an exactly two-fold degenerate eigenvalue of the
renormalized Hamiltonian, while it is not
an eigenvalue of the original fiber Hamiltonian 
unless the total momentum is zero.
These results hold true, for total momenta
inside a ball about zero of arbitrary radius
$\pmax>0$, provided that the coupling constant is
sufficiently small depending on $\pmax$ and
the ultra-violet cutoff.
Along the way we prove 
twice continuous differentiability and strict
convexity of the ground state energy
as a function of the total momentum inside that ball.

\smallskip

\end{abstract}

\maketitle
\setcounter{tocdepth}{1}
\tableofcontents
%
%
%

\section{Introduction and main results}\label{sec-intro}

\subsection{The general framework}

\noindent
The mathematically rigorous investigation of the infra-red (IR) problem 
in non-relativistic (NR) quantum electrodynamics (QED) 
has recently undergone some substantial progress. Notably,
infra-particle scattering states in the one-electron sector
have been constructed for the NR Pauli-Fierz model
in \cite{CFP2010} based on investigations of the corresponding
mass shell in \cite{CFP2009}. 
We recall that, by translation invariance, the NR Pauli-Fierz
Hamiltonian for a free electron 
interacting with the quantized radiation field
admits a fiber decomposition with respect to
the total momentum of the combined electron-photon system.
The mass shell is the ground state energy of the fiber
Hamiltonians considered as a function of the total momentum.
Another important recent result is the existence of the renormalized electron
mass. This has been established in \cite{FroehlichPizzo2010} where
the authors have been able to prove
twice continuous differentiability and
strict convexity of the mass shell in some ball about the origin.
Earlier works already provided {bounds} on the renormalized electron
mass \cite{BCFS2006,Chen2008}. The results of \cite{Chen2008} have been
used to discuss coherent 
IR representations in NR QED \cite{ChenFroehlich2007}.
All these results have been obtained at fixed ultra-violet cutoff
and for sufficiently small coupling constants.

While a general guideline for the mathematical treatment of the IR
problem and in particular of infra-particle scattering
has been settled long ago in a study of the
Nelson model \cite{Froehlich1973,Froehlich1974}, the recent progress in QED
is due to the development of new, sophisticated multi-scale techniques.
For instance, the analysis in \cite{BCFS2006,Chen2008} is based on the 
spectral renormalization group introduced by Bach, Fr\"ohlich, and Sigal.
A second method is the
{\em iterative analytic perturbation theory (IAPT)} 
developed mainly by Pizzo in his analysis of Nelson's model
\cite{Pizzo2003,Pizzo2005}. 
The latter method is employed in the papers \cite{CFP2009,FroehlichPizzo2010}
mentioned above and in
\cite{BFP2006,BFP2007,BFP2009} to provide expansions
of atomic ground state eigenvalues and eigenvectors
and of scattering amplitudes. Recently
the removal of the {ultra-violet}
cutoff in Nelson's model has been studied by means of the IAPT \cite{BDP2011}.
The general principles of the IAPT also serve as the starting point
of our own analysis.
An alternative procedure employing continuous flows 
to remove an artificial IR cutoff instead of the
discrete iteration steps used in the IAPT
can be found in \cite{BachKoenenberg2006}.
For a discussion of the physical background
and remarks on the historical development
of the analysis of the IR problem we refer the reader
to \cite{CFP2010,Froehlich1973}.

The objective of the present article is to establish
several of the afore-mentioned results in the {\em semi}-relativistic (SR)
Pauli-Fierz model, which is obtained by replacing
the NR kinetic energy term in the 
Pauli-Fierz operator by a square-root operator.
Its mathematical analysis has been
initiated in \cite{MiyaoSpohn2009},
where the bottom of the essential spectrum of the fiber Hamiltonians
is characterized. By addition of an
exterior Coulomb potential, $-\ee^2Z/|\V{x}|$ (with
$\ee^2Z\in[0,\nf{2}{\pi}]$), 
one may also define a semi-relativistic model for a hydrogen-like atom
coupled to the quantized radiation field.
Binding energies, exponential localization of low-lying
spectral subspaces, and the existence of ground states
of this atomic model have been studied in 
\cite{HiroshimaSasaki2010,KoenenbergMatte2011,KoenenbergMatte2012b,KMS2009a,MatteStockmeyer2009a}. A scalar square-root Hamiltonian appeared earlier
in the mathematical analysis of Rayleigh scattering \cite{FGS2001}. 
Notice that the vector potential is introduced via minimal coupling
in the SR Pauli-Fierz Hamiltonian.
Fiber Hamiltonians with a relativistic kinetic energy
for the matter particles and linearly coupled radiation fields
appear, e.g., in \cite{Froehlich1973,Moeller2005}.

A second purpose of our article is to
propose several new arguments or alterations of earlier ones
within the general frame of Pizzo's IAPT, in particular in the study
of the convergence of IR regularized ground state projections
and in the analysis of Hellmann-Feynman type formulas
for the derivatives of ground state energies.
(To mention some keywords for the experts: 
We do not employ contour integrals and avoid
repeated Neumann series expansions and
certain bounds relating expectations of operators
with expectations of their absolute values;
by a minor modification of the dressing transforms
we avoid the discussion of intermediate Hamiltonians.)
Although one might expect the analysis of square-root Hamiltonians to be
technically more involved we are able to establish essentially
all main results of \cite{ChenFroehlich2007,CFP2009,FroehlichPizzo2010}
in the semi-relativistic case.
We hope that some of our observations will be helpful in
forthcoming investigations including the NR case.

Another novelty achieved here is the study
of the {\em multiplicity} of the ground state eigenvalue of certain
renormalized fiber Hamiltonians in the presence of spin. 
(In the scalar Nelson model one may apply Perron-Frobenius
arguments to show non-degeneracy of ground states \cite{Froehlich1973}.)
As our corresponding argument is
essentially based on a certain relative form bound
required to get the IAPT started, it seems clear that
it also applies {\em mutatis mutandis} to the NR Pauli-Fierz 
and Nelson models.


\subsection{The model and main results}

\noindent
In this subsection we explain the model 
under investigation and
state our main results.
We also give a few comments on the proofs
and on the organization of this article.

The semi-relativistic Pauli-Fierz Hamiltonian
for a single free electron interacting with the quantized
radiation field is given by
\begin{equation}\label{def-SRPF-Ham}
\mathbb{H}_{\mathrm{sr}}:=
\sqrt{(\vsigma\cdot(-i\nabla_{\V{x}}\otimes\id+\ee\,\mathbb{A}))^2+\id}
+\id\otimes\Hf\,.
\end{equation}
It is acting in the tensor product 
$L^2(\RR^3_\V{x},\CC^2)\otimes\Fock{}{}
=\int_{\RR^3}^\oplus\CC^2\otimes\Fock{}{}\,d^3\V{x}$, 
where the bosonic Fock space $\Fock{}{}$
is modeled over the one-photon space $\HP:=L^2(\RR^3_{\V{k}}\times\ZZ_2)$;
see Appendix~\ref{app-Fock} for the definition of Fock spaces
and operators acting in them.
The vector
$\vsigma=(\sigma_1,\sigma_2,\sigma_3)$ 
contains the three Pauli spin matrices,
$\Hf$ is the radiation field energy,
and the parameter $\ee>0$ is eventually assumed to be small. 
For a single photon state $f\in\HP$, let $\ad(f)$ and $a(f)$ 
denote the standard
bosonic creation and annihilation operators acting in $\Fock{}{}$ and 
satisfying the following canonical commutation relations (CCR)
on some suitable dense domain,
\begin{equation}\label{CCR}
[a^\sharp(f),a^\sharp(g)]=0\,,\quad
[a(f),\ad(g)]=\SPn{f}{g}\,\id\,,\quad f,g\in\HP.
\end{equation}
Here $a^\sharp$ is $\ad$ or $a$. Denote the field operators as
\begin{equation}\label{viona}
\vp(f):=2^{-\nf{1}{2}}(\ad(f)+a(f))\,,\qquad
\vo(f):=2^{-\nf{1}{2}}(i\ad(f)-i a(f))\,,
\end{equation}
and write $\vp(\V{f}):=(\vp(f_1),\vp(f_2),\vp(f_3))$,
for a vector of single photon states $\V{f}=(f_1,f_2,f_3)$.
(More precisely, \eqref{viona} defines essentially self-adjoint operators
whose closures are henceforth denoted again by the same symbols.)
In this notation
the quantized vector potential is given as 
$\mathbb{A}:=\int_{\RR^3}^\oplus\id_{\CC^2}\otimes\mathbb{A}(\V{x})\,d^3\V{x}$
with
\begin{align}\label{def-AA}
\mathbb{A}(\V{x}):=\vp(e^{-i\V{k}\cdot\V{x}}\V{G})\,,
\qquad\V{G}(\V{k},\lambda):=
(2\pi)^{-\nf{3}{2}}\,\id_{|\V{k}|<\kappa}\,
|\V{k}|^{-\nf{1}{2}}\veps_\lambda(\V{k})\,.
\end{align}
Together with $\V{k}/|\V{k}|$ the two polarization vectors
$\veps_0(\V{k})$ and $\veps_1(\V{k})$ appearing here 
form an orthonormal
basis of $\RR^3$, for almost every $\V{k}$.
The number $\kappa>0$ is an ultra-violet cut-off parameter.
As $\kappa\to\infty$, the values of $\ee>0$
covered by our main results will go to zero.

As the operator given by \eqref{def-SRPF-Ham} is invariant
under space translations of the total electron-photon system
it is unitarily equivalent to a direct integral
\begin{align}\label{def-Ham(P)}
\mathbb{H}_{\mathrm{sr}}\cong
\int_{\RR^3}^\oplus\Ham(\V{P})\,d^3\V{P}\,,
\quad\Ham(\V{P}):=\sqrt{(\vsigma\cdot(\V{P}-\Pf+\ee\V{A}))^2+\id}+\Hf\,.
\end{align}
The vector $\V{P}\in\RR^3$ is interpreted as the total momentum
of the combined electron-photon system.
Moreover, $\V{A}:=\vp(\V{G})$, and $\Pf$
is the photon momentum operator.
The present article deals with the spectral analysis
of the fiber Hamiltonians $\Ham(\V{P})$ acting in $\CC^2\otimes\Fock{}{}$, whose
mathematically precise definition 
is discussed in Section~\ref{sec-rb}. 
(For a precise definition of $\mathbb{H}_{\mathrm{sr}}$ we refer to
\cite{HiroshimaSasaki2010,KMS2010,MatteStockmeyer2009a}; in the
present paper we study only the fiber Hamiltonians explicitly.)
In
Section~\ref{sec-rb} it turns out that $\Ham(\V{P})$ is self-adjoint
on the domain of $\Hf$, for all $\ee,\kappa>0$.
(This improves on a result in \cite{MiyaoSpohn2009}
where $\ee$ and/or $\kappa$ are assumed to be small.)
Our aim is to analyze the infimum of its spectrum,
$$
E(\V{P}):=\inf\Spec[\Ham(\V{P})]\,,
$$
as a function of $\V{P}\in\RR^3$, and to address the question whether
it be a two-fold degenerate eigenvalue or not.  
Notice that the shape of $E(\V{P})$ is not fixed by relativity
because of the ultra-violet cutoff.
We collect our first main results in the following theorem.
It applies to total momenta contained in balls
\begin{equation}\label{def-Bp}
\cB_\pmax:=\{\V{P}\in\RR^3:\,|\V{P}|<\pmax\}\,,
\quad 0<\pmax<\infty\,,
\end{equation}
of bounded but arbitrary large radius. 
The possibility to choose $\pmax$ large is due to the
semi-relativistic nature of our model.
In the non-relativistic Pauli-Fierz model Part~(1) of the next theorem
would be false, if $\pmax$ were too large.

\begin{thm}\label{thm-intro1}
For all $\kappa,\pmax>0$, we find $\ee_0>0$ such that
the following holds, for all $\ee\in(0,\ee_0]$:

\smallskip

\noindent(1) $E$ is twice continuously differentiable
and strictly convex on $\cB_\pmax$ and it attains its
unique global minimum at zero, 
$E(\V{0})=\inf\{E(\V{P}):\,\V{P}\in\RR^3\}$.

\smallskip

\noindent(2) $E(\V{0})$ is a two-fold degenerate eigenvalue
of $\Ham(\V{0})$. The expectation value of the photon number operator
in any corresponding eigenstate is finite.  

\smallskip

\noindent(3) If $\V{P}\in\cB_\pmax\setminus\{\V{0}\}$,
then $E(\V{P})$ is not an eigenvalue of $\Ham(\V{P})$.
\end{thm}

\begin{proof}
The assertions are contained in
Theorems~\ref{thm-ren-mass} and~\ref{mainthm},
Corollary~\ref{cor-N}, and Theorem~\ref{thm-absence}; see also \eqref{block}
and the paragraph preceding it.
\end{proof}

\begin{rem} Here we comment on the three items of the above theorem:

\smallskip

\noindent
(1): It is possible to show that the curvature of the
mass shell is strictly smaller than the one of the mass shell
with $\ee=0$, i.e. of $\sqrt{\V{P}^2+1}$. In other words,
the renormalized electron mass is strictly larger than its bare mass $1$.
We shall prove this assertion in a separate paper \cite{KoenenbergMatte2012b}.

As explained below we approximate $E$ by a sequence, $\{E_j\}_j$,
of ground state energies of IR cutoff operators to prove (1).
Estimates on the convergence rates for the
derivatives of $E_j$ are stated in Theorem~\ref{thm-ren-mass}.

The analog of (1) for the NR Pauli-Fierz model has been proven
in \cite{FroehlichPizzo2010}.

\smallskip

\noindent(2): It has already been observed  in \cite{MiyaoSpohn2009} that
every (speculative) eigenvalue of $\Ham(\V{P})$ is {\em at least}
two-fold degenerated. In fact, this follows from Kramer's degeneracy theorem
as there is an anti-linear involution commuting with $\Ham(\V{P})$.
Part~(2) of Theorem~\ref{thm-intro1} is actually a special case
of Theorem~\ref{thm-intro2}(1) below.

In the NR case, the existence of ground state eigenvectors at $\V{P}=\V{0}$
is proven in \cite{CFP2009}. It also follows from Chen's spectral
renormalization group analysis \cite{Chen2008}.

\smallskip

\noindent(3): The asserted phenomenon is called
{\em absence of a proper mass shell}; 
following \cite{Schroer1963}, one also says the system exhibits
an {\em infra-particle behavior}.
 
In the NR Pauli-Fierz model it is known
that, {\em if} $E$ is differentiable at some arbitrary $\V{P}\in\RR^3$,
then for $E(\V{P})$ to be an eigenvalue it is necessary that
$\nabla E(\V{P})=\V{0}$ \cite{HaslerHerbst2008}.
In our case 
$\nabla E\not=\V{0}$
on $\cB_\pmax\setminus\{\V{0}\}$ and $\nabla E(\V{0})=\V{0}$
by Theorem~\ref{thm-intro1}(1).
The exact NR analog of (3) follows from \cite{ChenFroehlich2007}.
\hfill$\Diamond$
\end{rem}

\smallskip

\noindent
The IAPT applied in this paper
comprises an inductive spectral analysis of a certain
sequence of Hamiltonians whose interaction term is cut off
in the infra-red. Thus,
to obtain our results we introduce a sequence of IR cutoff parameters\footnote{One could choose a slightly larger value than
$\nf{1}{2}$ in the definition of $\rho_j$; smaller choices
would restrict the range of allowed values for $\ee$.},  
\begin{align}\label{def-rhoj}
\rho_j:=\kappa\,(\nf{1}{2})^j,\quad j\in\NN_0=\{0,1,2,\ldots\;\}\,,\quad
\rho_\infty:=0\,,
\end{align}
and approximate $\Ham(\V{P})$ by IR cutoff Hamiltonians,
$$
\Ham_j^\infty(\V{P}):=\sqrt{(\vsigma\cdot(\V{P}-\Pf+\ee\smash{\V{A}_j}))^2+\id}
+\Hf\,,\quad\V{A}_j:=\vp(\id_{\rho_j\le|\V{k}|}\V{G})\,.
$$
Hence, the set of non-zero photon momenta is split
into a sequence of annuli,
\begin{equation}\label{def-Akj}
\dot{\RR}^3_{\V{k}}=\cA_0\cup\bigcup_{j\in\NN_0}\cA_j^{j+1}
,\quad
\cA_k^j:=\{\rho_j\le|\V{k}|<\rho_k\},\quad
\cA_j:=\{\rho_j\le|\V{k}|\},
\end{equation}
with corresponding one-photon Hilbert spaces,
\begin{equation}\label{def-HPj}
\HP_k^j:=L^2(\cA_k^j\times\ZZ_2)\,,\qquad
\HP_j:=L^2(\cA_j\times\ZZ_2)\,.
\end{equation}
It is crucial for the whole analysis that the orthogonal
splittings
$$
\HP=\HP_j\oplus\HP_j^\infty,\qquad\HP_{j+1}=\HP_{j}\oplus\HP_j^{j+1},
$$
give rise to the following isomorphisms of Fock spaces,
\begin{equation}\label{def-Fockkj}
\Fock{}{}=\Fock{j}{}\otimes\Fock{j}{\infty},
\qquad\Fock{j+1}{}=\Fock{j}{}\otimes\Fock{j}{j+1}.
\end{equation}
Here $\Fock{k}{j}$ is the bosonic Fock space modeled over
$\HP_k^j$ and $\Fock{j}{}$ the one over $\HP_j$; see Appendix~\ref{app-Fock}.

In the following discussion we assume that $\pmax>0$ is fixed, $\V{P}\in\cB_\pmax$,
and $\ee>0$ is sufficiently small.
As it turns out, then, in Section~\ref{sec-gap}, suitable versions of
$\Ham_j^\infty(\V{P})$ acting only in 
$\CC^2\otimes\Fock{j}{}$ (resp. $\CC^2\otimes\Fock{j+1}{}$)
-- below denoted by $\Ham_j(\V{P})$ (resp. $\Ham_j^{j+1}(\V{P})$) -- 
have a strictly positive spectral gap above their twice degenerate
ground state eigenvalue
$$
E_j(\V{P}):=\inf\Spec[\Ham_j^\infty(\V{P})]=\inf\Spec[\Ham_j(\V{P})]=
\inf\Spec[\Ham_j^{j+1}(\V{P})]\,.
$$
Moreover, $E_j(\V{P})\to E(\V{P})$ and $E_0(\V{P})=\sqrt{\V{P}^2+1}$.
Furthermore, $E_j$ is real analytic on $\cB_\pmax$, so that its derivatives
can be represented by Hellmann-Feynman (HF) type formulas; 
see Section~\ref{sec-P-dep}.
In Section~\ref{sec-C2}
the latter formulas are used to show that 
$\sum_j(E_j-E_{j+1})$ is absolutely convergent w.r.t. the norm on
$C^2_b(\cB_\pmax)$, so that $E\in C^2_b(\cB_\pmax)$.
To estimate $E_j-E_{j+1}$, one employs HF formulas
involving $\Ham_j^{j+1}(\V{P})$ for the derivatives of $E_j$
and represents derivatives of $E_{j+1}$ by means of a (partially)
{\em dressing transformed} version of $\Ham_{j+1}(\V{P})$,
namely
\begin{equation}\label{def-U-intro}
\check{\Ham}_{j+1}(\V{P}):=U_j(\V{P})\,\Ham_{j+1}(\V{P})\,U_j^*(\V{P})\,,
\quad U_j(\V{P}):=e^{-i\ee\vo(f_j(\V{P}))},
\end{equation}
where the coherent factor is given by the formula
determined in \cite{ChenFroehlich2007},
\begin{equation}\label{def-ff}
f_j(\V{P};\V{k},\lambda):={\id_{\rho_{j+1}\le|\V{k}|<\rho_j}}\,
\frac{\V{G}(\V{k},\lambda)\cdot\nabla E_j(\V{P})}{
|\V{k}|-\V{k}\cdot\nabla E_j(\V{P})}\,,\quad j\in\NN_0\,.
\end{equation}
Roughly speaking, an application of $U_j^*(\V{P})$ to ground state
eigenvectors of $\Ham_j^{j+1}(\V{P})$, which look like
$\Psi_j\otimes\Omega_j^{j+1}$ ($\Omega_j^{j+1}=$ vacuum in $\Fock{j}{j+1}$),
turns them into better test functions for $\Ham_{j+1}(\V{P})$
by dressing $\Omega_j^{j+1}$ in a cloud of soft photons.
The soft photon cloud is due to the movement of the total
system at momentum $\V{P}$ and, accordingly,
$U_j(\V{0})=\id$. We remark that the formula for $f_j$
is essentially determined by the IR bound \eqref{rosa0} below.

The analysis of $E_j-E_{j+1}$ by means of HF formulas also
furnishes convergence of the ground state eigenprojections
of the Hamiltonians
\begin{equation}
\wt{\Ham}_j^\infty(\V{P}):=
W_j(\V{P})\,\Ham_j^\infty(\V{P})\,W_j(\V{P})^*,\qquad
W_j(\V{P}):=\prod_{\ell=0}^{j-1}U_\ell(\V{P})\,.
\end{equation}
The operators $\wt{\Ham}_j^\infty(\V{P})$
are self-adjoint on the domain of $\Hf$ and converge in the norm
resolvent sense to some limit
Hamiltonian denoted by 
$$
\wt{\Ham}_\infty(\V{P}):=\underset{j\to\infty}{\textrm{norm-res-lim}}\,
\wt{\Ham}_j^\infty(\V{P})\,.
$$
For $\V{P}\not=\V{0}$,
this limit Hamiltonian is, however, not a unitary transform of 
our original Hamiltonian $\Ham(\V{P})$, although the latter
is the norm resolvent limit of $\{\Ham_j^\infty(\V{P})\}_j$. In fact,
if $\V{P}\in\cB_\pmax$ is non-zero, then
$\lim_{j\to\infty}\nabla E_j(\V{P})=\nabla E(\V{P})\not=\V{0}$,
and, hence, the sum $\sum_0^\infty f_j$
is {\em not} square-integrable.
As a consequence we cannot define a reasonable limit
of the unitaries $W_j$ in $\cB(\CC^2\otimes\Fock{}{})$.
However, we may define a unitary operator,
$$
W_\infty^*(\V{P}):\,\CC^2\otimes\Fock{}{}\longrightarrow\HR^\ren_{\V{P}}
:=\CC^2\otimes\Fock{0}{}\otimes\prod\otimes_{j\in\NN_0}^{\Omega^\ren_{\V{P}}}
\Fock{j}{j+1}\,,
$$
with values in some incomplete direct product
space (in the sense of \cite{vNeumann1939})
associated with 
the coherent state 
$$
\Omega^\ren_{\V{P}}:=\eta\otimes e^{i\ee\vo(f_0(\V{P}))}\Omega_0^1
\otimes e^{i\ee\vo(f_1(\V{P}))}\Omega_1^2 
\otimes e^{i\ee\vo(f_2(\V{P}))}\Omega_2^3
\otimes\ldots\in\HR^\ren_{\V{P}},
$$
where $\eta$ is a normalized vector
in the vacuum sector of $\CC^2\otimes\Fock{0}{}$.
Recall that
$\HR^\ren_{\V{P}}$ is a (topologically complete) subspace
of the complete direct product
space 
$\wh{\HR}:=\CC^2\otimes\Fock{0}{}\otimes\prod\otimes_{j\in\NN_0}\Fock{j}{j+1}$ 
\cite{vNeumann1939}.
(If we changed $\Omega^\ren_{\V{P}}$ slightly by replacing $E_j$
in \eqref{def-ff} by $E$, then it would still determine the 
same space $\HR^\ren_{\V{P}}$; see Remark~\ref{rem-Omega-ren}.)
We may then define a renormalized Hamiltonian,
\begin{equation}\label{def-Ham-ren}
\Ham^\ren(\V{P}):=W_\infty^*(\V{P})\,\wt{\Ham}_\infty(\V{P})\,W_\infty(\V{P})\,,
\qquad W_\infty(\V{P}):=W_\infty^*(\V{P})^*,
\end{equation}
which is an analog of the operators introduced in
\cite[\textsection 2.3]{Froehlich1973} for the Nelson model.
It can be used to describe {\em coherent IR representations} of
the $C^*$-algebra $\fA$, where
\begin{align}\label{def-fA}
\fA:=\ol{\fA^\circ}^{\|\cdot\|_{\cB(\CC^2\otimes\Fock{}{})}},
\qquad\fA_j:=\id\otimes\cB(\Fock{j}{j+1})\otimes\id\,,\;j\in\NN_0\,,
\end{align}
and $\fA^\circ$ is the $*$-sub-algebra of $\cB(\CC^2\otimes\Fock{}{})$
generated by all 
$\fA_j$, $j\in\NN_0$, and by $\fA_{-1}:=\cB(\CC^2\otimes\Fock{0}{})\otimes\id$.
Let $\fW_j$ denote the Weyl algebra generated by the 
operators $W(v):=e^{i\vo(\Im v_j)}e^{i\vp(\Re v_j)}$ with $v\in\sK_j^{j+1}$,
interpreted as a sub-algebra of $\fA_j$.
Let $\fW_{-1}\subset\fA_{-1}$ be the $C^*$-algebra generated
by all $M\otimes W(u)$ with $M\in\cB(\CC^2)$ and $u\in\sK_0$.
Then we may also consider
the $C^*$-algebra $\fW:=\ol{\fW^\circ}\subset\fA$, where
$\fW^\circ$ is the $*$-algebra generated by all
$\fW_j$, $j\ge-1$.

There is a natural isometric embedding $\pi:\fA\to\cB(\wh{\HR})$
such that every $\pi(A)$, $A\in\fA$, is reduced by $\HR^\ren_{\V{P}}$
\cite{vNeumann1939};
see Section~\ref{sec-CIRR} for more details.
Denoting the restriction of $\pi(A)$ to $\HR^\ren_{\V{P}}$
by $\pi_{\V{P}}(A)$ we obtain a representation
$$
\pi_{\V{P}}:\,\fA\longrightarrow\cB(\HR^\ren_{\V{P}})\,.
$$
To state our main results on these objects 
we need yet some more notation: $\id_M(T)$ is the
spectral projector associated with some self-adjoint operator, $T$,
and some Borel set, $M\subset\RR$; $\tr$ denotes the trace.

\begin{thm}\label{thm-intro2}
For all $\kappa,\pmax>0$, we find $\const,\ee_0>0$ such that the
following holds, for all $\V{P}\in\cB_\pmax$ and $\ee\in(0,\ee_0]$:

\smallskip

\noindent(1)
$E(\V{P})$ is a two-fold degenerate eigenvalue of  
$\wt{\Ham}_\infty(\V{P})$ and ${\Ham}^\ren({\V{P}})$, and
$$
\big\|\id_{\{E(\V{P})\}}(\wt{\Ham}_\infty(\V{P}))
-\id_{\{E_j(\V{P})\}}(\wt{\Ham}_j^\infty(\V{P}))\big\|
\le\const\,\ee\,(1+\const\,\ee)^j\rho_j
\xrightarrow{\;j\to\infty\;}0\,.
$$
(2) There is a unique linear functional,
$\omega_{\V{P}}\in\fA^*$, such that
$$
\omega_{\V{P}}(A)=\lim_{j\to\infty}
\tr\big[\id_{\{E_j(\V{P})\}}({\Ham}_j^\infty(\V{P}))\,A\big]\,,
\quad A\in\fA^\circ.
$$
(In particular, the limit exists.)
We have
$$
\omega_{\V{P}}(A)
=\tr\big[\id_{\{E(\V{P})\}}(\Ham^\ren(\V{P}))\,\pi_{\V{P}}(A)\big]\,,
\quad A\in\fA\,.
$$
(3) If $\V{Q}\in\cB_\pmax\setminus\{\V{P}\}$, then the
representations $\pi_{\V{P}}$ and $\pi_{\V{Q}}$ 
of $\fW$ are disjoint and, in particular,
$\omega_{\V{P}}$ is not a $\pi_{\V{Q}}$-normal functional
on $\fW$.
\end{thm}

\begin{proof}
(1) follows from \eqref{def-Ham-ren}, Theorem~\ref{mainthm}, 
and the explanation below \eqref{flann1}; see also \eqref{block}.
Part~(2) follows from the discussion in Section~\ref{sec-CIRR}.
With Theorem~\ref{thm-intro1} and Part~(2) at hand (3) follows from
arguments in \cite{Froehlich1973,KMcKW1966,Streit1967};
for the reader's convenience we sketch its proof
in Corollary~\ref{cor-disjoint}. The crucial point 
is that $\nabla E(\V{P})\not=\nabla E(\V{Q})$,
if $\V{P}\not=\V{Q}$, by the strict convexity of $E$ on $\cB_\pmax$.
Recall that disjointness of $\pi_{\V{P}}$ and $\pi_{\V{Q}}$ 
means that no $\pi_{\V{P}}$-normal positive functional on $\fW$
is $\pi_{\V{Q}}$-normal
and vice versa; by (2) $\omega_{\V{P}}$ is $\pi_{\V{P}}$-normal. 
\end{proof}

\begin{rem}
The Hilbert space $\CC^2\otimes\Fock{}{}$ is identified in a canonical way
with $\HR^\ren_{\V{0}}$. Choosing $\V{P}\not=\V{0}=\V{Q}$ in
Part~(3) then shows that $\omega_{\V{P}}$ cannot be represented
by a positive trace class operator acting on $\CC^2\otimes\Fock{}{}$.

In the NR case results similar to (2) and (3) can be found in
\cite{ChenFroehlich2007}. Compared to \cite{ChenFroehlich2007,Froehlich1973}
there are some simplifications in our proofs due to the fact that Pizzo's
recursive scheme establishes directly the existence of ground state projections
of $\wt{\Ham}_\infty(\V{P})$.
In \cite{ChenFroehlich2007,Froehlich1973} the analog of $\omega_{\V{P}}$
is defined via compactness arguments and analyzed 
by means of abstract results from the representation theory
of the CCR algebra. Statements similar to (1) for the NR model
have appeared in \cite{CFP2009,Pizzo2003}.
\hfill$\Diamond$ 
\end{rem}

\smallskip

\noindent{\bf Convention.}
The symbols $\const,\const',\ldots\;$
denote positive constants sometimes depending on $\kappa$
and upper bounds on $\ee$ or $|\V{P}|$, but on no other relevant parameters.
Their
values might change from one estimate to another.
$\dom(T)$ and $\fdom(T)$ denote the domain and form domain,
respectively, of a suitable operator $T$.


\section{Relative bounds and self-adjointness}\label{sec-rb}

\noindent
In Subsection~\ref{ssec-rb} below we derive some basic relative bounds
on various modified versions of the SR Pauli-Fierz fiber Hamiltonian
and study their self-adjointness and domains. It shall be very convenient
in the whole paper to work on a Hilbert space with {\em four}
spin degrees of freedom instead of only two.
For we may then employ fiber Dirac operators in the
computations which are linear in the vector potential
and whose absolute values can be represented by very handy formulas.
Similar ideas have been used earlier in 
\cite{KMS2009a,MatteStockmeyer2009a,MiyaoSpohn2009}.
The Dirac operators and some further necessary notation are introduced in the
following Subsection~\ref{ssec-4spinors}. 


\subsection{Operators on four-spinors}\label{ssec-4spinors}

\noindent
Recall the definition of the annuli $\cA_k^j$ and $\cA_j$
in \eqref{def-Akj} and of the corresponding Fock spaces $\Fock{k}{j}$
and $\Fock{j}{}$ explained below \eqref{def-Fockkj}; see also
Appendix~\ref{app-Fock} for their precise definition
as well as for the notation $d\Gamma$ used below.
We first fix the notation for the energy and momentum
of the photon field:
Writing $\omega(\V{k},\lambda):=|\V{k}|$ and 
$\V{m}(\V{k},\lambda)=\V{k}$, 
for $\V{k}\in\RR^3$, $\lambda\in\ZZ_2$, we define
\begin{align*}
\Hf^{(k,j)}&:=d\Gamma(\omega\!\!\upharpoonright_{\cA_k^j\times\ZZ_2})\,,
\qquad
\Hf^{(j)}:=d\Gamma(\omega\!\!\upharpoonright_{\cA_j\times\ZZ_2})\,,
\\
\Pf^{(k,j)}&:=d\Gamma(\V{m}\!\!\upharpoonright_{\cA_k^j\times\ZZ_2})\,,
\qquad
\,\Pf^{(j)}:=d\Gamma(\V{m}\!\!\upharpoonright_{\cA_j\times\ZZ_2})\,.
\end{align*}
At some points in our constructions there appear
unitarily equivalent sequences of IR cutoff
fiber Hamiltonians, whose norm resolvent limits are {\em not}
unitarily equivalent anymore. In the preparatory
Sections~\ref{sec-rb}--\ref{sec-gap} we wish to
treat these sequences in a unified fashion, whence we introduce
the following hypotheses on field operators:

\begin{assumption}\label{ariel}
For every $j\in\NN_0$, let
$(f_{j},\V{g}_{j})\in L^2(\cA_j^{j+1}\times\ZZ_2,\RR\times\RR^{3})$,
$(b_j,\V{c}_j)\in\RR\times\RR^3$,
and assume there exists $\const>0$ such that 
\begin{equation}\label{ariel2001}
\|\omega^{-\nf{1}{2}}(f_{j},\V{g}_{j})\|\le\const\,
\rho_j^{\nf{1}{2}},\quad
|(b_j,\V{c}_j)|\le\const\,\rho_j\,,\quad j\in\NN_0\,.
\end{equation}
For integers $0\le k<j\le\infty$ and $\ee>0$, set
\begin{align}\label{ariel2012}
\Phi_k^j&:=\!\sum_{k\le\ell<j}\!
\big(\vp(f_{\ell})+\ee\,b_\ell\big)\,,\quad
\V{A}_k^j:=\!\sum_{k\le\ell<j}\!
\big(\vp(\V{g}_{\ell})-\ee\,\V{c}_\ell\big)\,,
\end{align}
and abbreviate $\Phi_j:=\Phi_0^j$ and $\V{A}_j:=\V{A}_0^j$.
Finally, assume (for convenience) that
$\Hf^{(k,j)}\!+\ee\,\Phi_k^j\ge0$.
\hfill$\Diamond$
\end{assumption}

\smallskip

\noindent
The sums $\sum_{k\le\ell<j}$ in \eqref{ariel2012}, which
are infinite when $j=\infty$, are well-defined because
of \eqref{ariel2001} and define self-adjoint operators
whose domain contains the one of $(\Hf^{(j)})^{\nf{1}{2}}$;
see, e.g., \cite{ReedSimonII}.
Without further notice we shall often regard 
$\V{A}_k$ and $\V{A}_k^j$ as operators
in $\sF_j$ by identifying
$\V{A}_k\equiv\V{A}_k\otimes\id_{\sF_k^j}$
and $\V{A}_k^j\equiv\id_{\sF_k}\otimes\V{A}_k^j$.
The same convention is used for $\Phi_k^j$,
$\Hf^{(k)}$, $\Hf^{(k,j)}$, etc.

Next, we introduce the Hilbert spaces
\begin{equation}\label{def-HR-sC}
\HR_j:=\CC^4\otimes\sF_j\,,\quad
j\in\NN_0\cup\{\infty\}\,,\qquad \HR:=\HR_\infty\,.
\end{equation}
As usual we shall consider operators $L$ and $M$ acting in
$\sF_j$ and $\CC^4$, respectively,
also as operators in $\HR_j$ by identifying
$L\equiv\id_{\CC^4}\otimes L$ and $M\equiv M\otimes\id_{\sF_j}$. 
For instance this convention applies to the four Dirac matrices,
$$
\alpha_0:=\begin{pmatrix}\id&0\\0&-\id\end{pmatrix}
\in\LO(\CC^4)\,,
\qquad
\alpha_j:=\begin{pmatrix}0&\sigma_j\\\sigma_j&0\end{pmatrix}\in\LO(\CC^4)\,,
\;\;j=1,2,3\,.
$$
satisfying $\alpha_i=\alpha_i^*=\alpha_i^{-1}$
and the Clifford algebra relations
\begin{equation}\label{Clifford}
\{\alpha_i\,,\,\alpha_j\}=2\,\delta_{ij}\,\id\,,\quad i,j\in\{0,1,2,3\}\,.
\end{equation}
To show degeneracy of ground states 
we shall exploit the fact that the first and third
Pauli matrices are real while the second one is imaginary,
$$
\sigma_1=\begin{pmatrix}0&1\\1&0\end{pmatrix}\,,\quad
\sigma_2=\begin{pmatrix}0&-i\\i&0\end{pmatrix}\,,\quad
\sigma_3=\begin{pmatrix}1&0\\0&-1\end{pmatrix}\,.
$$
As a consequence of \eqref{Clifford} and the $C^*$-equality we obtain
\begin{equation}\label{C*}
\|\valpha\cdot\V{v}\|_{\LO(\CC^4)}=|\V{v}|\,,\quad\V{v}\in\RR^3,
\end{equation}
where 
$\valpha\cdot\V{v}:=\alpha_1v_1+\alpha_2v_2+\alpha_3v_3$.
It is a standard exercise using the CCR, \eqref{ariel2001}, and \eqref{C*}
to derive the following fundamental relative bounds,
\begin{equation}\label{akj-neu}
\big\|\Phi_k^j\,(\Hf^{(k,j)}\!+\rho_k)^{-\nf{1}{2}}\big\|
\le\const\,\rho_k^{\nf{1}{2}},
\quad
\big\|\valpha\cdot\V{A}_k^j\,(\Hf^{(k,j)}\!+\rho_k)^{-\nf{1}{2}}\big\|
\le\const\,\rho_k^{\nf{1}{2}},
\end{equation}
for all $k,j\in\NN_0\cup\{\infty\}$, $k<j$, and $\ee\in(0,\ee_0]$,
with $\const\equiv\const(\ee_0)$.

Let $\V{P}\in\RR^3$ in what follows.
We define the fiber Dirac operator,
$\D_k^j(\V{P})$, to be the {\em closure} of the
symmetric operator in $\HR_j$ given by
\begin{equation}\label{def-Dkj}
\dom(\Hf^{(j)})\ni\phi\,\longmapsto\,
\valpha\cdot(\V{P}-\Pf^{(j)}+\ee\,\V{A}_k)\,\phi+\alpha_0\,\phi
\end{equation}
and put $\D_j(\V{P}):=\D_j^j(\V{P})$.
Let 
$\sC_j\subset\HR_j$
denote the dense subspace of all 
$(\psi^{(n)})_{n=0}^\infty\in\HR_j=\bigoplus_{n=0}^\infty\CC^4\otimes[\sF_j]^{(n)}$
such that only finitely many $\psi^{(n)}$ are non-zero and each
$\psi^{(n)}$, $n\in\NN$, has a compact support.
According to Lemma~\ref{le-esa-wD} below
$\D_k^j(\V{P})$ is self-adjoint and $\sC_j$ is a core for $\D_k^j(\V{P})$.
Actually, it has already been remarked in \cite{MiyaoSpohn2009}
that the essential self-adjointness of $\D_k^j(\V{P})$ on 
$\sC_j$ follows from Nelson's commutator theorem with 
test operator $\Hf^{(j)}+1$. We give an alternative argument
in Lemma~\ref{le-esa-wD} which also shows that the
square, $\D_k^j(\V{P})^2$, is essentially self-adjoint on $\sC_j$.
A direct computation shows, however, that
\begin{equation}
\D_k^j(\V{P})^2=\cT_k^j(\V{P})\oplus\cT_k^j(\V{P})\quad
\textrm{on}\;\sC_j\,,
\end{equation}
where the direct sum refers to the splitting of the
spinor components $\CC^4=\CC^2\oplus\CC^2$, and where
$$
\cT_k^j(\V{P}):=\big(\vsigma\cdot(\V{P}-\Pf^{(j)}+\ee\,\V{A}_k)\big)^2+1\,.
$$
In particular,
\begin{equation}\label{SpecDirac}
\sigma(\D_k^j(\V{P}))\subset (-\infty,-1]\cup[1,\infty).
\end{equation}
Thus, using the essential self-adjointness of $\D_k^j(\V{P})^2$
it is easy to see that two 
possibilities to make sense out of the square root in \eqref{def-Ham(P)}
yield the same operator:
We may take the square root of the closure of $\cT_k^j(\V{P})$
defined by means of the spectral theorem or, equivalently, 
we may define the square root in \eqref{def-Ham(P)}
to be equal to the upper left (or lower right) 2\texttimes2 
block of the block-diagonal operator $|\D_k^j(\V{P})|$.
For technical reasons we find it convenient to work with
absolute values of the Dirac operator.

We finally define fiber Hamiltonians on four spinors by
\begin{equation}\label{def-HAMkj}
\HAM_k^j(\V{P}):=|\D_k^j(\V{P})|+\Hf^{(j)}+\ee\,\Phi_k\,,\quad
\HAM_j(\V{P}):=\HAM_j^j(\V{P})\,,
\end{equation}
on the domain $\dom(\HAM_k^j(\V{P})):=\dom(\Hf^{(j)})$.
As it turns out in the next subsection they are
self-adjoint on that domain.
According to the remarks above they are block diagonal,
\begin{equation}\label{block}
\HAM_k^j(\V{P})=\Ham_k^j(\V{P})\oplus\Ham_k^j(\V{P})\,,
\qquad\Ham_k^j(\V{P}):=\ol{\cT}_k^j(\V{P})^{\nf{1}{2}}+\Hf^{(j)}.
\end{equation}
Notice that $|\D_k^j(\V{P})|\le\HAM_k^j(\V{P})$ by the
last condition imposed in Assumption~\ref{ariel}.
Notice also that the photon momentum and energy operators
$\Pf^{(j)}$ and $\Hf^{(j)}$ contained in $\HAM_k^j(\V{P})$
act in the Fock space $\Fock{j}{}$ attached to the scale $j$
while the field operators $\V{A}_k$ and $\Phi_k$ in 
$\HAM_k^j(\V{P})$ act non-trivially only in 
$\Fock{k}{}$.

\smallskip

\noindent
{\bf Notational warning.} From Section~\ref{sec-DT} forth,
the symbols $\D_k^j(\V{P})$ and $\HAM_k^j(\V{P})$
will solely denote the operators obtained by choosing
$(f_j,\V{g}_j)=(0,\id_{\cA_j^{j+1}}\,\V{G})$ and $(b_j,\V{c}_j)=(0,\V{0})$. 
Then $\Ham(\V{P}):=\Ham_\infty^\infty(\V{P})$ is the physical operator
in \eqref{def-Ham(P)}.

\smallskip

\noindent
{\bf Convention.} Up to now we have always
kept the symbol $\V{P}$ in the notation.
Henceforth, we shall include it in the notation
only when new mathematical objects are defined
(to clarify their dependence on $\V{P}$).
Afterwards we then tacitly drop the
explicit reference to $\V{P}$ in most formulas
to reduce clutter.


\subsection{Basic relative bounds}\label{ssec-rb}

\noindent
We denote the resolvent 
of the fiber Dirac operator at $iy\in i\RR$ by
\begin{equation}\label{def-Rkj}
\R_k^j(iy)\equiv\R_k^j(\V{P},iy):=(\D_k^j(\V{P})-iy)^{-1},\qquad
\R_j(iy):=\R_j^j(iy)\,.
\end{equation}
As $0$ belongs to the resolvent set of $\D_k^j$ its sign function
can be represented as a strongly convergent principal value
(see, e.g., \cite[Page 359]{Kato}),
\begin{align}\label{for-sgnT}
\Sgn_k^j\,\psi\equiv
\Sgn_k^j(\V{P})\,\psi:=\D_k^j(\V{P})\,|\D_k^j(\V{P})|^{-1}\psi
=\lim_{\tau\to\infty}\int_{-\tau}^\tau
\R_k^j(\V{P},iy)\,\psi\,\frac{dy}{\pi}\,,
\end{align}
for all $\psi\in\HR_j$.
This yields a convenient representation of 
its absolute value,
\begin{align}\label{for-absvT}
|\D_k^j|\,\psi=\Sgn_k^j\,\D_k^j\,\psi=\lim_{\tau\to\infty}\int_{-\tau}^\tau
\big(\id+iy\,\R_k^j(iy)\big)\psi\,\frac{dy}{\pi}\,,
\quad \psi\in\dom(\D_k^j)\,.
\end{align}
The previous formula is very useful in combination with the bounds
\begin{align}\nonumber
\big|\SPb{&\phi}{(\R_j(iy)-\R_k^j(iy))\,
(\Hf^{(k,j)}\!+\rho_k)^{-\nf{1}{2}}\psi}\big|
\\\label{roswita}
&\le\const\,\ee\,\rho_k^{\nf{1}{2}}\|\R_j(-iy)\,\phi\|\,\|\R_k^j(iy)\,\psi\|
\le\const\,\ee\,\rho_k^{\nf{1}{2}}
(1+y^2)^{-1},
\end{align}
for all normalized $\phi,\psi\in\HR_j$; see 
Lemma~\ref{le-dominique} of the appendix.

Part~(i) of the next lemma improves on 
\cite[Proposition~1.2]{MiyaoSpohn2009}
where the same assertions have been derived,
for sufficiently small values of $\ee>0$.

\begin{lem}\label{le-lisa}
Let Assumption~\ref{ariel} be fulfilled. Then, for all
$\V{P}\in\RR^3$, $k,j\in\NN_0\cup\{\infty\}$, $k<j$,
$\ee,\delta>0$, and some constant $\const>0$
proportional to the one in \eqref{akj-neu}, the following holds:

\smallskip

\noindent
(i) $\HAM_j$ and $\HAM_k^j$ are self-adjoint on $\dom(\Hf^{(j)})$
and essentially self-adjoint on $\sC_j$. 

\smallskip

\noindent
(ii) For all $\psi\in\dom(\Hf^{(j)})$,
\begin{align}\nonumber
\|(\HAM_j-\HAM_k^j)\,\psi\|
&\le\const\,\ee\,\rho_k^{\nf{1}{2}}\,
\big\||\D_k^j|^{\nf{1}{4}}(\Hf^{(k,j)}\!+\rho_k)^{\nf{1}{2}}\,\psi\big\|
\\
&\le\delta\,\ee\,\rho_k^{\nf{1}{2}}\,\|\HAM_k^j\,\psi\|
+\const^4\,\ee\,\rho_k^{\nf{1}{2}}\,\|\psi\|/\delta^3,\label{susi}
\\\label{susiA}
\||\D_{\ell}^j|\,\psi\|,\,\|\Hf^{(j)}\,\psi\|&\le 
\const'(\ee_0)\,\|\HAM_k^j\,\psi\|\,,\qquad 0\le\ell\le j\,,\;\ee\le\ee_0\,.
\end{align}
\noindent(iii) $\HAM_k^\infty\to\HAM_\infty$ 
in the norm resolvent sense,
as $k\to\infty$.

\smallskip

\noindent(iv) 
The form domain of both $\HAM_j$ and $\HAM_k^j$
is $\fdom(\Hf^{(j)})$ and
the following form bounds hold true on $\fdom(\Hf^{(j)})$,
\begin{align}
\pm(\HAM_j-\HAM_k^j)
&\le\delta\,\ee\,\rho_k\,|\D_j|+
\const^2\,\ee\,
(\Hf^{(k,j)}\!+\rho_k)/\delta\,,\label{lisa1a}
\\
\pm(\HAM_j-\HAM_k^j)
&\le\delta\,\ee\,\rho_k\,|\D_k^j|+
\const^2\,\ee\,(1+(\delta\,\ee\,\rho_k)^2)\,
(\Hf^{(k,j)}\!+\rho_k)/\delta\,.\label{lisa1b}
\end{align}
\end{lem}

\begin{proof}
First, we derive the relative bounds of (ii) and (iv)
on the dense domain $\sC_j$:
Since $(\Phi_j,\V{A}_j)=(\Phi_k,\V{A}_k)+(\Phi_k^j,\V{A}_k^j)$
we may write, for $\phi,\psi\in\sC_j$, 
\begin{align*}
\SPn{\phi}{(\HAM_j-\HAM_k^j)\,\psi}
&=\ee\SPn{\Sgn_j\,\phi}{{\valpha}\cdot\V{A}_k^j\,\psi}
+\SPn{\phi}{(\Sgn_j-\Sgn_k^j)\,\D_k^j\,\psi}
+\ee\SPn{\phi}{\Phi_k^j\,\psi}
\\
&=\ee\SPn{{\valpha}\cdot\V{A}_k^j\,\phi}{\Sgn_k^j\,\psi}
+\SPn{\D_j\,\phi}{(\Sgn_j-\Sgn_k^j)\,\psi}
+\ee\SPn{\Phi_k^j\,\phi}{\psi}\,.
\end{align*}
In order to treat the difference of the sign functions 
let $(r,s)$ be either $(1,0)$ or $(0,1)$, and let
$\ve,\vk\in(0,1)$, $\ve+\vk=1$,
and $\phi,\psi\in\sC_j$.
Then a successive application of
\eqref{for-sgnT}, \eqref{roswita}, and 
$\|\,|\D_k^j|^{\nu}\R_k^j(\imath y)\|\le(1+y^2)^{-(1-\nu)/2}$, 
$\nu\in[0,1)$,
permits to get
\begin{align}\nonumber
\big|&\SPb{(\D_j)^r\,\phi}{(\Sgn_j-\Sgn_k^j)\,
(\Hf^{(k,j)}\!+\rho_k)^{-\nf{1}{2}}(\D_k^j)^{s}\psi}\big|
\\\nonumber
&=
\Big|\lim_{\tau\to\infty}\int_{-\tau}^\tau
\SPb{(\D_j)^r\,\phi}{(\R_j(iy)-\R_k^j(iy))\,
(\Hf^{(k,j)}\!+\rho_k)^{-\nf{1}{2}}(\D_k^j)^{s}\psi}\,\frac{dy}{\pi}\Big|
\\\nonumber
&\le\int_\RR\big\|\,|\D_j|^{r\vk}\R_j(-iy)\,|\D_j|^{r\ve}\phi\big\|
\,\const\,\ee\,\rho_k^{\nf{1}{2}}\,
\big\|\,|\D^j_k|^{s\vk}\R^j_k(iy)\,|\D_k^j|^{s\ve}\psi\big\|
\,\frac{dy}{\pi}
\\\label{sergio1} 
&\le\const(\ve)\,\ee\,\rho_k^{\nf{1}{2}}\,\big\|\,|\D_j|^{r\ve}\,\phi\big\|
\,\big\|\,|\D_k^j|^{s\ve}\,\psi\big\|\,.
\end{align}
Choosing $(r,s)=(0,1)$, $\ve:=1/4$,
and using \eqref{akj-neu}, $|\D^j_k|\ge1$, 
and the fact that $\D_k^j$ and $(\Hf^{(k,j)}+\rho_k)^{-\nf{1}{2}}$
commute on $\sC_j$,
we obtain 
\begin{align}\nonumber
\|(&\HAM_j-\HAM_k^j)\,\psi\|
=\sup_{\|\phi\|=1}\big|\SPb{\phi}{(\HAM_j-\HAM_k^j)\,\psi}\big|
\\\nonumber
&\le\ee\,\|\valpha\cdot\V{A}_k^j\,\psi\|
+\sup_{\|\phi\|=1}\big|\SPb{\phi}{(\Sgn_j-\Sgn_k^j)\,\D_k^j\,\psi}\big|
+\ee\,\|\Phi_k^j\,\psi\|
\\\nonumber
&\le\const\,\ee
\,\rho_k^{\nf{1}{2}}\,\big\|\,|\D^j_k|^{\nf{1}{4}}(\Hf^{(k,j)}\!+\rho_k)^{\nf{1}{2}}\psi\big\|
\\\nonumber
&\le\const\,\ee
\,\rho_k^{\nf{1}{2}}\,\big\|\,|\D^j_k|^{\nf{1}{2}}\psi\big\|^{\nf{1}{2}}
\,\big\|(\Hf^{(k,j)}\!+\rho_k)\,\psi\big\|^{\nf{1}{2}}
\\\label{sergio2}
&\le\const\,\ee
\,\rho_k^{\nf{1}{2}}\,\|(\HAM_k^j)^{\nf{1}{2}}\psi\|^{\nf{1}{2}}
\,\|\HAM_k^j\,\psi\|^{\nf{1}{2}}
\le\const\,\ee\,\rho_k^{\nf{1}{2}}\,\|\psi\|^{\nf{1}{4}}\,\|\HAM_k^j\,\psi\|^{\nf{3}{4}}.
\end{align} 
In the penultimate step we applied $|\D_k^j|\le\HAM_k^j$ 
to the left norm and used
that $|\D^j_k|+\Hf^{(k)}\!+\ee\,\Phi_k\ge0$
and $\Hf^{(k,j)}$ commute on $\sC_j$ and
$\Hf^{(k)}+\Hf^{(k,j)}=\Hf^{(j)}$ to bound the right one.
By Young's inequality
this implies~\eqref{susi}, for all $\psi\in\sC_j$. 
Choosing $(r,s)=(1,0)$, $\ve:=1/2$,
we also obtain \eqref{lisa1a},
which actually holds, for all $\delta>0$,
\begin{align*}
\big|\SPn{\phi}{(\HAM_j-\HAM_k^j)\,\phi}\big|
&\le\ee\|\valpha\cdot\V{A}_k^j\,\phi\|\,\|\phi\|
+\big|\SPn{\D_j\phi}{(\Sgn_j-\Sgn_k^j)\,\phi}\big|
+\ee\|\Phi_k^j\phi\|\,\|\phi\|
\\
&\le\const\,\ee\,\rho_k^{\nf{1}{2}}\,
\big\|\,|\D_j|^{\nf{1}{2}}\phi\big\|\,\|(\Hf^{(k,j)}\!+\rho_k)^{\nf{1}{2}}\,\phi\|\,,
\end{align*}
for every $\phi\in\sC_j$.
Setting $\Phi_k^j=0$ for the moment we see that
\eqref{lisa1a} with $1/\delta=2\ee\rho_k$ implies
$|\D_j|\le\const\,(|\D_k^j|+\Hf^{(k,j)})$,
from which we finally infer \eqref{lisa1b}
(for non-zero $\Phi_k^j$).

Let $0\le\ell\le j$ and $m:=\max\{k,\ell\}$. Choosing
$\V{A}_{m}^j=\V{0}$ and $\Phi_m^j=0$ (if $m<j$) for the moment we infer
$$
\|(\HAM_\ell^j-\HAM_k^j)\,\psi\|\le\ve\|\HAM_k^j\,\psi\|
+\const(\ve,\ee_0)\,\|\psi\|
\,,\;\;\ve>0\,,\quad
\|\HAM_\ell^j\,\psi\|\le\const'(\ee_0)\,\|\HAM_k^j\,\psi\|\,,
$$
from \eqref{susi}, which is already known to hold for $\psi\in\sC_j$.
Since $\HAM_0^j=((\V{P}-\Pf^{(j)})^2+\id)^{\nf{1}{2}}+\Hf^{(j)}$ is
obviously self-adjoint on $\dom(\Hf^{(j)})$ and essentially
self-adjoint on $\sC_j$ this proves Part~(i) by virtue
of the Kato-Rellich theorem. We also
conclude that \eqref{susi} extends to every $\psi\in\dom(\Hf^{(j)})$.
Since $\|\Hf^{(j)}\psi\|\le\|\HAM_0^j\,\psi\|$
and $\|\,|\D_\ell^j|\,\psi\|\le\const'(\|\HAM_\ell^j\,\psi\|+\|\Hf^{(j)}\psi\|)$
we further obtain \eqref{susiA}.

(iii): Choosing $j=\infty$ we readily infer from \eqref{susi} and the
second resolvent identity that
$\|(\HAM_k^\infty-i)^{-1}-(\HAM_\infty-i)^{-1}\|
\le\const\,\ee\,\rho_k^{\nf{1}{2}}\to0$,
as $k\to\infty$.
\end{proof}

\begin{rem}
We can use the arguments of the above proof to compare
Hamiltonians with the same scale parameters but for
different fibers. For instance, set $k=j$,
$\ee\,\V{A}_j^j=\V{h}$, and $\Phi_j^j=\Hf^{(j,j)}=0$.
In accordance with \eqref{akj-neu} we should then
replace $\ee\,\rho_k^{\nf{1}{2}}$ by $|\V{h}|$.
Making these changes
in \eqref{sergio1} and \eqref{sergio2} 
we immediately obtain
\begin{align}\label{sergio3}
\big\|(\HAM_k^j(\V{P}+\V{h})-\HAM_k^j(\V{P}))\psi\big\|
&\le\const|\V{h}|\,\big\||\D_k^j(\V{P})|^{\nf{1}{4}}\psi\big\|
\le\const|\V{h}|\,\|\HAM_k^j(\V{P})^{\nf{1}{4}}\psi\|,
\end{align}
for all $\V{P},\V{h}\in\RR$, $k,j\in\NN_0\cup\{\infty\}$,
$k\le j$, and $\psi\in\dom(\Hf^{(j)})$.
(To cover the cases $k<j$ simply choose suitable $\V{g}_k^j$.)
\hfill$\Diamond$
\end{rem}


\section{A priori bounds on the mass shell}\label{sec-P-dep}

\noindent
Let Assumption~\ref{ariel} be satisfied in the whole Section~\ref{sec-P-dep}.
From now on we use the following notation for ground state energies,
\begin{align*}
E_k^j(\V{P}):=\inf\Spec[\HAM_k^j(\V{P})]\,,\quad
E_j(\V{P}):=E_j^j(\V{P})\,,\quad E(\V{P}):=E_\infty(\V{P})\,.
\end{align*}
Moreover, we denote the spectral gap of $\HAM_k^{j}(\V{P})$ by
$$
\gap_k^j(\V{P}):=\inf\big\{\,\Spec\big[\HAM_k^j(\V{P})-E_k^j(\V{P})\big]
\setminus\{0\}\,\big\}\,,\qquad
\gap_j(\V{P}):=\gap_j^j(\V{P})\,.
$$
Here $\V{P}\in\RR^3$, $k,j\in\NN_0\cup\{\infty\}$, $k\le j$.
Notice that at this point we neither claim that $\gap_k^j$ be
strictly positive nor that $E_k^j$ be a simple eigenvalue. 
In fact, if $E_k^j$ is an
eigenvalue, then it will always be degenerate.
 
We now start to study the dependence of
$E_k^j$ on $\V{P}$. The main result of this section
is the lower Lipschitz type bound \eqref{eq3-neu2} 
obtained in Subsection~\ref{ssec-llb} below.
It serves as an
essential ingredient for the estimation of spectral gaps
in Pizzo's inductive scheme and,
in particular, similar bounds have been derived in the 
NR setting, e.g., 
in \cite{CFP2009,Pizzo2003}.
First, we collect, however,
some general results on $E_k^j$ and the Hamiltonians
in a series of remarks we shall repeatedly refer to in the
remaining part of the text.


\subsection{Derivatives of infra-red cut-off Hamiltonians
and their mass shells}

\begin{rem}\label{rem-rot-sym}
Let us choose
$(f_j,\V{g}_j)=(0,\V{G}_j^{j+1})$ and $(b_j,\V{c}_j)=(0,\V{0})$
in Assumption~\ref{ariel}, where
$\V{G}_j^{j+1}(\V{k},\lambda):=\id_{\cA_j^{j+1}}(\V{k})\,\V{G}(\V{k},\lambda)$,
$j\in\NN_0$, with
$\V{G}$ denoting the coupling function introduced in \eqref{def-AA}.  
In this case
it has already been noted in \cite{MiyaoSpohn2009}
that $E_k^j$ is rotationally symmetric,
i.e.
$E_k^j(\hat{U}\,\V{P})=E_k^j(\V{P})$, for all $\hat{U}\in\mathrm{SO}(3)$.
This follows from the fact that
the Hilbert space $\HR_j$ carries a unitary representation, $\pi$,
of $\mathrm{SU}(2)$ such that 
$\pi_U\,\HAM_k^j(\V{P})\,\pi_U^{-1}=\HAM_k^j(\hat{U}^{-1}\V{P})$,
where $\hat{U}\in\mathrm{SO}(3)$ corresponds to 
$U\in\mathrm{SU}(2)$ w.r.t. the universal
covering $\mathrm{SU}(2)\to\mathrm{SO}(3)$.
In particular, we may conclude that $\nabla E_k^j(\V{0})=\V{0}$, as soon as
we know that $E_k^j$ is differentiable at $\V{0}$.
\hfill$\Diamond$
\end{rem}

\begin{rem}\label{rem-typeA}
We now show that 
$\{\HAM_k^j(\V{P})\}_{\V{P}\in \RR^3}$ is an analytic family
of type A and record some frequently used formulas:
Representing the absolute value by means of \eqref{for-absvT}
and using a Neumann series expansion we deduce that
\begin{align}\nonumber
\HAM_k^j(\V{P}+\V{h})
&=\HAM_k^j(\V{P})+\underset{{\tau\to\infty}}{\textrm{s-lim}}\int_{-\tau}^\tau
\!\big(\R_k^j(\V{P}+\V{h},iy)-\R_k^j(\V{P},iy)\big)\frac{iy\,dy}{\pi}
\\\label{HAMComplex}
&=\sum_{\ell=0}^\infty\frac{1}{\ell!}\,\partial_{\V{h}}^\ell\HAM_k^j(\V{P})\,,
\quad\textrm{on}\;\dom(\Hf^{(j)})\,,\;|\V{h}|<1\,,
\end{align}
where we define (again dropping all $\V{P}$'s so that 
$\R_k^j(iy)\equiv\R_k^j(\V{P},iy)$)
\begin{align}\label{primus}
\partial_{\V{h}}\HAM_k^j\,\psi
&:=\lim_{\tau\to\infty}\int_{-\tau}^\tau\R_k^j(iy)
\,\valpha\cdot\V{h}\,\R_k^j(iy)\,\psi\,\frac{y\,dy}{i\pi}\,,
\quad\psi\in\dom(\Hf^{(j)})\,,
\\ \label{secundus}
\partial_{\V{h}}^\ell\HAM_k^j
&:= 
(-1)^{\ell+1}\ell!\int_\RR 
\R_k^j(iy)\,\{\valpha\cdot\V{h}\,\R_k^j(iy)\}^\ell
\,\frac{y\,dy}{i\pi}\in\cB(\HR_j)\,,\;\;\ell\ge2\,.
\end{align}
In fact, since $\|\valpha\cdot\V{h}\,\R_k^j(iy)\|\le|\V{h}|(1+y^2)^{-\nf{1}{2}}$
the integrals in \eqref{secundus} are absolutely convergent
and one easily verifies 
$\|\partial_{\V{h}}^\ell\HAM_k^j\|\le2\ell!|\V{h}|^\ell/\pi(\ell-1)$,
for $\ell\in\NN$, $\ell\ge2$. So, indeed, the part $\sum_{2}^\infty\!...$
of the series in \eqref{HAMComplex} converges in $\cB(\HR_j)$,
if $|\V{h}|<1$. It is then clear that the limit in \eqref{primus} exists.
Combining \eqref{sergio3} with \eqref{HAMComplex} we further
infer that the closure of the symmetric operator
$\partial_{\V{h}}\HAM_k^j$ -- 
henceforth again denoted by the same symbol --
is defined on a domain containing
the domain of $|\D_k^j|^{\nf{1}{4}}$ and
$\|\partial_{\V{h}}\HAM_k^j\,|\D_k^j|^{-\nf{1}{4}}\|
\le\const|\V{h}|+\bigO(|\V{h}|^2)$, thus
\begin{align}\label{sergio4}
\|\partial_{\V{h}}\HAM_k^j\,(\HAM_k^j)^{-\nf{1}{4}}\|\le
\|\partial_{\V{h}}\HAM_k^j\,|\D_k^j|^{-\nf{1}{4}}\|\le\const|\V{h}|\,.
\end{align}
The first order term in the series \eqref{HAMComplex} is
therefore an infinitesimal perturbation of the self-adjoint zeroth order term
$\HAM_k^j$ and, as all higher order terms are bounded, 
we conclude that the
series \eqref{HAMComplex} defines an extension of
$\{\HAM_k^j(\V{P})\}$ to an analytic family of type A
defined on $\{\V{P}\in\CC^3:\,|\Im\V{P}|<1\}$.
\hfill$\Diamond$
\end{rem}

\begin{rem}\label{rem-debora}
We note the following bound for later reference,
\begin{equation}\label{debora}
\big\|(\partial_{\V{h}}^\ell\HAM_j-\partial_{\V{h}}^\ell\HAM_k^j)(\Hf^{(k,j)}\!+\rho_k)^{-\nf{1}{2}}\psi\big\|
\le\const(\ell)\,\ee\,\rho_k^{\nf{1}{2}}\,|\V{h}|^\ell\,\|\psi\|\,,
\end{equation}
for $\psi\in\HR_j$, $\ell\in\NN$, $\ee\in(0,1]$.
It is easily derived as follows:
Represent the difference of the $\ell$-th derivatives by means
of \eqref{primus} or \eqref{secundus} and rearrange the integrants
in a telescopic sum of terms $\propto$
$(\R_j\,\valpha\cdot\V{h})^n
(\R_j-\R_k^j)\,(\valpha\cdot\V{h}\,\R_k^j)^{\ell-n}$, $n=0,\ldots,\ell$.
Then multiply the telescopic sum from the right with
$(\Hf^{(k,j)}\!+\rho_k)^{-\nf{1}{2}}$, observe that
$(\valpha\cdot\V{h}\,\R_k^j)^{\ell-n}$ and $(\Hf^{(k,j)}\!+\rho_k)^{-\nf{1}{2}}$
commute, and estimate all terms by $\|\R_a^b\|\le\Jy^{-\nf{1}{2}}$ and \eqref{roswita}.
Each summand in the telescopic sum is absolutely integrable.
(For $\ell=1$, start with $\psi\in\dom(\Hf^{(j)})$ and extend 
$(\partial_{\V{h}}\HAM_j-\partial_{\V{h}}\HAM_k^j)(\Hf^{(k,j)}\!+\rho_k)^{-\nf{1}{2}}$
to all of $\HR_j$ by continuity, preserving the same symbol for the extension.)
\hfill$\Diamond$
\end{rem}

\begin{rem}\label{rem-Hellmann-Feynman}
We next derive some Hellmann-Feynman type formulas for the derivatives of $E_k^j$.
To this end we assume that $\V{P}$ is contained in some
fixed open set, $\cV\subset\RR^3$, and that
$E_k^j$ is an isolated eigenvalue of constant, finite multiplicity
on $\cV$, i.e. it does not split when
$\V{P}$ varies in $\cV$. 
(For $j<\infty$, this assumption is verified in Theorem \ref{Thm:Gap} below 
where $\cV$ is a ball about the origin and
$\ee$ is assumed to be sufficiently small.)

Let $\Pi_k^j:=\id_{\{E_k^j\}}(\HAM_k^j)$ 
be the spectral projection onto the eigenspace corresponding to $E_k^j$.
By Remark~\ref{rem-typeA} and the above assumption
$E_k^j$ and $\Pi_k^j$ depend analytically on $\V{P}$ \cite{Kato}.
Then, by
differentiating  $\SPn{\HAM_k^j\, \phi}{\Pi_k^j\psi}=E_k^j\,\SPn{\phi}{\Pi_k^j\,\psi}$
we obtain the following Leibniz formula, 
for all $\phi\in \dom(\Hf^{(j)})$, $\psi\in \HR_j$,
$\mu\in\NN$, and $\V{h}\in \RR^3$,
\begin{equation}\label{michael}
\sum_{\nu=0}^\mu
{\mu \choose \nu}\SPn{\partial_{\V{h}}^\nu\HAM_k^j \,\phi}{\partial_{\V{h}}^{\mu-\nu}\Pi_k^j\,\psi}
=\sum_{\nu=0}^\mu{\mu \choose \nu} (\partial_{\V{h}}^\nu E_k^j)\, 
\SPn{ \phi}{\partial_{\V{h}}^{\mu-\nu}\Pi_k^j\,\psi}\,.
\end{equation}
In view of 
$\dom(\partial_{\V{h}}^\nu\HAM_k^j )\supset \dom(\Hf^{(j)})
=\dom(\HAM_k^j)\supset\Ran(\Pi_k^j)$ 
we infer from \eqref{michael} 
by induction on $\mu$ that 
$\partial_{\V{h}}^{\mu}\Pi_k^j\,:\HR_j\rightarrow\dom(\HAM_k^j)=\dom(\Hf^{(j)})$
and, in particular,
\begin{align}\label{eq1}
(&\partial_{\V{h}}\HAM_k^j-\partial_{\V{h}}E_k^j) \,\Pi_k^j
=-(\HAM_k^j-E_{k}^j)\,\partial_{\V{h}}\Pi_k^j\,,  
\\ \nonumber
(&\partial^2_{\V{h}}\HAM_k^j) \,\Pi_k^j
=-2(\partial_{\V{h}}\HAM_k^j -\partial_{\V{h}}E_k^j)\, \partial_{\V{h}}\Pi_k^j  
+ (\partial^2_{\V{h}}E_k^j)  \,\Pi_k^j
-(\HAM_k^j-E_{k}^j)\, \partial^2_{\V{h}}\Pi_k^j\,.
\end{align}
On the range of $(\Pi_k^j)^\perp=\id-\Pi_k^j$ the resolvent
$({\RES}_k^j)^\perp\in\cB((\Pi_k^j)^\perp\HR_j)$,
\begin{equation}\label{faysal}
({\RES}_k^j)^\perp
:=\big(\HAM_k^j\,(\Pi_k^j)^\perp-E_k^j\big)^{-1}(\Pi_k^j)^\perp,
\quad\|({\RES}_k^j)^\perp\|\le1\big/\gap_k^j\,,
\end{equation}
is well-defined and we deduce
from the first line in \eqref{eq1} that
\begin{align}\label{for-1Deriv}
(\Pi_k^j)^\perp\partial_{\V{h}}\Pi_k^j 
&= -({\RES}_k^j)^\perp (\partial_{\V{h}}\HAM_k^j)\Pi_k^j\,,\qquad 
(\partial_{\V{h}}E_k^j)\,\Pi_k^j=\Pi_k^j\,(\partial_{\V{h}}\HAM_k^j)\,\Pi_k^j\,.
\end{align}
Multiplying the second line in \eqref{eq1} by $\Pi_k^j$
and using \eqref{for-1Deriv} we further get
\begin{align}\label{for-2Deriv}
d\,\partial^2_{\V{h}}E_k^j
=&  \Tr\{\Pi_k^j(\partial^2_{\V{h}}\HAM_k^j) \,\Pi_k^j\}
-2\|(({\RES}_k^j)^\perp)^{\nf{1}{2}}\,(\partial_{\V{h}}\HAM_k^j)\, \Pi_k^j \|^2_{\HS}\,,
\end{align}
with $d:=\Tr\{\Pi_k^j\}$, $\Tr$ denoting the trace and $\|\cdot\|_{\HS}$
the Hilbert-Schmidt norm.
This formula will be used to prove the strict convexity of the mass shell.
\hfill$\Diamond$
\end{rem}


\subsection{Lower Lipschitz bound on the mass shell}
\label{ssec-llb}

\noindent
We are now heading towards a proof of the bound \eqref{eq3-neu2} below.
First, we consider  
$\HAM_0^j(\V{P})=((\V{P}-\Pf^{(j)})^2+1)^{\nf{1}{2}}+\Hf^{(j)}$.
Recall that
$\Pf^{(j)}$ and $\Hf^{(j)}$ act in the $n$-particle sectors
$[\sF_j]^{(n)}$, $n\in\NN$, 
simply by multiplication with $\V{k}_1+\dots+\V{k}_n$ and
$|\V{k}_1|+\dots+|\V{k}_n|$, respectively, and $\Pf^{(j)}\,\Omega_j=\V{0}$,
$\Hf^{(j)}\,\Omega_j=0$, where $\Omega_j$ is the vacuum vector 
in $\Fock{j}{}$, so that $[\Fock{j}{}]^{(0)}=\CC\,\Omega_j$.
In particular, it is an elementary exercise to derive the
following lemma. We present its proof only for the convenience of the reader.

\begin{lem}[{\bf Spectrum of $\boldsymbol{\HAM_0^j(\V{P})}$}]
\label{GapWithoutA}
For all $\V{P}\in\RR^3$ and $j\in\NN_0\cup\{\infty\}$, the following holds:

\smallskip

\noindent(i)
$E_0^j(\V{P})=(\V{P}^2+1)^{\nf{1}{2}}$ is a four-fold degenerate eigenvalue of
$H_0^j(\V{P})$. The corresponding eigenspace is 
$\{v\otimes\Omega_j:\,v\in\CC^4\}$.

\smallskip

\noindent(ii) 
$\gap_0^j(\V{P})\ge\min\{\rho_j,1\}/(2(|\V{P}|^2+1))$.
\end{lem}

\begin{proof}
(i): It is obvious that $(\V{P}^2+1)^{\nf{1}{2}}$ is an eigenvalue of 
$H_0^j(\V{P})$
and that $\CC^4\otimes\CC\,\Omega_j$ belongs to the corresponding eigenspace.
It follows from the arguments below that the latter is actually the whole
eigenspace (even in the case $j=\infty$)
and that $(\V{P}^2+1)^{\nf{1}{2}}$ is the 
infimum of the spectrum.

To prove (ii) we observe that
the spectrum of $\HAM_0^j(\V{P})$ restricted
to the invariant subspace $\{\CC^4\otimes[\sF_j]^{(0)}\}^\bot$ 
is given by the closure of
\begin{equation}\label{SpecJgleichNull}
\bigcup_{n=1}^\infty \,
\Big\{\Big( \Big[\V{P}-\sum_{i=1}^n\,\V{k}_i\Big]^2+1\Big)^{\nf{1}{2}}
+ \sum_{i=1}^n |\V{k}_i|\,:\,\V{k}_1,\ldots,\V{k}_n\in \mathcal{A}_j\Big\}\,.
\end{equation}
On account of
$\big|\sum_{i=1}^n \V{k}_i\big|
\le\sum_{i=1}^n|\V{k}_i|$,
for $\V{k}_1,\ldots,\V{k}_n\in \cA_j$,
this implies
\begin{align*}
\inf\Spec\big[H_0^j(\V{P})\!\!\upharpoonright
_{\{\CC^4\otimes[\sF_j]^{(0)}\}^\bot}\!\!\big]
&\ge
\inf_{\V{k}\in\RR^3}\big(\big([\V{P}-\V{k}]^2+1\big)^{\nf{1}{2}}+
\max\{\rho_j,|\V{k}|\}\big)
\\ \nonumber
&=
\inf_{r\ge 0}\big(\big([|\V{P}|-r]^2+1\big)^{\nf{1}{2}}
+\max\{\rho_j,r\}\big)\,.
\end{align*}
Writing $g(r):=([|\V{P}|-r]^2+1)^{\nf{1}{2}}
-(\V{P}^2+1)^{\nf{1}{2}}+\max\{\rho_j,r\}$, $r\in[0,\infty)$,
we find $g'>0$ on $[\rho_j,\infty)$,
thus $\min_{r\ge\rho_j} g(r)=g(\rho_j)$.
The possible minima of $g$ on $[0,\rho_j]$ are 
$g(\rho_j)$, $g(0)=\rho_j$, and $g(|\V{P}|)$, if $|\V{P}|\le\rho_j$.
Since $\RR\ni t\mapsto(t^2+1)^{\nf{1}{2}}$
is convex Taylor's formula yields, for all $\V{P}\in\RR^3$,
$$
g(\rho_j)\ge \rho_j-\frac{|\V{P}|}{(\V{P}^2+1)^{\nf{1}{2}}}\,\rho_j=
\frac{\rho_j}{(\V{P}^2+1)^{\nf{1}{2}}((\V{P}^2+1)^{\nf{1}{2}}+|\V{P}|)}
>\frac{\rho_j}{2(\V{P}^2+1)}\,.
$$
In the case $|\V{P}|\le\rho_j$ we apply the inequality
\begin{equation}\label{ineq-lea}
f(t):=1-(t^2+1)^{\nf{1}{2}}+t\ge\min\{t,1\}/2\,,\quad t\ge0\,,
\end{equation}
to get
$g(|\V{P}|)=1-(\V{P}^2+1)^{\nf{1}{2}}+\rho_j\ge\min\{\rho_j,1\}/2$.
(At least for $t\in[0,1]$, \eqref{ineq-lea} is clear by Taylor's formula;
for $t>1$, then, use $f(t)=tf(1/t)$.)
\end{proof}

\smallskip

\noindent
Note that the next lemma is valid even when $E_j$ is
in the essential spectrum of $\HAM_j$ and when no ground state exists.
In particular, it holds for $E=E_\infty$.

\begin{lem}\label{EnergyLemma}
(1) There exists $\const >0$ such that, for all
$\ee\in(0,1]$  and $\V{P}\in \RR^3$,
\begin{equation}\label{eq2-neu2}
(1-\const\,\ee)\,E_0(\V{P})\le E_j(\V{P})\le(1+\const\,\ee)\,E_0(\V{P})\,.
\end{equation}
(2)
There is some $\const >0$ such that,
for all $\ee\in(0,1]$,  $j\in \NN\cup\{\infty\}$,
normalized ${\V{u}}\in \RR^3$,  
and $\V{P}\in \RR^3$ for which
$\partial_{{\V{u}}}E_j(\V{P}):=\tfrac{d}{dt}E_j(\V{P}+ t\,{\V{u}})|_{t=0}$ 
exists,
\begin{equation}\label{eq1-neu2}
|\partial_{{\V{u}}} E_j(\V{P})|
\le 1-\frac{1}{2(1+\const\,\ee)\,E_0(\V{P})^2} +\const\,\ee\, E_0(\V{P})\,.
\end{equation}
(3) For all $\pmax>0$, we find $\ee_0>0$ and $\qmax\in(0,1)$ such that 
\begin{equation}\label{eq3-neu2}
E_j(\V{P}+\V{h})-E_j(\V{P})\ge -(1-\qmax)\,|\V{h}|,
\end{equation}
for all $\ee\in(0,\ee_0]$,
$\V{h}\in \RR^3$, $\V{P}\in\ol{\cB}_\pmax$, and $j\in\NN_0\cup\{\infty\}$.
\end{lem}
%
%
%
\begin{proof}
(1): The form bound \eqref{lisa1a} implies $\HAM_0^j\le(1+\const\,\ee)\,\HAM_j$
and \eqref{lisa1b} implies $\HAM_j\le(1+\const\,\ee)\,\HAM_0^j$.
Hence, Part~(1)
is a consequence of the variational
principle and $E_0^j=E_0$.

(2):
Let $\phi\in\dom(\Hf^{(j)})$ and $j\in\NN\cup\{\infty\}$. 
On account of Remark~\ref{rem-typeA},
\begin{align*}
\HAM_j(\V{P}+\V{h})\,\phi-\HAM_j(\V{P})\,\phi
=&\partial_{\V{h}}\HAM_j(\V{P}+\V{h})\,\phi+|\V{h}|^2\bigO(1)\,\phi\,,
\quad|\V{h}|\to0\,,
\end{align*}
where the norm of $\bigO(1)\in\cB(\HR_j)$ can be bounded
uniformly in $\V{P}\in\RR^3$ and in $j$.
Note that 
$\HAM_0^j(\V{P}+\V{h})=(\V{v}^2+1)^{\nf{1}{2}}+\Hf^{(j)}$ 
with $\V{v}:=\V{P}+\V{h}-\Pf^{(j)}$.
Since the components of $\V{v}$ and $\Pf^{(j)}$ commute strongly we obtain 
\begin{align*}
\SPn{\phi}{\partial_{\V{h}}\HAM_0^j(\V{P}+\V{h})\,\phi}
&=\SPn{\phi}{\V{h}\cdot\V{v}\,(\V{v}^2+1)^{-\nf{1}{2}}\,\phi}
\\
&\ge
-|\V{h}|\,\|\phi\|^2+ |\V{h}|\, \SPn{\phi}{(1-|\V{v}|\,(\V{v}^2+1)^{-\nf{1}{2}})\,\phi}
\\
&\ge-|\V{h}|\,\|\phi\|^2+|\V{h}|\, \SPn{\phi}{(\V{v}^2+1)^{-1}\,\phi}/2\,.
\end{align*}
By the trivial bound 
$(\V{v}^2+1)^{-1}\ge\HAM_0^j(\V{P}+\V{h})^{-2}$
(between multiplication operators) this yields
$\partial_{\V{h}}\HAM_0^j(\V{P}+\V{h})
\ge -|\V{h}|+|\V{h}|\,\HAM_0^j(\V{P}+\V{h})^{-2}/2$.
Next, 
\begin{equation*}
(\HAM_0^j)^{-2}-(\HAM_j)^{-2}=
 ((\HAM_0^j)^{-1}-\HAM_j^{-1})\,(\HAM_0^j)^{-1}+
\HAM_j^{-1}\, ((\HAM_0^j)^{-1}-\HAM_j^{-1})\,.
\end{equation*}
Using Lemma~\ref{le-lisa}(ii) and $\|(\HAM_0^j)^{-1}\|,\, \|\HAM_j^{-1}\| \le 1$ 
we find some $\const>0$ such that
\begin{align}
\|(\HAM_0^j)^{-2}-(\HAM_j)^{-2}\|\le2\,\|(\HAM_j-\HAM_0^j)(\HAM_0^j)^{-1}\|
\le 2\,\const\,\ee\,,
\end{align}
at every value of the total momentum.
We end up with
\begin{equation}
\partial_{\V{h}}\HAM_0^j(\V{P}+\V{h})\ge
- |\V{h}|+ (|\V{h}|/2)\,\HAM_j(\V{P}+\V{h})^{-2}-\const\,\ee\,|\V{h}|\,.
\end{equation}
By virtue of \eqref{debora},
$
\pm\{ \partial_{\V{h}}\HAM_j
-\partial_{\V{h}}\HAM_0^j\} \le \const\,\ee\,|\V{h}| \,
(\Hf^{(j)}+1)\le \const'\ee\,|\V{h}|\, \HAM_j
$, whence
\begin{align*}
\HAM_j&(\V{P}+\V{h})-E_j(\V{P})
\ge \HAM_j(\V{P}+\V{h})-\HAM_j(\V{P})
\\
&\ge - |\V{h}|+ (|\V{h}|/2)\,\HAM_j(\V{P}+\V{h})^{-2}-\const\,\ee\,|\V{h}|
- \const'\ee\,|\V{h}|\,\HAM_j(\V{P}+\V{h})-\bigO(|\V{h}|^2)\,.
\end{align*}
In particular, we obtain by means of the spectral calculus 
\begin{align}\label{eq2}
E_j(\V{P}+\V{h})-E_j(\V{P})
\ge&- |\V{h}|+ |\V{h}|\,E_j(\V{P}+\V{h})^{-2}/2-\const\,\ee\,|\V{h}|
\\ \nonumber
& - \const'\ee\,|\V{h}|\, E_j(\V{P}+\V{h})-\bigO(|\V{h}|^2)\,.
\end{align}
Now, let $\V{P}\in \RR^3$, let ${\V{u}} \in \RR^3$ be normalized, 
and assume $\partial_{{\V{u}}}E_j(\V{P})$ exists.
Replacing $\V{h}$ by $\pm\,s\,{\V{u}}$, for small $s>0$, in \eqref{eq2},
dividing by $\pm s$, and passing to the limit $s\searrow 0$, 
we obtain an estimate for $|\partial_{{\V{u}}}\, E_j(\V{P})|$
in terms of $E_j(\V{P})$.
Employing \eqref{eq2-neu2}, then, we arrive at \eqref{eq1-neu2}.

(3): In the following we borrow an argument from \cite{Froehlich1974},
namely the way to prove the concavity of $\Delta$. 

We fix $\V{P}\in\ol{\cB}_\pmax$ and some normalized 
$\V{u}\in\RR^3$, and define 
$$
e(t):=E_j(\V{P}+t\,{\V{u}})\,,\qquad 
h_\phi(t):=\SPn{\phi}{H_j(\V{P}+t\,{\V{u}})\,\phi}\,,
$$ 
for all normalized $\phi\in\dom(\Hf^{(j)})$ and all $t\in\RR$. 
Thanks to Remark~\ref{rem-typeA} we know that
$|h''_\phi(t)|=|\SPn{\phi}{\partial_{{\V{u}}}^2\HAM_j(\V{P}+t\,{\V{u}})\,\phi}|
\le \const_1$. Consequently,
the function defined by 
$\Delta_\phi(t):= -\const_1\,t^2+h_\phi(t)$, $t\in\RR$, 
is concave.
Thus, $\Delta(t):=\inf\{\Delta_\phi(t):\,\phi\in\dom(\Hf^{(j)}),\,\|\phi\|=1\}$, 
$t\in\RR$, 
defines a concave function, as well.
By a general theorem on concave functions, we know that the left derivative
$\Delta'_-(t)$ and right derivative $\Delta'_+(t)$ exist, and that 
$\Delta'_\pm$ are both 
decreasing, and coincide outside a countable set. $\Delta$ is differentiable
at every point $t\in\RR$ where $\Delta'_-(t)=\Delta'_+(t)$.
Moreover,
$\Delta(t)-\Delta(s)= \int_s^t \Delta'_\pm(r)\,dr$.
Since $e(t)=\Delta(t)+\const_1\,t^2$, the function $e$ also has left and right
derivatives $e_\pm'$ on $\RR$ which coincide almost everywhere, and
\begin{equation}\label{jinhuan}
e(t)-e(s)=\int_s^t e_\pm'(r)\,dr\,,\quad s,t\in\RR\,.
\end{equation}
Given $\V{h}\not=0$ we may choose $\V{u}:=\V{h}/|\V{h}|$ and,
inserting $t=|\V{h}|$ and $s:=0$ into \eqref{jinhuan}, we obtain
\begin{equation}\label{eq10}
E_j(\V{P}+\V{h})-E_j(\V{P})
=e(|\V{h}|)-e(0)
= \int_{0}^{|\V{h}|}e'(r)\,dr\,.
\end{equation}
Now, in order to prove \eqref{eq3-neu2} we may assume that
$E_j(\V{P}+\V{h})-E_j(\V{P})<-|\V{h}|/2$, which together with
\eqref{eq2-neu2}, $E_j\ge1$, and $|\V{P}|\le\pmax$ implies
$|\V{h}|\le2\,(1+\const\,\ee)\,E_0(\V{P})-2
\le2\,\pmax+1$, if $\ee>0$ is
sufficiently small, say $\ee\in(0,\ee_0]$. Thus,
$E_0(\V{Q})\le\const\,(\pmax+1)$,
for every $\V{Q}$
on the line segment from $\V{P}$ to $\V{P}+\V{h}$
we integrate over in \eqref{eq10} and for $\ee\in(0,\ee_0]$.
Thanks to \eqref{eq1-neu2} we may conclude that
$|\partial_{\V{u}}E_j(\V{Q})|
\le1+\const'\ee\,(\pmax+1)-1/4\const^2(\pmax+1)^2$,
for all $\V{Q}$ on the same segment, where
$\partial_{\V{u}}E_j(\V{Q})$ exists, 
and for small $\ee$.
Combining this with \eqref{eq10} and decreasing the
value of $\ee_0$, if necessary, we see that the assertion of (3)
holds true.
\end{proof}


\section{Spectral gaps}\label{sec-gap}

\noindent
In this section we derive lower bounds on
the gap above the first (degenerate) eigenvalues
of $\HAM_j(\V{P})$ and $\HAM_j^{j+1}(\V{P})$. 
This is done by induction on $j=0,1,2,\ldots\;$.
More precisely, as in \cite{Pizzo2003} we successively
estimate the gaps of 
$$
\HAM_0^{1}(\V{P}),\HAM_1(\V{P}),\HAM_1^{2}(\V{P}),
\ldots,\HAM_j(\V{P}),\HAM_j^{j+1}(\V{P}),\HAM_{j+1}(\V{P}),\ldots\;,
$$
assuming that $|\V{P}|\le\pmax$ and $\ee>0$ is sufficiently small.
The crucial step is taken in Lemma~\ref{le-gap1}
where the key ingredient \eqref{eq3-neu2} is applied
in order to estimate the gap of $\HAM_j^{j+1}(\V{P})$
in terms of the gap of $\HAM_j(\V{P})$.
After a technical intermediate step (Lemma~\ref{le-adam})
the final results of this section are stated in Theorem~\ref{Thm:Gap} below.
In the proof of Theorem~\ref{Thm:Gap} we perform the induction
and estimate the gap of $\HAM_{j+1}(\V{P})$
in terms of the gap of $\HAM_j^{j+1}(\V{P})$.
Since our model includes spin we have to give an additional
argument ensuring that the
ground state eigenvalue always stays four-fold degenerated.
This can, however, be done by an essentially well-known
\cite{MiyaoSpohn2009} application of Kramer's degeneracy theorem.
In \cite{MiyaoSpohn2009} the authors also show 
the existence of a spectral gap, 
for strictly positive photon masses and small $\ee>0$,
when a pre-factor $\gamma\in(0,1)$ is
introduced in front of the square-root in \eqref{def-Ham(P)}
(motivated by requirements of adiabatic
perturbation theory).
A bound on the spectral gap of 
IR cutoff fiber Hamiltonians with $\gamma=1$
has not yet been derived.

We denote the projection onto the vacuum sector
in $\sF_k^j$ by
$P_{\Omega_k^j}$.

\begin{lem}\label{le-gap1}
Let Assumption~\ref{ariel} be satisfied and
let $k,j\in\NN_0\cup\{\infty\}$, $k<j$.
Assume there exist
$\pmax,\ee_0>0$ such that,
for all $\V{P}\in\ol{\cB}_\pmax$ and $\ee\in(0,\ee_0]$,
$E_k(\V{P})$ is a 4-fold degenerate eigenvalue of $\HAM_k(\V{P})$.
Let $\Pi_k:=\id_{\{E_k\}}(\HAM_k)$
denote the spectral projection
onto the corresponding eigenspace.
Then the following holds,
at every $\V{P}\in\ol{\cB}_\pmax$ and for all $\ee\in(0,\ee_0]$:

\smallskip

\noindent(1) $E_k$ is also an
eigenvalue of $\HAM_k^j$, and the range of 
$\Pi_k\otimes P_{\Omega_k^j}$
belongs to the
corresponding eigenspace of $\HAM_k^j$.

\smallskip

\noindent 
(2) With $\qmax$ as in Lemma~\ref{EnergyLemma}(3) we have,
for all $\rho\ge 0$ and $\lambda>0$,
\begin{align}\label{gap1}
\gap_{k}^{j}&\ge
\min\big\{\gap_k,\,\qmax\,\rho_{j}\big\}\,,\qquad 
E_k^j=E_k\,,
\\\label{adam1}
\Hf^{(k,j)}+\rho&\le\max\{\qmax^{-1},\lambda^{-1}\}\,
(\HAM_k^j-E_k+\lambda\,\rho)\,.
\end{align}
(3)
If $j<\infty$ and $\gap_k>0$,  
then $\gap_{k}^{j}>0$ and 
$E_k$ is also a 4-fold degenerate eigenvalue of $\HAM_k^j$.
\end{lem}

\begin{proof} (1):
Recall that $\HAM_k^j=|\D_k^j|+\Hf^{(k)}+\ee\,\Phi_0^k+\Hf^{(k,j)}$,
where $\Hf^{(k,j)}\,\Omega_k^j=0$,
$|\D_k^j|\,(\id\otimes P_{\Omega_k^j})=|\D_k|\otimes P_{\Omega_k^j}$,
and $\HAM_k=|\D_k|+\Hf^{(k)}+\ee\Phi_0^k$.
Therefore,
$$
(\HAM_k^j-E_k)\,\Pi_k\otimes P_{\Omega_k^j}
=\{(\HAM_k-E_k)\,\Pi_k\}\otimes P_{\Omega_k^j}=0\,.
$$
(2):
The following direct integral representation holds,
for all $n\in\NN$,
\begin{align}\label{li2000}
\HAM_k^j&(\V{P})\!\!\upharpoonright_{\sH_k\otimes[\sF_k^j]^{(n)}}
\\\nonumber
&=\bigoplus_{\lambda_1,...,\lambda_n\in\ZZ_2}\int_{[\cA_k^j]^n}^\oplus
\Big\{\HAM_k\Big(\V{P}-\sum_{\ell=1}^n\V{k}_\ell\Big)
+\sum_{\ell=1}^n|\V{k}_\ell|\Big\}\,d^{3n}(\V{k}_1,\ldots,\V{k}_n)\,. 
\end{align}
By virtue of \eqref{eq3-neu2} and the triangle inequality we obtain
\begin{align*}
\HAM_k\Big(\V{P}-\sum_{\ell=1}^n\V{k}_\ell\Big)
\ge E_k\Big(\V{P}-\sum_{\ell=1}^n\V{k}_\ell\Big)
\ge E_k(\V{P})-(1-\qmax)\sum_{\ell=1}^n|\V{k}_\ell|\,,
\end{align*}
for all $\V{k}_1,\ldots,\V{k}_n\in\cA_k^j$.
Dropping again the argument $(\V{P})$, we further conclude
from the remarks in the first part of this proof that
\begin{align*}
(\HAM_k^j&-E_k)\,(\id\otimes P_{\Omega_k^j})
=(\HAM_k-E_k)\otimes P_{\Omega_k^j}
\ge \gap_k\,\Pi_k^\bot\otimes P_{\Omega_k^j}\,,
\end{align*}
where $\Pi_k^\bot:=\id-\Pi_k$.
Altogether we obtain
\begin{align}\nonumber
\HAM_k^j-E_k
&\ge\,\qmax\,\Hf^{(k,j)}+\gap_k\,\Pi_k^\bot\otimes P_{\Omega_k^j}
\\
&\ge{\qmax}\,\rho_{j}\,P_{\Omega_k^j}^\bot\label{susi2}
+\gap_k\,\Pi_k^\bot\otimes P_{\Omega_k^j}\ge0\,,
\end{align}
which yields \eqref{gap1}.
Combining the first estimate in \eqref{susi2}
with the obvious bound
$\Hf^{(k,j)}+\rho\le\max\{\qmax^{-1},\lambda^{-1}\}
\,(\qmax\,\Hf^{(k,j)}+\lambda\,\rho)$
we arrive at \eqref{adam1}.

(3): The assertion follows from Part~(1) and \eqref{susi2}. 
\end{proof}

\smallskip

\noindent
In what follows we shall use the notation
\begin{equation*}
\mu_5^{j}(\V{P}):=\!\!
\sup_{L\in\sG_4(\HR_{j})}\!\!\inf\big\{
\SPn{\phi}{(\HAM_{j}(\V{P})-E_{j}(\V{P}))\,\phi}:
\phi\in\fdom(\Hf^{(j)}),\|\phi\|=1,\phi\bot L
\big\},
\end{equation*}
where $\sG_4(\HR_{j})$ is the set of four-dimensional subspaces
of $\HR_{j}$. According to the minimax principle $\mu_5^{j}(\V{P})$
is the fifth eigenvalue of $\HAM_{j}(\V{P})-E_{j}(\V{P})$,
counting from below including multiplicities, 
or the lower bottom of the essential spectrum
of $\HAM_{j}(\V{P})-E_{j}(\V{P})$.
This notation shall be useful in a situation where we do not
already know whether $\mu_5^{j}(\V{P})$ be equal to $\gap_j(\V{P})$.

\begin{lem}\label{le-adam}
Let $k,j\in\NN_0\cup\{\infty\}$, $k<j$,
and let all assumptions of Lemma~\ref{le-gap1} be satisfied.
Then the following form bounds hold true on $\fdom(\Hf^{(j)})$,
at every $\V{P}\in\ol{\cB}_\pmax$ and for all $\ee\in(0,\ee_0]$ and $\lambda>0$,
with constants depending only on $\pmax$ and $\ee_0$,
\begin{align}\label{adam00}
|\D_j|&\le\max\{1,\const/\lambda\}
(\HAM_j-E_j+\lambda)\,,
\\\label{adam2}
\pm(\HAM_{j}-\HAM_k^{j})
&\le\const'\ee\,\max\{1,1/\lambda\}\,
(\HAM_k^{j}-E_k+\lambda\,\rho_k)\,,
\\\label{adam3}
\pm(\HAM_{j}-\HAM_k^{j})
&\le\const''\ee\,
(\HAM_{j}-E_k+\rho_k)\,.
\end{align}
Furthermore, 
\begin{align}\label{EnergyDiff}
|E_{j}-E_k|&\le \const\,\ee\,\rho_k\,,
\end{align}
and, if $j<\infty$ and $\gap_k>0$, then
\begin{align}\label{mu3}
\mu_5^{j}&\ge 
(1-\const\,\ee)\,\gap_k^{j}-\const\,\ee\,\rho_k\,.
\end{align}
\end{lem}

\begin{proof}
To start with we observe that \eqref{eq2-neu2} 
implies the upper bound,
$E_j(\V{P})\le
(1+\const\,\ee_0)\,\max_{|\V{Q}|\le\pmax}E_0(\V{Q})=:C$,
which is uniform in $j\in\NN\cup\{\infty\}$, 
$\ee\in(0,\ee_0]$, and $|\V{P}|\le\pmax$.
From now on we again drop the reference to $\V{P}$ in the notation.
Then
$|\D_j|\le\HAM_j-E_j+C$ 
(recall $\Hf^{(j)}+\ee\,\Phi_0^j\ge0$)
which implies \eqref{adam00}.
Combining \eqref{lisa1b} and  \eqref{adam1} we further obtain
\begin{equation}\label{adam7}
\pm(\HAM_{j}-\HAM_k^{j}) 
\le\const\,\ee\,\rho_k\,|\D_k^{j}|
+\const\,\ee\,(\HAM_k^{j}-E_k+\rho_k)
\end{equation}
on $\fdom(\Hf^{(j)})$.
Adding the term $\const\,\ee\,\rho_k\,(\Hf^{(j)}+\Phi_0^k-E_k+C)\ge0$
to the RHS of \eqref{adam7} 
and absorbing one factor $\rho_k\le\kappa$ in the constant
we arrive at \eqref{adam2}.
If, say, $\const'\ee\le1/2$, then 
\eqref{adam3} follows from \eqref{adam2} and some trivial
manipulations.
From \eqref{adam2} we further deduce that
$E_{j}-E_k\le 
\HAM_{j}-E_k\le (1+\const\,\ee)(\HAM_k^{j}-E_k)+\const\,\ee\,\rho_k$.
By Lemma~\ref{le-gap1} we have $E_k=\inf\Spec[\HAM_k^{j}]$
and, hence, the variational principle yields
$E_{j}-E_k\le\const\,\ee\,\rho_k$. 
In the same way we infer the bound $E_k-E_{j}\le\const'\,\ee\,\rho_k$
from \eqref{adam3}. Altogether this proves \eqref{EnergyDiff}.
Writing
$\cN_{j}:=\ker(\HAM_k-E_k)\otimes[\sF_k^{j}]^{(0)}$ and
$\cS_{j}:=\{\psi\in\fdom(\Hf^{(j)}):\,\|\psi\|=1\}$
and employing \eqref{adam2} and \eqref{EnergyDiff} 
we finally obtain
\begin{align*}
\mu_5^{j}&\ge\,\inf_{\phi\in\cN_{j}^\bot\cap\cS_{j}}\,
\SPn{\phi}{(\HAM_{j}-E_{j})\,\phi}
\\ \nonumber
&\ge\,
(1-\const\,\ee)
\inf_{\phi\in\cN_{j}^\bot\cap\cS_{j}}\,
\SPn{\phi}{(\HAM_k^{j}-E_k)\,\phi}+E_k-E_{j}
-\const\,\ee\,\rho_k
\\ \nonumber
&\ge\,(1-\const\,\ee)\,\gap_k^{j}- \const'\ee\,\rho_k\,.
\qedhere
\end{align*}
\end{proof}

\smallskip

\noindent
In the next theorem we combine the previous lemmata
in an induction argument to derive a bound
on the spectral gaps of $\HAM_j$ and $\HAM_j^{j+1}$, $j\in\NN$.
The argument based on Kramer's degeneracy theorem
which is used to show that the ground state
energies of both operators are four-fold degenerate eigenvalues
is essentially well-known; see \cite{MiyaoSpohn2009}.
Let
\begin{equation}\label{Def-theta}
\vartheta:=\begin{pmatrix}\sigma_2 & 0 \\ 0&-\sigma_2\end{pmatrix}\,\cC,\quad  
X_1:= \begin{pmatrix}\id & 0 \\ 0&0\end{pmatrix},\quad
X_2:= \begin{pmatrix}0 & 0 \\ \id&0\end{pmatrix},
\end{equation}
where 
$\sigma_2$ is the second Pauli matrix and $\cC$ denotes complex conjugation.

\begin{thm}[{\bf Spectral gaps of $\boldsymbol{\HAM_j}$
and $\boldsymbol{\HAM_j^{j+1}}$}]\label{Thm:Gap}
Let Assumption~\ref{ariel} be satisfied.
Then there exist $\ee_0>0$ and $q\in(0,\qmax)$ such that,
at every $\V{P}\in\ol{\cB}_\pmax $ and for all $\ee\in(0,\ee_0]$ and $j\in\NN_0$,
the following holds:

\smallskip

\noindent(1)
$E_j=\inf\Spec[\HAM_j]=\inf\Spec[\HAM_j^{j+1}]$ 
is a four-fold 
degenerate eigenvalue of both $\HAM_j$ and $\HAM_j^{j+1}$ and
\begin{equation}\label{gabi}
\gap_j\ge q\,\rho_j/2\quad
\textrm{and}\quad\gap_j^{j+1}\ge q\,\rho_j/2\,.
\end{equation}
\noindent(2)
The operators $\vartheta$, $X_1$, and $X_2$ commute with $\HAM_j$
and $\HAM_j^{j+1}$. If $\psi^{(1)}_j$ is a ground state
eigenvector of $\HAM_j$ 
satisfying $X_1\psi_j^{(1)}=\psi_j^{(1)}$ (which always exists), then four
mutually orthogonal ground state eigenvectors are given by $\psi_j^{(1)}$ and
\begin{equation}
\psi_j^{(2)}:= X_2\,\psi_j^{(1)},\quad \psi_j^{(3)}:= 
\vartheta\,\psi_j^{(1)},\quad \psi_j^{(4)}:= X_2\,\vartheta\,\psi_j^{(1)}.
\end{equation}
(3) Any vector $\psi_j\in\HR_j$ 
is a ground state eigenvector of $\HAM_j$,
if and only if 
$\psi_j\otimes\Omega_j^{j+1}$ 
is a ground state eigenvector of 
$\HAM_j^{j+1}$.
\end{thm}

\begin{proof}
(1)\&(3):
It suffices to prove the following assertions by induction on $j\in\NN_0$,
\begin{align*}
\sA_j\;:\Leftrightarrow\;
E_j\;\textrm{is a 4-fold eigenvalue of}\;
\HAM_j\;\textrm{and}\;\gap_j\ge q\,\rho_j/2\,.
\end{align*}
In fact, if
$j\in\NN_0$ and the assertion $\sA_j$ holds true, then Lemma~\ref{le-gap1}
implies $\gap_j^{j+1}\ge q\,\min\{\rho_j/2,\,\rho_{j+1}\}
= q\,\rho_j/2$ and all the remaining statements of (1) and (3).

Assertion $\sA_0$ follows, however, from Lemma~\ref{GapWithoutA},
if $|\V{P}|\le\pmax $, provided that $q$ is sufficiently
small.

Now, assume that $\sA_j$ holds true, for some $j\in\NN_0$. 
Then the bound \eqref{mu3} is available as well as
our conclusion $\gap_j ^{j+1}\ge{q}\,\rho_j/2$.
In combination they give
\begin{align*}
\mu_5^{j+1}
&\ge (1-\const\,\ee)\,q\,\rho_j/2
-\const\,\ee\,\rho_j
=q\,\rho_{j+1}\,(1-\const'\ee)
\ge q\,\rho_{j+1}/2,
\end{align*}
for small $\ee>0$.
In particular, there is some non-vanishing $\phi_{j+1}\in\dom(\HAM_{j+1})$ with 
$\HAM_{j+1}\,\phi_{j+1}=E_{j+1}\,\phi_{j+1}$.
Below we apply Kramer's degeneracy theorem
in order to show that $E_{j+1}$ is an at least four-fold degenerate
eigenvalue of $\HAM_{j+1}$. After that it also follows
that $\gap_{j+1}=\mu_5^{j+1}\ge q\rho_{j+1}/2$.

(2):
Since $\sigma_2=\sigma_2^*$ has
purely imaginary entries we have $\vartheta^2=-1$ and
\begin{equation}\label{eq:3}
\SPn{\vp}{\vartheta\,\psi}=-\SPn{\psi}{\vartheta\,\vp }\,,
\qquad \vp,\psi\in\sH_{j}\,.
\end{equation}
Obviously, $\vartheta$ leaves $\dom(\HAM_{j})=\dom(\Hf^{(j)})$
invariant and using $\{\sigma_k,\sigma_\ell\}=2\delta_{k\ell}\,\id_2$
and the fact that $\alpha_0$, $\alpha_1$, and $\alpha_3$ have real entries
it is straightforward to check that
$\vartheta\,\D_j=\D_j\,\vartheta$ on $\dom(\Hf^{(j)})$.
Since the fiber Dirac operator $\D_j$ is essentially
self-adjoint on $\dom(\Hf^{(j)})$ we deduce that
it commutes with $\vt$ on
$\vt\,\dom(\D_j)=\dom(\D_j)$ and
its resolvent satisfies $\vartheta\,\R_j(iy)=\R_j(-iy)\,\vartheta$
on $\HR_j$. Using, e.g., \eqref{for-absvT}
(and a substitution $y\to -y$)
we infer that $\vartheta\,|\D_j|=|\D_j|\,\vartheta$
on $\dom(\D_j)$. Altogether it follows that
$\vartheta\,\HAM_{j}=\HAM_{j}\,\vartheta$ on $\dom(\HAM_{j})$
and the same argument can be applied to $\HAM_j^{j+1}$.
In view of \eqref{block}
it is clear that $X_1$ and $X_2$ commute with the Hamiltonians.

So, let $\psi_j^{(1)}$ be as in the statement of Part~(2).
Then $\psi_j^{(3)}=\vartheta\,\psi_j^{(1)}\in\dom(\HAM_{j})$ and 
$\HAM_{j}\,\psi_j^{(3)}=E_{j}\,\psi_j^{(3)}$. 
Upon choosing $\vp:=\psi:=\psi_j^{(1)}$ in \eqref{eq:3}
we further see that $\psi_j^{(3)}\bot\psi_j^{(1)}$
and the assertion of Part~(2) becomes obvious.
\end{proof}


\section{The dressing transformation}\label{sec-DT}

\noindent
{\bf Notation for Sections \ref{sec-DT}--\ref{sec-CIRR} 
and Appendix~\ref{app-a(k)}.} 
From now on we reserve the symbols $D_k^j$, $H_k^j$,  and $R_k^j$
to denote the operators defined by 
\eqref{def-Dkj}, \eqref{def-HAMkj}, and \eqref{def-Rkj},
respectively, 
{\em for the special choices 
$(f_j,\V{g}_j)=(0,\V{G}_j^{j+1})$ and $(b_j,\V{c}_j)=(0,\V{0})$
in Assumption~\ref{ariel}}. Here
$\V{G}$ is the physical, ultra-violet cutoff
coupling function introduced in \eqref{def-AA} and 
$$
\V{G}_j(\V{k},\lambda):=\id_{\cA_j}(\V{k})\,\V{G}(\V{k},\lambda)\,,
\quad
\V{G}_j^{j+1}(\V{k},\lambda):=\id_{\cA_j^{j+1}}(\V{k})\,\V{G}(\V{k},\lambda)
\,,\quad j\in\NN_0\,.
$$
In fact, this choice is allowed in view of
\begin{equation}\label{Bounds0}
\|\omega^{\nu}\V{G}_j^{j+1}\|\le\const\,\rho_j^{1+\nu},
\quad\nu\in\{-\nf{1}{2},0,\nf{1}{2},1\}\,,\;j\in\NN_0\,.
\end{equation}
More explicitly, from now on we fix the notation
\begin{align*}
\D_k^j&:=\;\textrm{closure of}\;
\big\{\big(\valpha\cdot(\V{P}-\Pf^{(j)}+\ee\vp(\V{G}_k))
+\alpha_0\big)\!\!\upharpoonright_{\dom(\Hf^{(j)})}\big\}\,,
\\
\HAM_k^j&:=|\D_k^j|+\Hf^{(j)},\qquad\R_k^j(iy):=(\D_k^j-iy)^{-1}.
\end{align*}
The symbols $E_j$ and $\gap_j$ will always denote the ground state energy and the 
spectral gap, respectively, of $\HAM_j=\HAM_j^j$.

If we work with other choices of $(f_j,\V{g}_j)$ and $(b_j,\V{c}_j)$,
then we shall always indicate this by putting accents on top of the
symbols $D$, $H$, and $R$.

Pick some $\pmax>0$ and let $0<q<\qmax<1$ denote the
parameters appearing in Lemma~\ref{EnergyLemma}(3) and Theorem~\ref{Thm:Gap}.
Let also $j\in\NN_0$ 
and recall the definition \eqref{def-rhoj} of $\rho_j$.
We always assume that 
\begin{equation}\label{peter}
|\V{P}|\le\pmax\,.
\end{equation}
According to \eqref{EnergyDiff} and \eqref{peter}
we may choose $\ee>0$ so small that
\begin{equation}\label{diff-E}
|E_{j+1}-E_j|\le \const\,\ee\,\rho_j\le q\,\rho_j/8\,,
\end{equation}
and Theorem~\ref{Thm:Gap}(1) and \eqref{peter} imply
\begin{equation}\label{gap-j+1}
\gap_j\ge q\,\rho_j/2\,,\quad
\gap_{j+1}\ge q\,\rho_{j+1}/2= q\,\rho_j/4\,.
\end{equation}
We define spectral projections associated with the ground state energy,
\begin{align}\label{def-Pijj+1}
\Pi_j(\V{P})&:= \id_{\{E_j(\V{P})\}}(\HAM_j(\V{P}))\,,
\qquad
\Pi_j^{j+1}(\V{P}):= \id_{\{E_j(\V{P})\}}(\HAM_j^{j+1}(\V{P}))\,.
\end{align}
For later reference we note that Theorem~\ref{Thm:Gap}(3) 
and \eqref{peter} imply
\begin{equation}\label{for-Pijj+1}
\Pi_j^{j+1}(\V{P})=
\Pi_j(\V{P})\otimes P_{\Omega_j^{j+1}}
\,,\qquad P_{\Omega_j^{j+1}}:=
|\Omega_j^{j+1}\rangle\langle\Omega_j^{j+1}|\,,
\end{equation}
where ${\Omega_j^{j+1}}=(1,0,0,\ldots\:)$ is the vacuum vector in $\sF_j^{j+1}$.
Notice that $E_j$ is analytic on $\cB_\pmax$ 
by the discussion in Section~\ref{sec-P-dep}
and by Theorem~\ref{Thm:Gap}(1).

From now on $f_j$ will always denote the coherent factor
defined in \eqref{def-ff}.
By now we actually know that
$f_{j}$ is well-defined by \eqref{def-ff}, as \eqref{eq3-neu2} and \eqref{peter} imply
$|\nabla E_j|\le 1-\qmax<1$, for every small $\ee$.


\subsection{Definition and basic properties of dressing transformed operators}
\label{ssec-IAPT}

\noindent
Define $U_{j}$ by \eqref{def-U-intro},
i.e. $U_{j}(\V{P})=e^{-i\ee\vo(f_{j}(\V{P}))}$.
Obviously, $U_{j}$ acts non-trivially
only in the second tensor factor of $\HR_{j+1}=\HR_j\otimes\sF_j^{j+1}$. 
Standard arguments imply that $U_{j}$ maps $\dom([\Hf^{(j+1)}]^\nu)$ into 
itself, for every $\nu\ge1/2$; see Lemma~\ref{le-will} of the appendix.
Therefore, we may take strong
derivatives at $t=0$ of the following Weyl relations on
the domain of $[\Hf^{(j+1)}]^{\nf{1}{2}}$,
\begin{align*}
e^{-i\ee\vo(f_{j})}e^{it\vp(h)}e^{i\ee\vo(f_{j})}
&=e^{-it\ee\SPn{f_{j}}{h}}e^{it\vp(h)},
\\
e^{-i\ee\vo(f_{j})}e^{it\vo(h)}e^{i\ee\vo(f_{j})}
&=e^{it\vo(h)},
\end{align*} 
where $t\in\RR$ and 
$h\in L^2(\cA_j^{j+1}\times\ZZ_2)$ is real-valued.
Using also the relations \eqref{viona}
we thus obtain
\begin{align*}
U_{j}\,a^\sharp(h)\,U_{j}^*
&=a^\sharp(h)-\ee\,\SPn{h}{f_{j}}^\sharp\quad
\textrm{on}\;\dom([\Hf^{(j,j+1)}]^{\nf{1}{2}})\,.
\end{align*}
As a preparation for the succeeding subsections 
we next compute explicit representations of the
transformed operators
\begin{align}\label{def-checkH}
\check{\HAM}_{j+1}
&:=U_{j}\HAM_{j+1}U_{j}^*,
\quad\check{\D}_{j+1}:=U_{j}\D_{j+1} U_{j}^*,
\quad \chHf^{(j+1)}:=U_j\Hf^{(j+1)}U_j^*,
\end{align}
for all $j\in\NN_0$.
The action of $U_j$ can be expressed
by means of the quantities
\begin{equation*}
\V{F}_j^{j+1}:=f_{j}\,\V{k}+\V{G}_j^{j+1},\quad
\V{c}_j:=\SPn{f_{j}}{\V{F}_j^{j+1}+\V{G}_j^{j+1}}\,,
\quad b_j:=\SPn{f_{j}}{\omega\,f_{j}}\,,
\end{equation*}
which all depend on $\V{P}$ through $f_j\equiv f_j(\V{P})$.
They satisfy the bounds
\begin{align}\label{BoundsI}
|b_j|,|\V{c}_j|&\le\const\,\rho_j,
\qquad\big\|\omega^{\nu}\,\V{F}_j^{j+1}\big\|
\le\const\,\rho_j^{1+\nu},\quad \nu\in\{-\nf{1}{2},0,\nf{1}{2},1\}\,,
\end{align}
which follow from elementary computations
and $|\nabla E_j|\le1-\qmax<1$.
The constant $\const$ in \eqref{BoundsI} neither depends
on $j$, $\ee\in(0,\ee_0]$, nor on $\V{P}\in\ol{\cB}_\pmax$.
(It does depend on $\ee_0$ and $\pmax$.)
In the proof of Lemma~\ref{neu1}
the following identity leads to an absolutely crucial
cancellation of terms which, by counting powers of $\rho_j$,
seem to destroy the iterative analysis at first sight,
\begin{align}\label{rebecca100}
\V{F}_j^{j+1}\cdot\nabla E_j=\frac{1}{(2\pi)^{\nf{3}{2}}}
\frac{\id_{\cA_j^{j+1}}}{|\V{k}|^{\nf{1}{2}}}
\frac{\veps\cdot\nabla E_j}{1-\mr{\V{k}}\cdot\nabla E_j}
=\omega\,f_{j}\,,\quad \mr{\V{k}}:= |\V{k}|^{-1}\,\V{k}\,. 
\end{align}
Writing 
$\Hf^{(j+1)}\psi=\Hf^{(j)}\psi
+\sum_{\ell}\ad(\omega^{\nf{1}{2}}e_\ell)\,a(\omega^{\nf{1}{2}}e_\ell)\,\psi$, 
for some orthonormal basis, $\{e_\ell\}$, of $\HP_j^{j+1}$,
and for $\psi\in\dom(\Hf^{(j+1)})$,
we deduce that
\begin{align}\label{Hf-transformed}
\chHf^{(j+1)}
&=\Hf^{(j+1)}-\ee\,\vp(\omega\,f_{j})+
\ee^2b_j\qquad\quad\;
\textrm{on}\;\dom(\Hf^{(j+1)})\,,
\\\label{D-transformed}
\check{\D}_{j+1}
&=\D_j^{j+1}+\valpha\cdot\{\ee\,\vp(\V{F}_j^{j+1})
-\ee^2\V{c}_j\}\quad
\textrm{on}\;\dom(\Hf^{(j+1)})\,.
\end{align}

\begin{lem}
For all $\pmax>0$, we find $\ee_0>0$ such that, for all
$\V{P}\in\ol{\cB}_\pmax$, $\ee\in(0,\ee_0]$, and $j\in\NN_0$,
the operator $\check{\HAM}_{j+1}$ is self-adjoint
on $\dom(\Hf^{(j+1)})$ and essentially self-adjoint on $\sC_{j+1}$.
Moreover, $\check{\HAM}_{j+1}=|\check{\D}_{j+1}|+\Hf^{(j+1)}-\ee\,\vp(\omega\,f_{j})+
\ee^2b_j$.
\end{lem}

\begin{proof}
We may apply Lemma~\ref{le-lisa},
if $f_j$, $\V{g}_j$, $b_j$, and $\V{c}_j$ are chosen appropriately in Assumption~\ref{ariel},
which is then satisfied in view of 
\eqref{Bounds0} and \eqref{BoundsI}.
(Self-adjointness on $\dom(\Hf^{(j+1)})$ follows also
from \eqref{def-checkH} and Lemma~\ref{le-will}.)
\end{proof}

\smallskip

\noindent
To study the multiplicity of the ground state eigenvalue
of the renormalized Hamiltonians we also define the 
following dressing transforms,
\begin{equation}\label{willi1}
W_k(\V{P}):= \prod_{\ell=0}^{k-1}U_\ell(\V{P})\,,
\qquad k\in\NN\,.
\end{equation}
(While an analog of $W_k$ used in \cite{CFP2009,Pizzo2003} is defined only
in terms of $E_{k-1}$, our choice of $W_k$ involves the functions
$E_0,\ldots,E_{k-1}$. We think that this simplifies the analysis
by avoiding the discussion of the ``intermediate Hamiltonians'' 
appearing in \cite{CFP2009,Pizzo2003}.)
Without danger of confusion we shall consider $W_k$ as a unitary
operator in the various spaces $\HR_j$ with $j\in\NN\cup\{\infty\}$,
$j\ge k$, by identifying $W_k\equiv W_k\otimes\id$.
We abbreviate
\begin{align*}
g_k(\V{P})&:=\sum_{\ell=0}^{k-1} f_{\ell}(\V{P})\,,
&\V{F}_k(\V{P}):=\sum_{\ell=0}^{k-1} \V{F}_{\ell}^{\ell+1}(\V{P})\,,\;\qquad
\\
\V{C}_k(\V{P})&:=\SPn{g_k(\V{P})}{\V{G}_k+\V{F}_k(\V{P})}\,,
&B_k(\V{P}):=\SPn{g_k(\V{P})}{\omega\,g_k(\V{P})}\,,
\end{align*}
where $k\in\NN$. 
In order to apply the results of
Section~\ref{sec-gap} we actually have to consider the momenta appearing
in the transformation $W_k$ and in the Hamiltonian as
separate parameters. 
For the functions introduced in Assumption~\ref{ariel} do not
depend on the total momentum.
So, let $k,j\in\NN_0\cup\{\infty\}$, $k\le j$, $|\V{P}|\le\pmax$,
and $\V{Q}\in\RR^3$.
Then we define
\begin{equation}
\wt{\D}_k^j(\V{Q},\V{P})
:=\valpha\cdot\{\V{Q}-\Pf^{(j)}+\ee\,\vp(\V{F}_k(\V{P}))-\ee^2\V{C}_k(\V{P})\}
+\alpha_0\,,
\end{equation}
a priori as an essentially self-adjoint operator on $\sC_j$,
and then by taking its closure; see Lemma~\ref{le-esa-wD}. Furthermore,
we define, on the domain $\dom(\Hf^{(j)})$,
\begin{align}\label{def-wtHinfty}
\wt{\HAM}_k^j(\V{Q},\V{P})&:=
|\wt{\D}_k^j(\V{Q},\V{P})|+\Hf^{(j)}-\ee\,\vp(\omega\,g_k(\V{P}))+\ee^2
B_k(\V{P})\,.
\end{align}
As usual we write
$\wt{\HAM}_j:=\wt{\HAM}_j^j$, $j\in\NN_0\cup\{\infty\}$.
By definition and \eqref{Hf-transformed}\&\eqref{D-transformed},
\begin{equation}\label{willi2}
\wt{\HAM}_k^j(\V{Q},\V{P}):=W_k(\V{P})\,\HAM_k^j(\V{Q})\,W_k^*(\V{P})
\,,\quad k\le j\,,\;k<\infty\,,
\end{equation}
and analogously for $\wt{\D}_k^j(\V{Q},\V{P})$, $k<\infty$.
Notice that the unitary transforms  $W_k$ do not have a limit
since $\sum_0^\infty f_j\notin L^2(\RR^3\times\ZZ_2)$, so that 
\eqref{willi1} does not make sense in the case $k=\infty$ and there is 
no analog of \eqref{willi2} with $k=\infty$.
Finally, we abbreviate
\begin{equation}\label{willi22}
\wt{\HAM}_k^j(\V{P}):=\wt{\HAM}_k^j(\V{P},\V{P})\,,\quad|\V{P}|\le\pmax\,.
\end{equation}

\begin{lem}\label{le-wtHAM}
Let $\pmax>0$ be arbitrary and $\ee_0>0$ be sufficiently small.
Then the following holds, for all $\V{P}\in\ol{\cB}_\pmax$, $\V{Q}\in\RR^3$,
$\ee\in(0,\ee_0]$, and $j\in\NN_0\cup\{\infty\}$:

\smallskip

\noindent
(1) 
The operator $\wt{H}_j^\infty(\V{Q},\V{P})$ is self-adjoint on $\dom(\Hf)$
and essentially self-adjoint on $\sC_\infty$.
Moreover, $\wt{H}_j^\infty(\V{Q},\V{P})\to\wt{H}_\infty(\V{Q},\V{P})$, 
$j\to\infty$,
in the norm resolvent sense.

\smallskip

\noindent
(2) 
$\inf\Spec[\wt{H}_j^\infty(\V{Q},\V{P})]=E_j(\V{Q})$.

\smallskip

\noindent
(3) We find $\const,\const'>0$,
depending only on $\pmax,\ee_0$, such that
\begin{align}\nonumber
\pm(\wt{\HAM}_{\infty}(\V{P})-\wt{\HAM}_j^{\infty}(\V{P}))
&\le\const\,\ee\,
\big(\wt{\HAM}_{\infty}(\V{P})-E_j(\V{P})+\rho_j\big)
\\\label{adam33}
&\le\const'\ee\,
\big(\wt{\HAM}_{\infty}(\V{P})-E_\infty(\V{P})+\rho_j\big)\,.
\end{align}
\end{lem}

\begin{proof}
Upon choosing 
$f_j$, $\V{g}_j$, $b_j$, and $\V{c}_j$ appropriately 
we are in the situation of Assumption~\ref{ariel}.
In particular,
(1) follows directly from Lemma~\ref{le-lisa}.

(2): If $j\in\NN_0$, then $\wt{H}_j^\infty$ is unitarily
equivalent to ${H}_j^\infty$ (by \eqref{willi2}) and the statement is clear.
Since
$\wt{H}_k^\infty\to\wt{H}_\infty$, $k\to\infty$,
in the norm resolvent sense we have
$\inf\Spec[\wt{H}_k^\infty]\to\inf\Spec[\wt{H}_\infty]$.
This proves the assertion for $j=\infty$, since we also know that 
$E_k\to E_\infty$.

(3): The first inequality in \eqref{adam33} is a special
case of \eqref{adam3}. The second one follows by applying
\eqref{EnergyDiff}.
\end{proof}

\begin{rem}
We close this subsection by explaining why the following
bound to be used later on
is actually a special case of \eqref{debora},
 \begin{align}\label{debora2}
\big\|\big(U_j\partial_{\V{h}}^\ell\HAM_{j+1}U_j^*
-\partial_{\V{h}}^\ell\HAM_j^{j+1}\big)\,\psi\big\| 
&\le\const\,\ee\,\rho_j^{\nf{1}{2}}\,|\V{h}|^\ell\|(\Hf^{(j,j+1)}+\rho_j)^{\nf{1}{2}}\psi\|\,,
\end{align}
on $\ol{\cB}_\pmax$ and 
for all $\ell=1,2$, 
$\V{h}\in\RR^3$, $\ee\in(0,\ee_0]$, 
$j\in\NN_0$, $\psi\in \dom(\Hf^{(j+1)})$,
and some $\const,\ee_0>0$.
In fact, setting 
\begin{equation}\label{def-checkR}
\check{\R}_{j+1}(iy):=U_j\,\R_{j+1}(iy)\,U_j^*=(\check{D}_{j+1}-iy)^{-1},
\quad y\in\RR\,,\;j\in\NN_0\,,
\end{equation}
and using \eqref{primus} and \eqref{secundus} we observe that
$$
U_j\partial_{\V{h}}^\ell\HAM_{j+1}U_j^*\,\psi
=(-1)^{\ell+1}\lim_{\tau\to\infty}\int_{-\tau}^\tau
\check{\R}_{j+1}(iy)\,\{\valpha\cdot\V{h}\,\check{\R}_{j+1}(iy)\}^\ell\,\psi
\,\frac{y\,dy}{i\pi}\,.
$$
On the RHS we have the derivatives of $\check{H}_{j+1}$ w.r.t.
the total momentum {\em when the total momentum $\V{P}$ in the transformation
$U_j(\V{P})$ is kept fixed}. Thus, we are precisely in the situation
of Remark~\ref{rem-debora}.
\hfill$\Diamond$
\end{rem}


\subsection{Comparison of transformed and non-transformed operators}
\label{ssec-comp}

\noindent
In the following important lemma we compare the Hamiltonians
$\check{\HAM}_{j+1}$ and $\HAM_j^{j+1}$ on the
range of the projection $\Pi_j^{j+1}$
given by \eqref{def-Pijj+1} and \eqref{for-Pijj+1}. 
It will be applied in the proof
of Lemma~\ref{le-L} below.
Thanks to \eqref{rebecca100}
the error term $A_j$ in the next lemma is
of order $\ee\,\rho^2_j$. If it were of order
$\ee\,\rho_j$ only, then the proof of Lemma~\ref{le-L} did not work.
The term $\ee^2\vk_j$ in the next lemma, with 
$$
\vk_j:= \V{c}_j\cdot\nabla E_j-\SPn{f_j}{\omega\,f_j}
=\SPn{f_j}{\nabla E_j\cdot\V{G}_j^{j+1}}\ge0\,,\quad
\vk_j=\bigO(\rho_j)\,,
$$
does not cause any harm to Lemma~\ref{le-L} since it is non-negative.
(Even if it were negative one could exploit that it is of second order in $\ee$.)
Besides \eqref{def-checkR} we use the following notation
for resolvents of the Dirac operator,
$$
\R_j(iy)=(\D_j-iy)^{-1},\qquad\R_j^{j+1}(iy)=(\D_j^{j+1}-iy)^{-1}.
$$


\begin{lem}\label{neu1}
For all $\pmax>0$, there exist $\const,\ee_0>0$ such that,
for all $\ee\in(0,\ee_0]$ and $j\in\NN_0$, we find
some rank four operator $A_j\in\cB(\HR_{j+1})$ with 
$\|A_j\|\le \const\, \ee\,\rho^2_j$ and
\begin{align*}
\big(\check{\HAM}_{j+1}&-\HAM_j^{j+1}\!+\ee^2\vk_j\big)\,\Pi_j^{j+1}
=\ee(\Pi_j^\perp \nabla H_j \Pi_j)\otimes
\{\ad(\V{F}_j^{j+1})-\ee\,\V{c}_j\}P_{\Omega_j^{j+1}}\!+A_j.
\end{align*}
\end{lem}
%
%
\begin{proof} 
Let $\vt\in\dom(\Hf^{(j)})$.
Abbreviating 
$\Delta_j(y):=\check{\R}_{j+1}(\imath y)-\R_j^{j+1}(\imath y)$
we have
\begin{align}
I&:=\{|\check{\D}_{j+1}|-|\D_j^{j+1}|\}\,\vt\otimes \Omega_j^{j+1}
=\label{DiffAbsD}
\lim_{\tau\to \infty} 
\int_{-\tau}^\tau \!\Delta_j(y)(\vt\otimes\Omega_j^{j+1})\, 
\frac{\imath y\,dy}{\pi}\,.
\end{align}
A brief computation using \eqref{D-transformed} 
(and taking Lemma~\ref{le-dominique} into account) yields
\begin{align}\label{lucia0a}
(\check{D}_{j+1}-iy)\,\Delta_j(y)
&=\valpha\cdot\{\ee^2\V{c}_j-\ee\,\vp(\V{F}_j^{j+1})\}\,
\R_j^{j+1}(iy)\;\;\textrm{on}\;\dom(\Hf^{(j+1)}),
\\\label{lucia0b}
\Delta_j(y)(\vt\otimes\Omega_j^{j+1})
&=\check{\R}_{j+1}(\imath y)\,
\valpha\,(\R_j(iy)\,\vt)\otimes\{\ee^2\V{c}_j\,\Omega_j^{j+1}-\ee\,
|\V{F}_j^{j+1}\SR\}.
\end{align}
We recall that $|\V{c}_j|,\|\V{F}_j^{j+1}\|\le\const\,\rho_j$, 
$\|\R_j^k(iy)\|\le(1+y^2)^{-\nf{1}{2}}$, and observe that
the resolvent $\R_j^{j+1}(iy)$ leaves the subspaces
$\sX_n:=\HR_j\otimes[\sF_j^{j+1}]^{(n)}$, $n=0,1,\ldots\,$, invariant.
Moreover, the restriction of the field operator
$\vp(\V{F}_j^{j+1})$ is bounded as a map from
$\sX_0\oplus\sX_1$ to $\sX_0\oplus\sX_1\oplus\sX_2$,
with norm $\le\const\,\|\V{F}_j^{j+1}\|\le\const'\rho_j$.
Taking all these remarks into account we infer from \eqref{lucia0a} that
\begin{align}\label{lucia1}
\|\Delta_j(y)\!\!\upharpoonright_{\sX_0\oplus\sX_1}\!\!\|
\le\const\,\ee\,\rho_j\,(1+y^2)^{-1}\,.
\end{align}
Rearranging \eqref{lucia0b} using 
$\R_j^{j+1}(\psi\otimes\Omega_j^{j+1})=(\R_j\psi)\otimes\Omega_j^{j+1}$
we further obtain
\begin{align}\label{lucia2}
\Delta_j(y)(\vt\otimes\Omega_j^{j+1})
&=\big(\R_j(iy)\,\valpha\,\R_j(iy)\,\vt\big)
\otimes\{\ee^2\V{c}_j\,\Omega_j^{j+1}\}
-V+W\,,
\\\nonumber
V&:=\R_j^{j+1}(iy)\,\valpha\,(\R_j(iy)\,\vt)\otimes|\ee\,\V{F}_j^{j+1}\SR\,,
\\\nonumber
W&:=
\Delta_j(y)\,\valpha\,(\R_j(iy)\,\vt)\otimes\{\ee^2\V{c}_j\,\Omega_j^{j+1}-\ee\,
|\V{F}_j^{j+1}\SR\}\,.
\end{align}
We represent $V\in\HR_j\otimes[\sF_j^{j+1}]^{(1)}$
as $V=\{V_\lambda(\V{k})\}\in L^2(\cA_j^{j+1}\times\ZZ_2,\HR_j)$
with
\begin{align*}
V_\lambda(\V{k})=\ee\,\V{F}_j^{j+1}(\V{k},\lambda)
\,\R_j(\V{P}-\V{k},iy)\,\valpha\,\R_j(\V{P},iy)\,\vt\,,
\quad\textrm{a.e.}\;\V{k}\in\cA_j^{j+1},\;\lambda\in\ZZ_2\,.
\end{align*}
Since $|\V{k}|\le\rho_j$, for $\V{k}\in\cA_j^{j+1}$,
the resolvent identity implies
$\R_j(\V{P}-\V{k},iy)=\R_j(\V{P},iy)+\bigO(\rho_j\,(1+y^2)^{-1})$,
and together with $\|\V{F}_j^{j+1}\|\le\const\,\rho_j$ this gives
$$
V=\big(\R_j(iy)\,\valpha\,\R_j(iy)\,\vt\big)\otimes|\ee\,\V{F}_j^{j+1}\SR
+\bigO\big(\ee\,\rho_j^2\,(1+y^2)^{-\nf{3}{2}}\,\|\vt\|\big)\,.
$$
The inequality \eqref{lucia1} applied to the third 
member on the RHS of \eqref{lucia2} yields
\begin{align*}
\|W\|&\le\const\,\ee\,\rho_j(1+y^2)^{-1}\|\R_j(iy)\vt\|
\{\ee^2|\V{c}_j|+\ee\|\V{F}_j^{j+1}\|\}
\le\const'\ee^2\rho_j^2(1+y^2)^{-\nf{3}{2}}\|\vt\|.
\end{align*}
Altogether, employing the formula \eqref{primus}
for the derivative of $\HAM_j$ and using $\int_\RR dy/(1+y^2)<\infty$,
we see that the term in \eqref{DiffAbsD} can be written as
\begin{align}\nonumber
I&=
(\nabla\HAM_j\,\vt)\otimes
\{\ee\,|\V{F}_j^{j+1}\SR-\ee^2\V{c}_j\,\Omega_j^{j+1}\}
+\bigO\big(\ee\,\rho_j^2\,\|\vt\|\big)\,.
\end{align}
But we have $(\check{\HAM}_{j+1}-\HAM_j^{j+1})(\vt\otimes\Omega_j^{j+1})=I+II$ with
$$
II:=(\chHf^{(j+1)}-\Hf^{(j+1)})(\vt\otimes\Omega_j^{j+1})
=\ee\vt\otimes\{\ee\|\omega^{\nf{1}{2}}f_j\|^2\Omega_j^{j+1}
-|\omega\,f_j\SR\}\,.
$$
Applying these formulas for every $\vt\in\Ran(\Pi_j)\subset\dom(\Hf^{(j)})$
we arrive at
\begin{align*}
(\check{\HAM}_{j+1}-\HAM_j^{j+1})\,\Pi_j^{j+1}
=&
\ee\,\nabla H_j\,\Pi_j\otimes |\V{F}_j^{j+1}\SR\SL\Omega_j^{j+1}|
-\ee\,\Pi_j\otimes |\omega f_j\SR\SL\Omega_j^{j+1}|
\\
&+\big(
\ee^2\,\SPn{f_j}{\omega\,f_j}- \ee^2\,\V{c}_j\cdot \nabla H_j\big)\,
\Pi_j^{j+1}+\bigO(\ee\,\rho_j^2)\,,
\end{align*} 
and we conclude by means of
\eqref{rebecca100} and 
$\nabla H_j\,\Pi_j =\nabla E_j\,\Pi_j+\Pi_j^\perp\nabla H_j\,\Pi_j$.
\end{proof}


\subsection{Comparison of resolvents}
\label{ssec-res-comp}

\noindent
In order to control the difference of various operators
attached to succeeding scales
in the IAPT  it is necessary
to prepare a number of resolvent comparison estimates.
This is done in the present subsection.
To start with we introduce some notation for various resolvents:

For all $\V{P}\in\ol{\cB}_\pmax$ and $j\in\NN_0$, 
we abbreviate (compare \eqref{faysal} and \eqref{gap-j+1})
\begin{align}\nonumber
\cR_j(\V{P},z)&:=\big(\HAM_j(\V{P})-E_j(\V{P})+z\big)^{-1},\qquad \Re z>0\,,
\\\nonumber
\cR_j^{j+1}(\V{P},z)&:=\big(\HAM_j^{j+1}(\V{P})-E_j(\V{P})+z\big)^{-1},\quad \Re z>0\,,
\\\nonumber
\cR_j^\perp(\V{P})&:=\big(\HAM_j(\V{P})\,\Pi_j^\perp(\V{P})-E_j(\V{P})\big)^{-1}
\,\Pi_j^\perp(\V{P})\,,
\\\label{def-wtR}
\wh{\cR}_j(\V{P},\V{k})&:=\big(\HAM_j(\V{P}-\V{k})-E_j(\V{P})+|\V{k}|\big)^{-1},
\quad\V{k}\in\RR^3.
\end{align}
We shall see in the proof of Lemma~\ref{rr-le}(1) that
the resolvent in \eqref{def-wtR} is actually well-defined.
Since the spectral gap of $\check{\HAM}_{j+1}$
is not smaller than $q\,\rho_j/4$
we further know that the restriction of $\check{\HAM}_{j+1}-E_{j+1}$
to the range of $\check{\Pi}_{j+1}^\bot:=\id-\check{\Pi}_{j+1}$, with
\begin{equation}\label{def-checkPi}
\check{\Pi}_{j+1}:=\id_{\{E_{j+1}\}}(\check{H}_{j+1})\,,
\end{equation}
is continuously invertible, for every $z\in\CC$, $\Re z>-q\,\rho_j/4$.
In fact, defining
\begin{align}\label{li1}
(\HAM_j^{j+1})^\bot\!&:=\HAM_j^{j+1}(\Pi_j^{j+1})^\bot\!,
\;\;\;(\RES_j^{j+1})^\bot:=((\HAM_j^{j+1})^\bot-E_j)^{-1}(\Pi_j^{j+1})^\bot\!,
\\\label{li2}
\check{\HAM}_{j+1}^\bot&:=\check{\HAM}_{j+1}\check{\Pi}_{j+1}^\bot,
 \quad\quad\,\,\check{\RES}_{j+1}^\bot(z)
:=(\check{\HAM}_{j+1}^\bot-E_{j+1}+z)^{-1}
\,\check{\Pi}_{j+1}^\bot\,,
\end{align}
on $\ol{\cB}_\pmax$ and for $j\in\NN_0$ and $z$ as above,
we infer from \eqref{diff-E} and \eqref{gap-j+1} that
$$
\|(\RES_j^{j+1})^\bot\|=(\gap_j^{j+1})^{-1}\le\const\,\rho_j^{-1},\quad
\|\check{\RES}_{j+1}^\bot(z)\|\le(q\rho_j/4+\Re z)^{-1}.
$$

\begin{lem}\label{rr-le}
Pick $\pmax>0$ and let $\delta_j$ be some real-valued function
of small $\ee>0$ with
$\delta_j=\bigO(\ee\,\rho_j)$. 
Then we find $\const,\ee_0>0$ such that the following
holds, for all $\V{P}\in\ol{\cB}_\pmax$, $\ee\in(0,\ee_0]$, and $j\in\NN_0$:

\smallskip

\noindent(1)
For all $\psi\in \HR_j$ and $\V{k}\in \cA_{j+1}$, 
\begin{equation}\label{ResRepl1}
\big\|\wh{\RES}_j(\V{P},\V{k})\,\Pi_j^\bot(\V{P})\,\psi\big\|
\le \const\,\|{\RES}_j^\perp(\V{P})\,\psi\|\,.
\end{equation}
(2)
For all $\psi\in\HR_j\otimes\{\CC\,\Omega_j^{j+1}\oplus L^2(\cA_j^{j+1}\times\ZZ_2)\}
\subset\HR_{j+1}$, $s\in\{\nf{1}{2},1\}$,  
we have the bounds
\begin{align}\nonumber
\big\|\check{\RES}_{j+1}^\bot(\delta_j)^s\, (\Pi_j^{j+1})^\perp \psi
&-\check{\Pi}_{j+1}^\bot((\RES_j^{j+1})^\bot)^s\,\psi\big\|
\\
&
\le\,\const\,\ee\,\rho_j^{1-s}\big(\|\psi\|+
\|(\RES_j^{j+1})^\perp\,\psi\|\big)\,,\label{uwe3}
\\\label{BoundRes}
\big\|\check{\RES}_{j+1}^\bot(\delta_j)\,({\Pi}_j^{j+1})^\perp\,\psi\big\|
&\le(1+\const\,\ee)\,\|(\RES_j^{j+1})^\bot\,\psi\|
+\const\,\ee\,\|\psi\|\,.
\end{align}
\end{lem}
%
%
\begin{proof}
(1): On account of \eqref{diff-E} and $|\V{k}|\ge\rho_{j+1}=\rho_j/2$ we have 
\begin{equation}\label{rene}
r_j(\V{P},\V{k}):=E_j(\V{P}-\V{k})-E_j(\V{P})+|\V{k}|\ge\rho_j/4\,.
\end{equation}
Hence, $E_j(\V{P})-|\V{k}|$ belongs to the resolvent set of both
$\HAM_j(\V{P}-\V{k})$ and $\HAM_j(\V{P})$ and \eqref{sergio3} 
together with the second resolvent identity implies
\begin{align*}
\big\|\big(\wh{\RES}_j(\V{P},\V{k})-\RES_j(\V{P},|\V{k}|)\big)\,\Psi\big\|
&\le\const|\V{k}|\,\big\|\HAM_j(\V{P}-\V{k})^{\nf{1}{4}}\wh{\RES}_j(\V{P},\V{k})\big\|\,
\big\|\RES_j(\V{P},|\V{k}|)\,\Psi\big\|\,,
\end{align*}
for all $\Psi\in\HR_j$.
Here $|\V{k}|\le\rho_j$ and, by the spectral calculus,
the first norm on the RHS is not greater than
$E_j(\V{P}-\V{k})^{\nf{1}{4}}\big/r_j(\V{P},\V{k})
\le\const(\pmax)/\rho_j$.
Choosing $\Psi:=\Pi_j^\perp\psi$ 
we obtain \eqref{ResRepl1} since, of course,
$\|\RES_j(\V{P},|\V{k}|)\,\Pi_j^\perp\psi\|\le\|\RES_j^\perp(\V{P})\,\psi\|$.

(2): We may apply Lemma~\ref{le-lisa} with
$\V{A}_j^{j+1}=\vp(\V{F}_j^{j+1})-\ee\V{c}_j$
and $\Phi_0^j=0$, 
$\Phi_j^{j+1}=\vp(\omega\,f_j)+\ee\SPn{f_j}{\omega\,f_j}$.
Indeed, let $t\ge0$, $\ve>0$, and $z:=t+i\ve$. Then 
$\delta_j=\bigO(\ee\,\rho_j)$, $|E_j-E_{j+1}|=\bigO(\ee\,\rho_j)$,
 and the first inequality in \eqref{susi} imply
\begin{align*}
N_z^{(1)}&:=\big\|\check{\Pi}_{j+1}^\perp\big(\check{\RES}_{j+1}(z+\delta_j)-
\RES_j^{j+1}(z)\big)\,(\Pi_j^{j+1})^\perp\psi\big\|
\\
&\le\const\,\ee\,\rho_j^{\nf{1}{2}}\,
\|\check{\Pi}_{j+1}^\perp\check{\RES}_{j+1}(z+\delta_j)\|
\\
&\quad\cdot
\big\|\,|\D_j^{j+1}|^{\nf{1}{2}}(\Hf^{(j,j+1)}\!+\rho_j)^{\nf{1}{2}}(\Pi_j^{j+1})^\perp
\RES_j^{j+1}(z)\,\psi\big\|\,.
\end{align*}
Here,  \eqref{diff-E} and \eqref{gap-j+1} permit to get 
\begin{align*}
\|\check{\RES}_{j+1}(z+\delta_j)\,\check{\Pi}_{j+1}^\perp\|
&\le (\gap_{j+1}+\delta_j+t)^{-1}\le \const\, (\rho_j+t)^{-1},
\end{align*}
uniformly in $\ve>0$ and small $\ee>0$.
Moreover, $(\Hf^{(j,j+1)}\!+\rho_j)^{\nf{1}{2}}$ commutes strongly with
$(\Pi_j^{j+1})^\perp$ and $\RES_j^{j+1}(z)$.
Since also 
$|\D_j^{j+1}|\le(\HAM_j^{j+1}-E_j)+E_j$
and $E_j\le\const(\pmax)$ on $\ol{\cB}_\pmax$ we deduce that
\begin{align*}
\theta&:=
\lim_{\ve\searrow0}
\big\|\,|\D_j^{j+1}|^{\nf{1}{2}}(\Hf^{(j,j+1)}\!+\rho_j)^{\nf{1}{2}}(\Pi_j^{j+1})^\perp
\RES_j^{j+1}(z)\,\psi\big\|
\\
&\le
\big\|(\Hf^{(j,j+1)}\!+\rho_j)^{\nf{1}{2}}
\big((\RES_j^{j+1})^\perp\big)^{\nf{1}{2}}\,\psi\big\|
+\const\,\big\|(\Hf^{(j,j+1)}\!+\rho_j)^{\nf{1}{2}}
(\RES_j^{j+1})^\perp\,\psi\big\|\,.
\end{align*}
Next, we use that $\psi$ and, hence, $((\RES_j^{j+1})^\perp)^s\,\psi$,
belong to 
$\sX_j:=\HR_j\otimes\{\CC\,\Omega_j^{j+1}\oplus L^2(\cA_j^{j+1}\times\ZZ_2)\}$
and that $\eta\in\sX_j$
entails
$\|(\Hf^{(j,j+1)}\!+\rho_j)^{\nf{1}{2}}\eta\|\le(2\rho_j)^{\nf{1}{2}}\|\eta\|$.
Taking this observation into account we obtain
$$
\theta\le\const\,\rho_j^{\nf{1}{2}}\big(\|((\RES_j^{j+1})^\perp)^{\nf{1}{2}}\psi\|+
\|(\RES_j^{j+1})^\perp\,\psi\|\big)
\le\const'\rho_j^{\nf{1}{2}}\big(\|\psi\|+
\|(\RES_j^{j+1})^\perp\,\psi\|\big)\,.
$$
Altogether we arrive at
\begin{align*}
\lim_{\ve\searrow0}N_{t+i\ve}^{(1)}\le
{\const\,\ee\,\rho_j}{(\rho_j+t)^{-1}}
\big(\|\psi\|+
\|(\RES_j^{j+1})^\perp\,\psi\|\big)\,,
\end{align*}
for every $\psi\in\sX_j$. For $t=0$, this is \eqref{uwe3} with $s=1$,
which immediately implies \eqref{BoundRes}.

Let $N_0^{(\nf{1}{2})}$ denote the LHS of \eqref{uwe3} with $s=1/2$.
We have
$\check{\RES}_{j+1}^\bot(\delta_j)^{\nf{1}{2}} 
= \int_0^\infty\check{\RES}_{j+1}^\bot(t+\delta_j)\,t^{-\nf{1}{2}}dt/\pi$
and an analogous representation of the 
square root of $(\RES_j^{j+1})^\perp$. Therefore,
\begin{align*}
N_0^{(\nf{1}{2})}&\le\int_0^\infty \!N_t^{(1)}\,\frac{dt}{\pi t^{\nf{1}{2}}}
\le\const\,\ee\int_0^\infty\!\frac{\rho_j}{\rho_j+t}\,\frac{dt}{\pi t^{\nf{1}{2}}}
\big(\|\psi\|+
\|(\RES_j^{j+1})^\perp\,\psi\|\big)\,,
\end{align*}
which is \eqref{uwe3} with $s=1/2$, as the integral on the RHS equals $\rho_j^{\nf{1}{2}}$.
\end{proof}


\section{Regularity of the ground state energy}\label{sec-C2}

\noindent
The main objective of this section is to show that $E_\infty$
is twice continuously differentiable and strictly convex
on $\cB_\pmax$, at least for small $\ee>0$. The starting point
of our analysis are the Hellmann-Feynman type formulas
\eqref{for-1Deriv} and \eqref{for-2Deriv} for the first and
second derivatives of the approximating functions $E_j$, $j\in\NN$.
To make use of these formulas
we first have to control the quantities
\begin{equation}\label{def-K}
K_j^{(s)}:=\|
({\RES}_{j}^\bot)^s\,{\Pi}_{j}^\bot\,
\nabla{\HAM}_{j}\,{\Pi}_{j}\|_{\HS}\,,\qquad j\in\NN_0,\;s\in \{ \nf{1}{2}\,,1\}\,.
\end{equation}
Here and henceforth we set
$\|\V{T}\|_*^2:=\|T_1\|_*^2+\|T_2\|_*^2+\|T_3\|_*^2$, for a triple
of operators $\V{T}=(T_1,T_2,T_3)$ and a suitable norm $\|\cdot\|_*$.
The expression in \eqref{def-K} is well-defined since
$\Ran({\Pi}_{j})\subset\dom(\HAM_j)=\dom(\Hf^{(j)})
\subset\dom(\partial_{\mu}\HAM_j)$ and $\tr[{\Pi}_{j}]=4$.
From \eqref{sergio4} we immediately obtain the a priori bounds
\begin{equation}\label{sergio7}
\|{\Pi}_{j}^\bot\,
\nabla{\HAM}_{j}\,{\Pi}_{j}\|_{\HS}\le2\,\|\nabla{\HAM}_{j}\,{\Pi}_{j}\|
\le\const\,,
\qquad
K_j^{(s)}\le\const\,\rho_j^{-s},
\end{equation}
uniformly on $\ol{\cB}_\pmax$.
If we seek for a better bound on $K_j^{(s)}$, showing that it actually
stays bounded, as $j\to\infty$, 
we must not estimate the effect of the
resolvent in \eqref{def-K} trivially, as we just did, by its norm which is
of order $\rho_j^{-s}$. It is a guiding theme in
Pizzo's iterative analytic perturbation theory
that expressions like $K_j^{(s)}$ are
estimated inductively, i.e. by relating $K_j^{(s)}$
to its precessor $K_{j-1}^{(s)}$.
In fact, the estimation of $K_j^{(s)}$ is enmeshed in
the analysis of the ground state
projections of $\wt{\HAM}_j^\infty$ and vice versa.
The convergence of the latter is shown in the
next section as an immediate consequence of the results
obtained in the present one.

Let $\kappa,\pmax>0$ be fixed in the whole section.
The first step is to bound the difference of the projections
defined in \eqref{def-Pijj+1} and \eqref{def-checkPi} in terms
of $K_{j}^{(1)}$:

\begin{lem}\label{le-L}
We find $\const,\ee_0>0$ such that, for all 
$\V{P}\in\ol{\cB}_\pmax$, $\ee\in(0,\ee_0]$, and $j\in\NN_0$,
\begin{align*}
\|\check{\Pi}_{j+1}-\Pi_j^{j+1}\|_{\HS},\,
\|\check{\Pi}_{j+1}^\bot \Pi_j^{j+1}\|_{\HS},\,
\|(\Pi_j^{j+1})^\bot \check{\Pi}_{j+1}\|_{\HS}
\le\const\,\ee\,\rho_j\,(K_{j}^{(1)}+1)\,.
\end{align*}
\end{lem}

\begin{proof}
\emph{Step 1.}
First, we consider $\|\check{\Pi}_{j+1}^\bot\Pi_j^{j+1}\|_{\HS}$.
Let $\vk_j\ge0$ be the constant appearing in Lemma~\ref{neu1}.
By \eqref{diff-E} and \eqref{gap-j+1}
we know that the restriction of $\check{\HAM}_{j+1}-E_j+\ee^2\vk_j$
to the range of $\check{\Pi}_{j+1}^\bot$ is
continuously invertible. 
In fact, \eqref{gap-j+1} implies
$
\|(\check{\HAM}_{j+1}^\bot-E_j+\ee^2\vk_j)^{-1}\check{\Pi}_{j+1}^\bot\|
\le\const\, \rho_j^{-1}
$, for small $\ee$, where we use the notation \eqref{li2}.
Since $\check{\HAM}_{j+1}$ and $\HAM_j^{j+1}$ have the same domain, which
contains the range of $\Pi_j^{j+1}$, this permits us to write
\begin{equation}\nonumber
\check{\Pi}_{j+1}^\bot\,\Pi_j^{j+1}
=
(\check{\HAM}_{j+1}^\bot-E_j+\ee^2\vk_j)^{-1}\,\check{\Pi}_{j+1}^\bot
 \,(\check{\HAM}_{j+1}-\HAM_j^{j+1}+\ee^2\vk_j)\,\Pi_j^{j+1}.
\end{equation}
Applying Lemma~\ref{neu1} we obtain
\begin{align*}
\check{\Pi}_{j+1}^\bot\Pi_j^{j+1}
=
(\check{\HAM}_{j+1}^\bot-E_j+\ee^2\vk_j)^{-1}\check{\Pi}_{j+1}^\bot
\,\Theta_j+\bigO(\ee\,\rho_j)\,,
\end{align*}
where the rank four operator
\begin{equation*}
\Theta_j:=
\ee\,(\Pi_j^\bot\,\nabla\HAM_j\,\Pi_j)\otimes\{\ad(\V{F}_j^{j+1})-\ee\,\V{c}_j\}\,P_{\Omega_j^{j+1}}
\end{equation*}
satisfies $\|\Theta_j\|_{\HS}\le\const\,\ee\,\rho_j$ 
according to \eqref{BoundsI} and \eqref{sergio7}.
Since we have
$\Ran({\Theta}_j)\subset 
\Ran(\Pi_j^{j+1})^\perp\cap\{\HR_j\otimes(\CC\,\Omega_j^{j+1}\oplus L^2(\cA_j^{j+1}\times\ZZ_2))\}$ 
we may apply \eqref{BoundRes} 
(with $\delta_j=E_{j+1}-E_j$ using $\vk_j\ge0$)
to deduce that
\begin{equation*}
\|\check{\Pi}_{j+1}^\bot\,\Pi_j^{j+1}\|_{\HS}\le (1+\const \,\ee)\,
\| ({\RES}_j^{j+1})^\perp\,\Theta_j\|_{\HS}+\const\,\ee\,\rho_j\,.
\end{equation*}
Next, we consider the action of $({\RES}_j^{j+1})^\perp$ on $\Theta_j$ more closely,
taking into account that
$(\Pi_j^{j+1})^\bot=\Pi_j^\bot\otimes P_{\Omega_j^{j+1}}+\id\otimes P_{\Omega_j^{j+1}}^\bot$ and, hence,
\begin{align*} 
({\RES}_j^{j+1})^\perp\!\!\upharpoonright_{ \HR_j\otimes \CC \Omega_j^{j+1}}&=\RES_j^\perp,
\\
({\RES}_j^{j+1})^\perp\!\!\upharpoonright_{ \HR_j\otimes L^2(\cA_j^{j+1}\times \ZZ_2)}
&=\bigoplus_{\lambda\in\ZZ_2}\int_{\cA_j^{j+1}}^\oplus
\big(\HAM_j(\V{P}-\V{k})-E_j(\V{P})+|\V{k}|\big)^{-1}d^3\V{k}\,;
\end{align*} 
recall \eqref{li2000} and \eqref{rene}.
The resolvents under the direct integral are just the ones in \eqref{def-wtR}.
For the member of $\Theta_j$ without creation operator we thus have
$$
({\RES}_j^{j+1})^\perp(\Pi_j^\bot\,\nabla\HAM_j\,\Pi_j\otimes P_{\Omega_j^{j+1}})
=(\RES_j^\bot\nabla\HAM_j\,\Pi_j)\otimes P_{\Omega_j^{j+1}}\,.
$$
Applying \eqref{ResRepl1} we further deduce that
\begin{equation*}
\big\|({\RES}_j^{j+1})^\perp(\Pi_j^\bot\,\nabla\HAM_j\,\Pi_j)\otimes
\ad(\V{F}_j^{j+1})\,P_{\Omega_j^{j+1}}\big\|_{\HS}
\le\const \,\|\V{F}_j^{j+1}\|\|\RES_j^\bot\nabla\HAM_j\,\Pi_j\|_{\HS}\,,
\end{equation*}
by computing the Hilbert-Schmidt norm on the LHS in a basis of the form
$\{e_i\otimes e_\ell'\}$, where $\{e_i\}$ and $\{e_\ell'\}\ni\Omega_j^{j+1}$ are
orthonormal bases of $\HR_j$ and $\sF_j^{j+1}$, respectively.
Using also the a priori bounds \eqref{sergio7},
$\|\Theta_j\|_{\HS}\le \const\,\ee\,\rho_j$, and $\|\V{F}_j^{j+1}\|,|\V{c}_j|\le\const\,\rho_j$, 
we obtain, for sufficiently small $\ee$,
\begin{equation}\label{neu-eq1}
\|\check{\Pi}_{j+1}^\bot\, \Pi_j^{j+1}\big\|_{\HS}
\le\const\,\ee\,\rho_j\,(K_{j}^{(1)}+1)\le 1/2\,.
\end{equation}
\emph{Step 2.} Next, we turn to $\|\check{\Pi}_{j+1}-\Pi_j^{j+1}\|_{\HS}$ 
and $\|(\Pi_j^{j+1})^\bot\,\check{\Pi}_{j+1}\|_{\HS}$.

Choose any orthonormal basis, $\{\chi_k\}_{k=1}^\infty$, of $\HR_{j+1}$
with $\check{\Pi}_{j+1}\,\chi_k=\chi_k$, for $k=1,2,3,4$, and
$\check{\Pi}_{j+1}^\bot\,\chi_k=\chi_k$, for $k\ge 5$. Then
\begin{align}\label{eq-neu7}
\|\check{\Pi}_{j+1}-\Pi_j^{j+1}\|^2_{\HS}
&=\sum_{k=1}^\infty\|(\check{\Pi}_{j+1}-\Pi_j^{j+1})\,\chi_k\|^2\\ \nonumber
&=\sum_{k=1}^4\|(\Pi_j^{j+1})^\bot\,\check{\Pi}_{j+1}\,\chi_k\|^2
+\sum_{k=5}^\infty\|\Pi_j^{j+1}\,\check{\Pi}_{j+1}^\bot\,\chi_k\|^2.
\end{align}
The last expression on the right side of \eqref{eq-neu7} is just
\begin{equation}\label{neu-eq2}
\sum_{k=5}^\infty\|\Pi_j^{j+1}\,\check{\Pi}_{j+1}^\bot\,\chi_k\|^2=
\|\Pi_j^{j+1}\check{\Pi}_{j+1}^\bot\|_{\HS}^2=\|\check{\Pi}_{j+1}^\bot\Pi_j^{j+1}\|_{\HS}^2\,.
\end{equation}
Thanks to Theorem \ref{Thm:Gap}(1) we may pick a normalized 
$\phi_1 \in \Ran(\Pi_j^{j+1})$ satisfying $X_1\,\phi_1=\phi_1$.
By \eqref{neu-eq1} we have $\|\check{\Pi}_{j+1}\, \phi_1\|\ge 1/2$. 
Thus, we may define $\chi_1 := \check{\Pi}_{j+1} \phi_1 /\|\check{\Pi}_{j+1} \phi_1\|$. 
Clearly, $X_1\,\chi_1=\chi_1$. Since $X_1$, $X_2$, and $\vt$ commute with $U_j$ we 
may choose
$\chi_2:=X_2\,\chi_1$, $\chi_3:= \vt\,\chi_1$, and $\chi_4=X_2\,\vt\,\chi_1$;
see Theorem \ref{Thm:Gap}(2).
Moreover, %
\begin{align} \nonumber
\|(\Pi_j^{j+1})^\bot&\,\check{\Pi}_{j+1}\|^2_{\HS}=
\sum_{k=1}^4\|(\Pi_j^{j+1})^\bot\,\check{\Pi}_{j+1}\,\chi_k\|^2
=4\,\|(\Pi_j^{j+1})^\bot\,\check{\Pi}_{j+1}\,\chi_1\|^2
\\ \label{neu-eq3}
&\le 16 \,\|(\Pi_j^{j+1})^\bot\check{\Pi}_{j+1}\,\Pi_j^{j+1}\phi_1\|^2
\le 16 \,\|\check{\Pi}_{j+1}^\bot\,\Pi_j^{j+1}\|^2_{\HS}.
\end{align}
Combining \eqref{neu-eq1}, \eqref{eq-neu7}, \eqref{neu-eq2}, and \eqref{neu-eq3} we conclude the proof.
\end{proof}
%
%
\begin{lem}\label{le-rek}
There exist $\const,\ee_0>0$ such that, for all $\V{P}\in\ol{\cB}_\pmax$,
$\ee\in(0,\ee_0]$, and $j\in\NN_0$,
\begin{equation}\label{for-rek}
K_{j+1}^{(1)} \le(1+\const\,\ee)\,K_{j}^{(1)}+\const\,\ee\,,
\quad  
|K_{j+1}^{(\nf{1}{2})}-K_{j}^{(\nf{1}{2})}|\le \const \,\ee\, \rho_j^{\nf{1}{2}}(K_j^{(1)}+1).
\end{equation}
\end{lem}

\begin{proof}
Let $s\in\{\nf{1}{2},1\}$. 
We derive a bound on the difference between
the two numbers
\begin{align*}
a_s&:=\big\|(\check{\RES}_{j+1}^\bot)^{s}\,
(U_j\,\nabla\HAM_{j+1}\,U_j^*)\,\check{\Pi}_{j+1}\big\|_{\HS}\,,
\\
a_s'&:=\big\|\{(\RES_j^{j+1})^\bot\}^{s}\,
\nabla\HAM_j^{j+1}\,\Pi_j^{j+1}\big\|_{\HS}\,,
\end{align*}
with
\begin{equation}\label{wdh-chR}
\check{\RES}_{j+1}^\bot
:=\big(\check{H}_{j+1}\,\check{\Pi}_{j+1}^\bot-E_{j+1}\big)^{-1}\,
\check{\Pi}_{j+1}^\bot=U_j\,\RES_{j+1}^\bot\,U_j^*\,.
\end{equation}
Notice that $a_s=K^{(s)}_{j+1}$ by unitary invariance of the Hilbert-Schmidt norm,
and $a_s'=K^{(s)}_j$ because 
$\Pi_j^{j+1}=\Pi_j\otimes P_{\Omega_j^{j+1}}$
and $(\RES_j^{j+1})^\bot$ and $\nabla\HAM_j^{j+1}$
reduce to $\RES_j^\bot$ and $\nabla\HAM_j$, respectively,
on $\HR_j\otimes\CC\,\Omega_{j}^{j+1}$.
To compare these two numbers we successively replace each operator in $a_s$
by a corresponding one associated with the preceding scale $j$.
More precisely, we estimate
\begin{align}\label{nicole}
|a_s-a_s'|&\le
|a_s-b_s|+|b_s-c_s|+|c_s-d_s|+|d_s-e_s|+|e_s-a_s'|\,,
\end{align}
with
\begin{align*}
b_s&:=\big\|(\check{\RES}_{j+1}^\bot)^{s}\,
(U_j\,\nabla\HAM_{j+1}\,U_j^*)\,\Pi_j^{j+1}\big\|_{\HS}\,,
\\
c_s&:=\big\|(\check{\RES}_{j+1}^\bot)^{s}\,
\nabla\HAM_j^{j+1}\,\Pi_j^{j+1}\big\|_{\HS}\,,
\\
d_s&:=\big\|(\check{\RES}_{j+1}^\bot)^{s}\,(\Pi_j^{j+1})^\bot\,\nabla\HAM_j^{j+1}\,\Pi_j^{j+1}\big\|_{\HS}\,,
\\
e_s&:=\big\|\check{\Pi}^\perp_{j+1}\,\{({\RES}_j^{j+1})^\perp\}^s\,\nabla\HAM_j^{j+1}\,\Pi_j^{j+1}\big\|_{\HS}\,.
\end{align*}
Each of the following five steps deals with one of the absolute values
on the RHS of \eqref{nicole}.

\smallskip

\noindent
{\em Step 1.}
First, we replace the projection
$\check{\Pi}_{j+1}$ in $a_s$ by $\Pi_j^{j+1}$,
\begin{align}
\label{lara1}
|a_s-b_s|
&\le 
\big\|U_j\,(\nabla\HAM_{j+1}(\RES_{j+1}^\bot)^s)\,U_j^*\big\|\, 
\big\|\Pi_{j}^{j+1}-\check{\Pi}_{j+1}\big\|_{\HS}.
\end{align}
From \eqref{sergio4} and the spectral calculus we deduce that 
\begin{align}
\big\|
\nabla\HAM_{j+1}\,(\RES_{j+1}^\bot)^s\,\big\|
&\le\|\nabla\HAM_{j+1}\,\HAM_{j+1}^{-\nf{1}{2}}\|\,
\Big\|\frac{\HAM_{j+1}^{\nf{1}{2}}\,\Pi_{j+1}^\bot}{(\HAM_{j+1}-E_{j+1})^s}\,
\Big\|\le\const\,\rho_j^{-s}.\label{lara1a}
\end{align}
By Lemma \ref{le-L} we obtain
$|a_s-b_s|\le\const\,\ee\,\rho_j^{1-s}\,(K_{j}^{(1)}+1)$.

\smallskip

\noindent
{\em Step 2.} Next, we replace the velocity 
$U_j\,\nabla\HAM_{j+1}\,U_j^*$
in $b_s$ by $\nabla\HAM_j^{j+1}$.
This is just a direct application of \eqref{debora2} which 
together with 
$\|({\check{\cR}}_{j+1}^\bot)^s\|\le\const\,\rho_j^{-s}$ and  
$(\Hf^{(j,j+1)}+\rho_j)^{\nf{1}{2}}\,\Pi_j^{j+1}=\rho_j^{\nf{1}{2}}\,\Pi_j^{j+1}$ 
implies
\begin{equation*}
\big\|(\check{\cR}_{j+1}^\bot)^s\,
\big(U_j\nabla\HAM_{j+1}U_j^*-\nabla\HAM_j^{j+1}\big)
\,\Pi_{j}^{j+1}\big\|_{\HS} 
\le\const\,\ee\,\rho_j^{1-s}.
\end{equation*}
By definition of $b_s$ and $c_s$ we thus have
$|b_s-c_s|\le\const\,\ee\,\rho_s^{1-s}$.

\smallskip

\noindent
{\em Step 3.}
By \eqref{for-1Deriv}, i.e. 
$\Pi_j^{j+1}\nabla\HAM_{j}^{j+1}\,\Pi_j^{j+1}
=\nabla E_{j}\,\Pi_j^{j+1}$, 
and by Lemma \ref{le-L} we obtain 
\begin{align}\nonumber
|c_s-d_s|
&\le
|\nabla E_j|\,
\|(\check{\RES}_{j+1}^\bot)^s\|\,
\big\|\check{\Pi}_{j+1}^\bot\,\Pi_j^{j+1}\big\|_{\HS}
\le
\const\,\ee\,\rho_j^{1-s}\,\big(K_j^{(1)}+1\big)\,.
\end{align}
{\em Step 4.} We employ \eqref{uwe3} 
with $\delta_j:=0$
to deduce that 
$|d_s-e_s|\le \const\,\ee\,\rho_j^{1-s}(a'_1+ 1)$,
where $a'_1=K_j^{(1)}$.

\smallskip

\noindent
{\em Step 5.} By Lemma \ref{le-L} and the fact that the Hilbert-Schmidt
norm dominates the operator norm,
\begin{equation}
|e_s-a_s'|\le \|\check{\Pi}_{j+1}\,(\Pi_j^{j+1})^\perp\|_{\HS}\,a_s'\le 
\const\,\ee\,\rho_j\,(K_j^{(1)}+1)\,a_s'\,.
\end{equation}
Collecting the results of the above steps
and using $a_s=K_{j+1}^{(s)}$, $a'_s=K_{j}^{(s)}$, and 
$a_s'\le \const\,\rho_j^{-s}$ (by \eqref{sergio7}),
we arrive at 
$|K_{j+1}^{(s)}-K_{j}^{(s)}|\le \const\,\ee\,\rho_j^{1-s}(K_{j}^{(1)}+1)$,
which implies \eqref{for-rek}.
\end{proof}

%
%
\begin{cor}\label{Cor-Karl}
There exist $\const,\ee_0>0$ such that, for all $\V{P}\in\ol{\cB}_\pmax$,
$\ee\in(0,\ee_0]$, and $j\in\NN_0$, we have
\begin{align*}
K_{j}^{(1)}&\le(1+\const\,\ee)^{j}-1\,,\quad 
K_{j}^{(\nf{1}{2})}\le\const\,\ee\,
\big(1- (\nf{1}{2})^{\nf{j}{2}}(1+\const\,\ee)^j\big)\,,
\end{align*}
as well as
\begin{align*}
\|\check{\Pi}_{j+1}-\Pi_j^{j+1}\|_{\HS},\,
\| \check{\Pi}_{j+1}^\bot \Pi_j^{j+1}\|_{\HS},\,
\| (\Pi_j^{j+1})^\bot \check{\Pi}_{j+1}\|_{\HS}
\le\const\,\ee\,(1+\const\,\ee)^j\rho_j\,.
\end{align*}
\end{cor}
\begin{proof} First, we prove the bound on $K^{(1)}_j$ by means of \eqref{for-rek} using $K^{(s)}_0=0$,
which follows from the fact that $\partial_{\V{h}}\HAM_0$ and $\Pi_0$ are merely
multiplication operators which commute with each other.
After that we use again \eqref{for-rek} to obtain the recursion formula
$K^{(\nf{1}{2})}_{j+1}\le K^{(\nf{1}{2})}_{j}+\const\,\ee\,(1+\const\,\ee)^j\rho_j^{\nf{1}{2}}$,
which together with $K^{(\nf{1}{2})}_{0}=0$ yields the second asserted bound.
The remainder of the proof
follows from Lemma~\ref{le-L} and the bound on $K^{(1)}_j$
we have just obtained. 
\end{proof}
%
%

\begin{thm}\label{thm-ren-mass}
There exist $\const,\ee_0>0$ such that the following assertions hold,
for all $\ee\in(0,\ee_0]$: 
We have
$E_\infty=\lim_{j\to \infty }E_j$ in the norm of $C_b^2(\cB_{\pmax})$
and $\|E_\infty-E_0\|_{C_b^2(\cB_{\pmax})}\le \const\,\ee$. More precisely,
we have the following estimates,
\begin{align}\label{ConvergenceRate1}
|E_{j}-E_\infty|&\le\const\, \ee\,\rho_j\,,
\\ \label{ConvergenceRate2}
 |\partial_{\V{h}} (E_{j}- E_\infty)|&\le\const\, \ee \, (1+\const\,\ee)^j\rho_j\,,
\\ \label{ConvergenceRate3}
 |\partial_{\V{h}}^2 (E_{j}- E_\infty)|&\le
\const\, \ee \,(1+\const\,\ee)^j\rho_j^{\nf{1}{2}},
\end{align}
uniformly on $\mathcal{B}_{\pmax}$ 
and for all $\V{h}\in\RR^3$, $|\V{h}|=1$, and $j\in\NN_0$.
In particular, $E_\infty $ is strictly convex on $\mathcal{B}_{\pmax}$
and attains its global minimum, $\inf_{\V{P}\in \RR^3} E_\infty$, at $\V{P}=\V{0}$.
\end{thm}

\begin{proof} 
For sufficiently small $\ee$, we may infer \eqref{ConvergenceRate1}
from \eqref{diff-E}; compare \eqref{sigrid2} below.
Recall the formulas in \eqref{for-1Deriv} and \eqref{for-2Deriv} 
for the first and second derivative of $E_k^j=E_k$.
When we represent $\partial_{\V{h}}^\nu E_{j+1}$, $\nu=1,2$, by means of these
formulas we actually 
replace all involved operators by unitarily equivalent ones
using the fact that
the trace and the Hilbert-Schmidt norm are invariant under conjugation with $U_j$.
%
For instance, $4\,\partial_{\V{h}}E_{j+1}$ is equal to $a_1$ with
$$
a_\nu:=
\tr\big[\check{\Pi}_{j+1}\,U_j\,\{\partial_{\V{h}}^\nu{\HAM}_{j+1}\}\,U_j^*\,\check{\Pi}_{j+1}\big]\,,
\quad\nu=1,2\,.
$$
In order to deal with operators defined on the same Hilbert space
we represent $\partial_{\V{h}}^\nu E_j$, $\nu=1,2$, in terms of
the Hamiltonian $\HAM_j^{j+1}$. 
In particular, we write $4\,\partial_{\V{h}}E_j=a_1'$ with
$$
a'_\nu:=\tr\big[{\Pi}_j^{j+1}\,\{\partial_{\V{h}}^\nu{\HAM}_j^{j+1}\}\,{\Pi}_j^{j+1}\big]\,,
\quad\nu=1,2\,.
$$
We plan to estimate
$|a_\nu-a_\nu'|\le|a_\nu-b_\nu|+|b_\nu-c_\nu|+|c_\nu-a_\nu'|$ with
\begin{align*}
b_\nu&:= \tr\big[\check{\Pi}_{j+1}\,U_j\,\{\partial_{\V{h}}^\nu {\HAM}_{j+1}\}\,U_j^*\,\Pi_j^{j+1}\big]\,,
\\
c_\nu&:= \tr\big[\check{\Pi}_{j+1}\,\{\partial_{\V{h}}^\nu{\HAM}_j^{j+1}\}\,\Pi_j^{j+1}\big]\,,
\end{align*}
for $\nu=1,2$.
Here, we have
$$
|a_\nu-b_\nu|\le\|\check{\Pi}_{j+1}\,U_j\{\partial_{\V{h}}^\nu{\HAM}_{j+1}\}U_j^*\|_{\HS}\,
\|\Pi_j^{j+1}-\check{\Pi}_{j+1}\|_{\HS}\,,
$$
and because of
$\|\check{\Pi}_{j+1}\,U_j\{\partial_{\V{h}}^\nu{\HAM}_{j+1}\}\,U_j^*\|^2_{\HS}
\le4\,\|\{\partial_{\V{h}}^\nu{\HAM}_{j+1}\,\}\,\Pi_{j+1}\|^2$
it follows from \eqref{sergio7} and Corollary~\ref{Cor-Karl}
that $|a_\nu-b_\nu|\le\const \,\ee\,(1+\const \,\ee)^j\rho_j$, $\nu=1,2$.
By \eqref{debora2} we further have
\begin{align*}
|b_\nu-c_\nu|\le& 
\big\|\big(U_j\{\partial_{\V{h}}^\nu{\HAM}_{j+1}\}U_j^*
-\partial_{\V{h}}^\nu{\HAM}_j^{j+1}\big)(\Hf^{(j,j+1)}\!+\rho_j)^{-1/2}\big\|\\
&\cdot \|(\Hf^{(j,j+1)}\!+\rho_j)^{1/2}\Pi_j^{j+1}\|_{\HS}
\cdot\|\check{\Pi}_{j+1}\|_{\HS}\le \const\,\ee\,\rho_j\,,
\end{align*}
since $(\Hf^{(j,j+1)}\!+\rho_j)^{1/2}\,\Pi_j^{j+1}=\rho_j^{1/2}\,\Pi_j^{j+1}$ 
by Theorem \ref{Thm:Gap}(3).
Again by \eqref{sergio7} and Corollary~\ref{Cor-Karl} we finally know that
\begin{align*}
|c_\nu-a'_\nu|\le& \|\{\partial_{\V{h}}^\nu{\HAM}_j^{j+1}\}\,\Pi_j^{j+1}\|_{\HS}\,
\|\Pi_j^{j+1}-\check{\Pi}_{j+1}\|_{\HS}
\le  \const \,\ee\,(1+\const \,\ee)^j\rho_j\,.
\end{align*}
Altogether this yields
\begin{equation}\label{sigrid1}
4\,|\partial_{\V{h}}(E_j-E_{j+1})|\stackrel{\textrm{if}\,\nu=1}{=}
|a_\nu-a_\nu'|\le\const\,\ee\,(1+\const \,\ee)^j\rho_j\,,
\end{equation}
which implies \eqref{ConvergenceRate2}, if $\ee$ is sufficiently small; 
compare \eqref{sigrid2} below.

Next, we turn to the second member in the formula \eqref{for-2Deriv} 
for the second derivatives.
We have
\begin{align}\label{sigrid1a}
4\,\partial_{\V{h}}^2(E_j-E_{j+1})&=a_2'-a_2+2\,\triangle\,,
\end{align}
with
$$
\triangle:=\|(\check{\RES}_{j+1}^\bot)^{\nf{1}{2}}
U_j\{\partial_{\V{h}}\HAM_{j+1}\}U_j^*\check{\Pi}_{j+1}\|_{\HS}^2
-\|((\RES_j^{j+1})^\bot)^{\nf{1}{2}}
\{\partial_{\V{h}}\HAM_j^{j+1}\}\Pi_j^{j+1}\|_{\HS}^2.
$$
Employing Lemma \ref{le-rek} and Corollary \ref{Cor-Karl} we infer that
\begin{align}\label{sigrid1b}
|\triangle|
&\le (K_{j+1}^{(\nf{1}{2})}+K_{j}^{(\nf{1}{2})})\,|K_{j+1}^{(\nf{1}{2})}-K_{j}^{(\nf{1}{2})}|
\le \const \,\ee^2 \,(1+\const\,\ee)^j\rho_j^{\nf{1}{2}}.
\end{align}
Combining \eqref{sigrid1}--\eqref{sigrid1b} 
and using $\rho_k=\rho_j\,(\nf{1}{2})^{k-j}$, $k\ge j\ge0$,
we get
\begin{align}\label{sigrid2}
\sum_{k=j}^\infty |\partial_{\V{h}}^2 (E_{k+1}- E_k)|
&\le\const\,\ee\sum_{k=j}^\infty(1+\const\,\ee)^k\rho_k^{\nf{1}{2}}
\le\frac{\const\,\ee}{1-b}(1+\const\,\ee)^j\rho_j^{\nf{1}{2}},
\end{align}
uniformly on $\mathcal{B}_{\pmax}$,
provided that $\ee>0$ is sufficiently small with
$b:=(1+\const\,\ee)\,(\nf{1}{2})^{\nf{1}{2}}<1$.
By the Weierstra{\ss} test this implies \eqref{ConvergenceRate3}.
Since $E_j$ and $E_\infty$ are rotationally symmetric this
also implies convergence in $C^2(\cB_\pmax)$.

To discuss the convexity of $E_\infty$ we
recall that $E_0(\V{P})= (\V{P}^2+1)^{\nf{1}{2}}$.
Since $\inf_{\mathcal{B}_{\pmax}}\partial_{\V{P}}^2E_0>0$,
we obtain $\inf_{\mathcal{B}_{\pmax}}\partial_{\V{P}}^2E_\infty>0$
from \eqref{ConvergenceRate3}, provided that $\ee$ is small enough. So
$E_\infty$ is strictly convex on $\mathcal{B}_{\pmax}$. By
rotational symmetry, $\nabla E_\infty(\V{0})=\V{0}$,
thus $E_\infty$ attains its unique minimum in $\mathcal{B}_{\pmax}$ 
at $\V{0}$. Thanks to Lemma \ref{EnergyLemma}(1) we know, however, that 
for small $\ee$, the global infimum of $E_\infty$ is located
in $\mathcal{B}_{\pmax}$, i.e. at $\V{P}=\V{0}$.
\end{proof}


\section{Existence and multiplicity of ground states}
\label{sec-ex}

\noindent
Let $P_{\Omega_j^\infty}$ be the projection onto the vacuum sector in $\sF_j^\infty$.
On $\HR$ we define 
\begin{equation*}
\wt{\Pi}_j^\infty(\V{P}) 
:= \wt{\Pi}_j(\V{P})\otimes P_{\Omega_j^\infty}\,,
\quad\textrm{with}\quad
\wt{\Pi}_j(\V{P}):= W_j(\V{P})\,\Pi_j(\V{P})\,W_j^*(\V{P})\,.
\end{equation*}
Here $j\in\NN_0$ and the unitaries $W_j$ are
given by \eqref{willi1}. Then
\begin{equation*}
\wt{\Pi}_{j+1}^\infty-\wt{\Pi}_j^\infty= 
W_{j-1} (\check{\Pi}_{j+1}-\Pi_j^{j+1})W_{j-1}^*\otimes P_{\Omega_{j+1}^\infty}
\,,\quad j\in\NN_0\,.
\end{equation*}
Hence, by Corollary~\ref{Cor-Karl} we may define the rank four projection
\begin{equation}\label{flann1}
\wt{\Pi}_\infty(\V{P}):= \lim_{j\to\infty}\wt{\Pi}_{j}^\infty(\V{P})\,,
\end{equation}
if $\ee$ is sufficiently small.
The goal of this section is to show that
$\wt{\Pi}_\infty$ is the ground state eigenprojection
of the operator $\wt{\HAM}_\infty$ introduced in
\eqref{def-wtHinfty}--\eqref{willi22} and discussed
in Lemma~\ref{le-wtHAM}. This will prove the first assertion 
in Theorem~\ref{thm-intro2}(1).
From \eqref{block} and Corollary~\ref{Cor-Karl} we then also
infer that the rate of convergence asserted in 
Theorem~\ref{thm-intro2}(1) is correct.

In order to show that the range of $\wt{\Pi}_\infty$
is {\em the whole} eigenspace of $\wt{\HAM}_\infty$
corresponding to $E_\infty$ we need
the next proposition. Its proof is based on
the following consequence of
\eqref{EnergyDiff} and \eqref{adam33},
\begin{equation}\label{adam-infty}
\const_1\,(\wt{\HAM}_\infty-E_\infty+\ee\,\rho_j)\ge \wt{H}_j^\infty-E_j\,.
\end{equation}

\begin{prop}\label{prop-horst}
Let $\V{P}\in\cB_\pmax$ and suppose that
$\phi$ is a normalized ground state eigenvector
of $\wt{\HAM}_\infty(\V{P})$. If $\ee>0$ is sufficiently small
depending only on $\pmax$, then 
\begin{equation}\label{stefanie1}
\liminf_{j\to\infty}\|\wt{\Pi}_j^\infty\,\phi\|>0\,.
\end{equation}
\end{prop}

\begin{proof}
Defining 
$F_j(t):=(1+t\const_1\,\ee\,\rho_j)^{\nf{1}{2}}
(1+t\const_1(\wt{\HAM}_\infty-E_\infty+\ee\,\rho_j))^{-\nf{1}{2}}$, 
$t\ge0$, $j\in\NN$,
we observe that
\begin{align*}
1&=\lim_{t\to\infty}\|F_j(t)\,\phi\|
\le\lim_{t\to\infty}\|F_j(t)\,\id\otimes P_{\Omega_j^\infty}\,\phi\|
+\|\id\otimes P_{\Omega_j^\infty}^\bot\,\phi\|
\\
&\le\lim_{t\to\infty}\|F_j(t)\,\wt{\Pi}_j^\bot\otimes P_{\Omega_j^\infty}\,\phi\|
+\|\wt{\Pi}_j\otimes P_{\Omega_j^\infty}\,\phi\|
+\|\id\otimes P_{\Omega_j^\infty}^\bot\,\phi\|\,,
\end{align*}
for every $j\in\NN$. If $\ee>0$ is sufficiently small, then
\eqref{adam-infty} and the operator monotonicity of the inversion
$T\mapsto T^{-1}$ permit to get
\begin{align*}
\|F_j(t)\,\Pi_j^\bot\otimes P_{\Omega_j^\infty}\,\phi\|^2
&\le(1+t\const_1\,\ee\,\rho_j)\,
\big\|(1+t\,(\wt{\HAM}_j^\infty-E_j))^{-\nf{1}{2}}\,
\wt{\Pi}_j^\bot\otimes P_{\Omega_j^\infty}\,\phi\big\|^2
\\
&=(1+t\const_1\,\ee\,\rho_j)\,
\big\|\{(1+t\,(\wt{\HAM}_j-E_j))^{-\nf{1}{2}}
\wt{\Pi}_j^\bot\}\otimes P_{\Omega_j^\infty}\,\phi\big\|^2
\\
&\le\frac{1+t\,\const_1\,\ee\,\rho_j}{1+t\,\gap_j}
\le\frac{1+t\,\const_1\,\ee\,\rho_j}{1+t\,q\,\rho_j/2}
\:\xrightarrow{\;t\to\infty\;}\:{2\const_1\,\ee}/{q}\le\const_2<1\,,
\end{align*}
for all $j\in\NN$.
Since also $\lim_{j\to\infty}\|(\id\otimes P^\perp_{\Omega_j^{\infty}})\,\phi\|=0$
(by dominated convergence) this implies
$0<1-\sqrt{\const_2}
\le\liminf_{j\to\infty}\|\wt{\Pi}_j^\infty\,\phi\|$.
\end{proof}

\smallskip

\noindent
For later use we record the following observation.
In the case $\V{P}\not=\V{0}$ we shall use it to produce a contradiction
showing that $\phi$ as in the following statement cannot exist.

\begin{cor}\label{cor-horst}
Given $\pmax>0$ we find $\ee_0>0$ such that, if
$\V{P}\in\cB_\pmax$, $\ee\in(0,\ee_0]$, and {\em if}
$\phi$ is a normalized ground state eigenvector
of ${\HAM}_\infty(\V{P})$, then 
$\{{\Pi}_j\otimes P_{\Omega_j^\infty}\,\phi\}_{j\in\NN}$ contains
a subsequence with a non-zero
weak limit, $\phi'\not=0$.
\end{cor}

\begin{proof}
The bound \eqref{adam-infty} holds true also with $\wt{H}_\infty$
and $\wt{H}_j^\infty$ replaced
by $H_\infty$ and $H_j^\infty$, respectively.
Hence, by exactly the same proof as above (just drop the tildes)
we obtain 
$\liminf_{j\to\infty}\|{\Pi}_j\otimes P_{\Omega_j^\infty}\,\phi\|
\ge1-\sqrt{\const_2}$.

Now, the bounded sequence defined by
$\phi_j:={\Pi}_j\otimes P_{\Omega_j^\infty}\,\phi$, $j\in\NN$,
contains a weakly convergent subsequence, say 
$\phi_{\vk}'=\phi_{j_\vk}$, $\vk\in\NN$.
Denoting its weak limit by $\phi'$, we have 
$\SPn{\phi}{\phi'}=\lim_{\vk\to\infty}\SPn{\phi}{\phi_{\vk}'}\ge(1-\sqrt{{\const_2}})^2$,
thus $\phi'\not=0$.
\end{proof}

\begin{thm}[{\bf Ground states}]
\label{mainthm}
For every $\pmax>0$, there exists $\ee_0>0$ such that,
for all $\V{P}\in\cB_\pmax$, the ground state energy $E_\infty(\V{P})$
is an exactly four-fold degenerate eigenvalue of
$\wt{\HAM}_\infty(\V{P})$ and the corresponding eigenprojection
is given by $\wt{\Pi}_\infty(\V{P})$ (defined in \eqref{flann1}).
In particular, $E_\infty(\V{0})$ is a four-fold degenerate eigenvalue of
${\HAM}_\infty(\V{0})$.
\end{thm}

\begin{proof}
According to Lemma~\ref{le-lisa}(iii)
$\wt{\HAM}_j^\infty\to\wt{\HAM}_\infty$ in norm 
resolvent sense and
$E_\infty=\lim_{j\to\infty}E_j=\inf\Spec[\wt{\HAM}_\infty]$. 
Together with Lemma~\ref{le-gap1}(1) and \eqref{flann1} this implies
$\wt{\Pi}_\infty=\lim_{j\to\infty}(\wt{\HAM}_j^\infty-E_j+1)^{-1}\wt{\Pi}_j^\infty
=(\wt{\HAM}_\infty-E_\infty+1)^{-1}\,\wt{\Pi}_\infty$, 
which shows that $\Ran(\wt{\Pi}_\infty)\subset\dom(\wt{\HAM}_\infty)$ and 
$\wt{\HAM}_\infty\,\wt{\Pi}_\infty=E_\infty\,\wt{\Pi}_\infty$.

Suppose $\phi$ is some normalized ground state eigenvector
of $\wt{\HAM}_\infty$ contained in the range of $\wt{\Pi}_\infty^\bot$.
By Proposition~\ref{prop-horst} and \eqref{flann1} we then
obtain the contradiction 
$0=\SPn{\phi}{\wt{\Pi}_\infty\,\phi}
=\lim_{j\to\infty}\|\wt{\Pi}_j^\infty\,\phi\|^2>0$.
Therefore, $\wt{\Pi}_\infty=\id_{\{E_\infty\}}(\wt{\HAM}_\infty)$.
\end{proof}


\section{Absence of ground states at  non-zero momenta}
\label{sec-non-ex}

\noindent
While the Hamiltonians $\wt{\HAM}_\infty(\V{P})$ possess ground state
eigenvectors, for small $\ee$, the original Hamiltonians ${\HAM}_\infty(\V{P})$
do not, unless $\V{P}$ is equal to zero.
The latter assertion is proved in the present section.
The starting point is the following bound implied by
Appendix~\ref{app-a(k)} where we again use the definition
\eqref{def-wtR} of the resolvent $\wh{\cR}_j(\V{P},\V{k})$.
We also recall the notation
$$
(a_\lambda(\V{k})\,\psi)^{(n)}(\V{k}_1,\lambda_1,\dots,\V{k}_n,\lambda_n)
=
(n+1)^{\nf{1}{2}}\psi^{(n+1)}(\V{k},\lambda,\V{k}_1,\lambda_1,\dots,\V{k}_n,\lambda_n)\,,
$$
almost everywhere, for $n\in\NN_0$,
$\psi=(\psi^{(n)})_{n=0}^\infty\in\Fock{j}{}$,
and $a_\lambda(\V{k})\,\Omega_j=0$.

\begin{lem}\label{le-norah}
Let $\V{P}\in\RR^3$, $j\in\NN_0\cup\{\infty\}$, $\ee>0$, and let 
$\phi_j(\V{P})$ be
a normalized ground state eigenvector of $\HAM_j(\V{P})$.
Then, for almost every $\V{k}$, 
\begin{align}
&\big\|a_\lambda(\V{k})\,\phi_j(\V{P})
+\ee\,\wh{\cR}_j(\V{P},\V{k})
{\V{G}_j(\V{k},\lambda)}\cdot\nabla\HAM_j(\V{P})\,\phi_j(\V{P})\big\|
\le\const\,\ee\,\frac{\id_{\rho_j\le|\V{k}|<\kappa}}{|\V{k}|^{\nf{1}{2}}}.
\label{norah0}
\end{align}
\end{lem}

\begin{proof}
In Lemma~\ref{le-a(k)} below we derive a
representation formula
for $a_\lambda(\V{k})\,\phi_j$ which implies
the asserted estimate. (Use $|\V{k}|\wh{\cR}_j(\V{P},\V{k})=\bigO(1)$.)
\end{proof}

\smallskip

\noindent
Together with Corollary~\ref{Cor-Karl} this implies the
following analog of an estimate stated in 
\cite[Proposition~5.1]{ChenFroehlich2007}
(with an improved exponent on the RHS, in fact):

\begin{prop}\label{prop-rosa}
Let $\pmax>0$.
Then there exist $\const,\ee_0>0$ such that, for all
$\V{P}\in\cB_\pmax$, $\ee\in(0,\ee_0]$, $j\in\NN_0$, every 
normalized ground state eigenvector, $\phi_j(\V{P})$, of $\HAM_j(\V{P})$,
and almost every $\V{k}$, 
\begin{align}\label{rosa0}
\Big\|a_\lambda(\V{k})\,\phi_j(\V{P})
+\ee\,
\frac{{\V{G}_j(\V{k},\lambda)}\cdot\nabla E_j(\V{P})}{
|\V{k}|-\V{k}\cdot\nabla E_j(\V{P})}\,\phi_j(\V{P})\Big\|
\le\const\,\ee\,\frac{\id_{\rho_j\le|\V{k}|<\kappa}}{|\V{k}|^{\nf{1}{2}}}\,.
\end{align}
\end{prop}

\begin{proof}
On account of \eqref{ResRepl1} and Corollary~\ref{Cor-Karl} we have
\begin{align}\label{rosa1}
\|\wh{\cR}_j(\V{P},\V{k})\,\Pi_j^\bot\partial_{\V{h}}\HAM_j\,\phi_j\|
&\le\const\,\|\RES_j^\bot\,\partial_{\V{h}}\HAM_j\,\phi_j\|\le
\const\,K_j^{(1)}\le\const'.
\end{align}
(Here and below all derivatives are evaluated at $\V{P}$.)
On account of \eqref{norah0} and 
$|\V{G}_j(\V{k},\lambda)|\le\const\,
\id_{\rho_j\le|\V{k}|<\kappa}\,|\V{k}|^{-\nf{1}{2}}$
it thus remains to treat 
$\wh{\cR}_j(\V{P},\V{k})\,\Pi_j\,\nabla\HAM_j\,\phi_j
=\nabla E_j\,\wh{\cR}_j(\V{P},\V{k})\,\phi_j$.
For this we again write
$\RES_j(|\V{k}|)=(\HAM_j(\V{P})-E_j(\V{P})+|\V{k}|)^{-1}$.
Then the second resolvent formula implies
\begin{align}\nonumber
\wh{\cR}_j(\V{P},\V{k})\,\phi_j
&=\RES_j(|\V{k}|)\,\phi_j+\wh{\cR}_j(\V{P},\V{k})\,
\{\HAM_j(\V{P})-\HAM_j(\V{P}-\V{k})\}\,\RES_j(|\V{k}|)\,\phi_j
\\\label{rosa2}
&=|\V{k}|^{-1}\,\phi_j+|\V{k}|^{-1}\,\wh{\cR}_j(\V{P},\V{k})\,
\{{\V{k}}\cdot\nabla\HAM_j+\bigO(|\V{k}|^2)\}\,\phi_j\,.
\end{align}
Writing
$\nabla\HAM_j\,\phi_j=\nabla E_j\,\phi_j+\Pi_j^\bot\nabla\HAM_j\,\phi_j$
and solving \eqref{rosa2} for $\wh{\cR}_j(\V{P},\V{k})\,\phi_j$ 
yields
\begin{align*}
\wh{\cR}_j(\V{P},\V{k})\,\phi_j
=\frac{1}{|\V{k}|-{\V{k}}\cdot\nabla E_j}
\big(\phi_j+\wh{\cR}_j(\V{P},\V{k})\,
\{\Pi_j^\bot{\V{k}}\cdot\nabla\HAM_j+\bigO(|\V{k}|^2)\}\,\phi_j\big)\,.
\end{align*}
Here $\wh{\cR}_j(\V{P},\V{k})\,
\{\Pi_j^\bot{\V{k}}\cdot\nabla\HAM_j+\bigO(|\V{k}|^2)\}\,\phi_j=\bigO(|\V{k}|)$
by \eqref{rosa1}.
\end{proof}

\smallskip

\noindent
Recall that $\nabla E_j(\V{0})=\V{0}$, for all $j\in\NN_0$, so that
the coherent factor in the formula \eqref{rosa0} vanishes at $\V{P}=\V{0}$.
As an immediate corollary we observe that the expectation value
of the number operator, $N:=d\Gamma(1)$, in a state belonging to the range of
$\Pi_\infty(\V{0}):=\wt{\Pi}_\infty(\V{0})=\id_{\{E_\infty(\V{0})\}}(\HAM_\infty(\V{0}))$
is finite. Recall also $\wt{\Pi}_j^\infty(\V{0})\to\Pi_\infty(\V{0})$ in norm,
as $j\to\infty$.

\begin{cor}\label{cor-N}
There exist $\const,\ee_0>0$ such that
$\Ran(\Pi_\infty(\V{0}))\subset\dom(N^{\nf{1}{2}})$
and $\|N^{\nf{1}{2}}\,\Pi_\infty(\V{0})\|\le\const\,\ee$, 
for all $\ee\in(0,\ee_0]$.
\end{cor}

\begin{proof}
Let $\phi\in\dom(N^{\nf{1}{2}})$ and $\psi\in\HR$ both be normalized.
Then, by \eqref{rosa0},
\begin{align*}
|\SPn{N^{\nf{1}{2}}&\phi}{\Pi_\infty(\V{0})\,\psi}|^2\le
\big|\lim_{j\to\infty}\SPn{\phi}{N^{\nf{1}{2}}\,\wt{\Pi}_j^\infty(\V{0})\,\psi}\big|^2
\\
&\le\sup_{j\in\NN}\sum_{\lambda\in\ZZ_2}\int_{\RR^3}
\big\|a_\lambda(\V{k})\,\wt{\Pi}_j^\infty(\V{0})\,\psi\big\|^2d^3\V{k}
\le\const\,\ee\int_{|\V{k}|<\kappa}
\frac{d^3\V{k}}{|\V{k}|}\,,
\end{align*}
which implies the assertion, since $N^{\nf{1}{2}}$ is self-adjoint.
\end{proof}

\smallskip

\noindent
In an essentially traditional fashion we next infer the
absence of ground states of $\HAM_\infty(\V{P})$, $\V{P}\not=\V{0}$, 
from Proposition~\ref{prop-rosa}; compare, e.g.,
\cite{HaslerHerbst2008,Schroer1963}.

\begin{thm}\label{thm-absence}
Given $\pmax>0$ we find $\ee_0>0$ such that, for all
$\V{P}\in\cB_\pmax\setminus\{\V{0}\}$, and $\ee\in(0,\ee_0]$, 
the ground state energy
$E_\infty(\V{P})$ is not an eigenvalue of $\HAM_\infty(\V{P})$.
\end{thm}

\begin{proof}
We write $H=H_\infty$, $E=E_\infty$.
Suppose that $\phi\in\dom(\HAM(\V{P}))$
is normalized and
$\HAM(\V{P})\,\phi=E(\V{P})\,\phi$. To get a contradiction we exploit that
$\nabla E(\V{P})\not=\V{0}$ which follows from the strict convexity of $E$
on $\cB_{\pmax}$ and $\nabla E(\V{0})=\V{0}$.

Let $\phi_j:=\Pi_j\otimes P_{\Omega_j^\infty}\,\phi$. 
Borrowing an idea from \cite{HaslerHerbst2008} we pick some
$\eta\in\dom(N^{\nf{1}{2}})$, $\|\eta\|=1$.
Furthermore, let
$g_r:=\id_{\{r\le|\V{k}|<\kappa\}}\V{G}\cdot\nabla E/(|\V{k}|-\V{k}\cdot\nabla E)$, 
for $r\in(0,\kappa)$.
Then $\|g_r\|^2=\const_1\,\ln(\kappa/r)$ with some
$\const_1\in(0,\infty)$ because $0<|\nabla E|<1$.
Finally, let 
$g^{(j)}_{\rho_j}:=\id_{\{\rho_j\le|\V{k}|<\kappa\}}\V{G}\cdot\nabla E_j/(|\V{k}|-\V{k}\cdot\nabla E_j)$. 
By virtue of \eqref{rosa0} we then have
\begin{align}\nonumber
\|g_r\|\,&\big\|(N+1)^{\nf{1}{2}}\eta\big\|
\\\nonumber
&\ge\big|\SPb{\ad_\lambda(g_r)\,\eta}{\phi_j}\big|
=\Big|\sum_{\lambda\in\ZZ_2}\int_{\RR^3}{g}_r(\V{k},\lambda)\,
\SPb{\eta}{a_\lambda(\V{k})\,\phi_j}\,d^3\V{k}\Big|
\\\nonumber
&\ge\ee\,|\SPn{g_r}{g^{(j)}_{\rho_j}}\SPn{\eta}{\phi_j}|
-\const\,\ee
\sum_{\lambda\in\ZZ_2}\int_{r\le|\V{k}|<\kappa}
\frac{|\veps_\lambda(\mr{\V{k}})\cdot\nabla E|}{1-\mr{\V{k}}\cdot\nabla E}
\frac{d^3\V{k}}{|\V{k}|^2}
\,,
\end{align}
for every $j\in\NN$. By Corollary~\ref{cor-horst}
$\{\phi_j\}_{j\in\NN}$ converges weakly to some non-zero vector $\phi'$
along some subsequence. We fix $\eta$ such that $\SPn{\eta}{\phi'}>0$.
Since also $\nabla E_j\to\nabla E$,  
thus $\SPn{g_r}{g^{(j)}_{\rho_j}}\to\|g_r\|^2$, $j\to\infty$, we arrive at
$\const_1^{\nf{1}{2}}\,\ln(\kappa/r)^{\nf{1}{2}}
\|(N+1)^{\nf{1}{2}}\eta\|\ge
\ee\,\const_1\,\ln(\kappa/r)\,\SPn{\eta}{\phi'}
-\const'\ee\,(\kappa-r)$.
For sufficiently small $r\in(0,\kappa)$,
this gives a contradiction.
\end{proof}


\section{Coherent infra-red representations}\label{sec-CIRR}

\noindent
In this final section we discuss certain representations
of the Weyl algebra associated with the functionals
$$
\sigma_{\V{P}}(A):=\tr[\wt{\Pi}_{\infty}(\V{P})\,A]\,,
\quad A\in\cB(\HR)\,,\;\V{P}\in\cB_\pmax\,.
$$
For the most part we proceed along the lines of \cite{Froehlich1973}.
We shall also recall various arguments from \cite{Streit1967}.
There is, however, one simplification:
Thanks to the existence of the limit projection $\wt{\Pi}_{\infty}$
guaranteed by the IAPT we may define $\sigma_{\V{P}}$ directly
as the corresponding trace functional.
In \cite{ChenFroehlich2007,Froehlich1973} the analog of
$\sigma_{\V{P}}$ is defined via compactness arguments
and some additional arguments are required in order to study
its GNS representation. We also point out that the proof of the final
result of this section, Corollary~\ref{cor-disjoint} below, is based
on the strict convexity of $E_\infty$ (or rather uniform strict convexity of $E_j$, 
$j\gg1$).

In the rest of this section we fix $\pmax>0$ and we always assume
that $\ee>0$ is sufficiently small, depending only on $\pmax$ and $\kappa$.

Recall that $\wt{\Pi}_{\infty}$ has been obtained as a limit
of dressing transformed ground state projections and that the
sequence of dressing transformations $W_j$
has no limit {\em in the Banach space $\cB(\HR)$}.
We can, however, give a mathematical meaning to a `limit
transform' of sorts as follows:

We consider the complete infinite direct product space
(CDPS) \cite{vNeumann1939}
$$
\wh{\HR}:=\HR_0\otimes\prod\otimes_{j\in\NN_0}\sF_j^{j+1}
=\HR_k\otimes\prod\otimes_{j>k}\sF_j^{j+1}.
$$
The Hilbert space
$\wh{\HR}$ is non-separable and contains the coherent states
$$
\Omega^{\ren}_{\V{P}}:=
\eta\otimes\prod\otimes_{j\in\NN_0}\Omega_{j,\V{P}}^{j+1}\,,
\quad\Omega_{j,\V{P}}^{j+1}
:=U_j^*(\V{P})\,\Omega_j^{j+1},\quad\V{P}\in\cB_\pmax\,.
$$
Here $\eta$ is some fixed normalized element of the vacuum sector in $\HR_0$.
Recall from \cite{vNeumann1939}
that two vectors $\Phi=\phi_0\otimes\phi_1\otimes\dots$
and $\Psi=\psi_0\otimes\psi_1\otimes\dots$ in $\wh{\HR}$
are called equivalent, $\Phi\approx\Psi$,
iff $\sum_{j=0}^\infty|\SPn{\phi_j}{\psi_j}-1|<\infty$.
Moreover, $\Phi$ is called a $C_0$-vector iff
$\sum_{j=0}^\infty\big|\|\phi_j\|-1\big|<\infty$ and,
if $\Phi$ and $\Psi$ are both $C_0$-vectors, they are called
weakly equivalent, $\Phi\underset{w}{\approx}\Psi$, 
iff $\sum_{j=0}^\infty\big||\SPn{\phi_j}{\psi_j}|-1\big|<\infty$.
The separable Hilbert subspace of $\wh{\HR}$ generated by the equivalence
class of some $C_0$-vector $\Phi$ is denoted
by $\HR_0\otimes\prod\otimes_{j\in\NN_0}^\Phi\sF_{j}^{j+1}$.
It is called an incomplete direct
product space (IDPS). Although we are now dealing with four-spinors
we denote the IDPS corresponding to the coherent states defined above 
as in the introduction by
$$
\HR_{\V{P}}^{\ren}:=
\HR_0\otimes\prod\otimes_{j\in\NN_0}^{\Omega^{\ren}_{\V{P}}}\sF_{j}^{j+1}.
$$
Now, we can, loosely speaking, `invert' the `limit dressing transformation' 
by regarding its `inverse' as a unitary map
from $\HR$ to $\HR^{\ren}_{\V{P}}$.
More precisely, denoting elements of $\HP_j^{j+1}$ by $h^{(j)}$ etc.,
we set
$$
b_{\V{P}}^\sharp(h^{(j)}):=U_j(\V{P})^* a^\sharp(h^{(j)})\,U_j(\V{P})
=a^\sharp(h^{(j)})+\ee\SPn{h^{(j)}}{f_j(\V{P})}^\sharp,
\quad j\in\NN_0\,,
$$
and define a unitary map 
$W_{\V{P}}^*:\sH=\HR_0\otimes\Fock{0}{\infty}\to\HR^{\ren}_{\V{P}}$ first by
the following formula and then by linear and isometric extension:
\begin{align*}
{W}_{\V{P}}^*\,
\zeta\otimes\prod_{j=0}^N\prod_{\vk=1}^{n_j}\ad(h^{(j)}_{\vk})\,\Omega
&:=\zeta\otimes\bigotimes_{j=0}^N
\Big(\prod_{\vk=1}^{n_j}b_{\V{P}}^\dagger(h^{(j)}_{\vk})
\,\Omega_{j,\V{P}}^{j+1}\Big)
\otimes\prod\otimes_{j>N}\Omega_{j,\V{P}}^{j+1}\,.
\end{align*}
Here $n_j$ may be equal to zero in which case $\prod_{\vk=1}^0\ad:=\id$, etc.   
The isometry is clear from the definition of the scalar product on $\wh{\HR}$. 
After isometric extension ${W}_{\V{P}}^*$ is
surjective because the set of vectors appearing on the RHS
of the previous definition is total in $\HR^{\ren}_{\V{P}}$, when
the $h^{(j)}_{\vk}$ are chosen from some orthonormal basis
of $\HP_j^{j+1}$, for every $j\in\NN_0$. In particular,
$W_{\V{0}}^*$ is just the natural identification
$\HR\equiv\HR^{\ren}_{\V{0}}$.

Let $\fA_j$, $\fA^\circ$, $\fA$, $\fW_j$, $\fW^\circ$, and $\fW$
be the $*$-algebras defined in
and below \eqref{def-fA} with $\CC^2$ replaced by $\CC^4$.
There is an isometry, $\pi:\,\fA^\circ\to\cB(\wh{\HR})$, with
\begin{align*}
\pi(A){\prod}\otimes_k\psi_k
=\prod\otimes_k\psi_k'\,,\qquad
\psi_j':=A'\psi_j\,,\quad\psi_k'=\psi_k\,,\;k\not=j\,,
\end{align*}
where $A=\id\otimes A'\otimes\id\in\fA_j$ and 
${\prod}\otimes_k\psi_k\in\wh{\HR}$.
By its isometry $\pi$ extends to all of $\fA$.
Then $\pi(\fW_j)''=\pi(\fA_j)''$ and, hence,
$\pi(\fA^\circ)$ and $\pi(\fW^\circ)$ generate the same von Neumann algebra
$$
B^\#:=\pi(\fA^\circ)''=\pi(\fW^\circ)''
\subsetneq\cB(\wh{\HR})\,,\qquad B^\#\supset\pi(\fA)\,.
$$
The spaces $\HR_{\V{P}}^\ren$ are reducing subspaces
for every element of $B^\#$
\cite[Theorem~X(I)]{vNeumann1939}.
We denote restriction to $\HR_{\V{P}}^\ren$ by a
subscript $\V{P}$ and set $\pi_{\V{P}}(A):=\pi(A)_{\V{P}}$, $A\in\fA$.
Then we have $\pi_{\V{P}}(\fW_j)''=\pi_{\V{P}}(\fA_j)''$ and, hence,
\begin{equation*}
\cB(\HR_{\V{P}}^\ren)=\big\{Q_{\V{P}}:\,Q\in B^\#\big\}
=\pi_{\V{P}}(\fA^\circ)''
=\pi_{\V{P}}(\fW^\circ)''.
\end{equation*}
The three identities above follow from
\cite[Theorem~X(II)]{vNeumann1939}, \cite[Part~II, Prop.~II.12]{GWT1969},
\cite[Lemma~2]{Streit1967}, respectively. In particular,
it follows (as in \cite{Streit1967}) that 
the representation $\pi_{\V{P}}:\fW\to\cB(\HR_{\V{P}}^\ren)$ is irreducible.

Next, we define functionals on $\cB(\HR^{\ren}_{\V{P}})$ by
$$
\omega_{\V{P}}^\ren:=\sigma_{\V{P}}\circ\alpha_{\V{P}}\,,
\qquad\alpha_{\V{P}}(T):=W_{\V{P}}\,T\,W_{\V{P}}^*\,,
\quad T\in\cB(\HR^{\ren}_{\V{P}})\,. 
$$
It is then a direct consequence of the various definitions recalled above
that, for every local observable $A\in\fA^\circ$, we have
$\alpha_{\V{P}}(\pi_{\V{P}}(A))=W_\ell\,A\,W_\ell^*$, with some
$\ell\in\NN$. Therefore,
\begin{align*}
\omega_{\V{P}}^\ren\big(\pi_{\V{P}}(A)\big)
&=\tr\big[\wt{\Pi}_\infty\,W_\ell\,A\,W_\ell^*\big]
=\lim_{j\to\infty}\tr\big[W_j\,\Pi_j^\infty\,W_j^*\,W_\ell\,A\,W_\ell^*\big]
\\
&=\lim_{j\to\infty}\tr[\Pi_j^\infty A]=:\omega_\V{P}(A)\,,\quad
A\in\fA^\circ\,,\;\;\ell\;\textrm{suitable}\,,
\end{align*}
because $W_\ell\,A\,W_\ell^*=W_j\,A\,W_j^*$, if $j\ge\ell$ and $\ell$ is large enough.
This shows that the limits in the second line
define a continuous linear functional on $\fA^\circ$
whose unique continuous extension to $\fA$ --
henceforth again denoted by the same symbol -- is given by
$\omega_{\V{P}}(A)=\omega_{\V{P}}^\ren(\pi_{\V{P}}(A))$, $A\in\fA$.

Given that $E$ is strictly convex the following corollary
is folkloric; see \cite{Froehlich1973,KMcKW1966,Streit1967}. 
We include its proof for convenience of the reader.

\begin{cor}\label{cor-disjoint}
Let $\V{P},\V{Q}\in\cB_\pmax$, $\V{P}\not=\V{Q}$.
Then $\Omega_{\V{P}}^{\ren}$ and $\Omega_{\V{Q}}^{\ren}$
are not weakly equivalent. Therefore,
$\HR_{\V{P}}^{{\ren}}\bot\,\HR_{\V{Q}}^{\ren}$ (as subspaces of $\wh{\HR}$),
and the representations $\pi_{\V{P}}$ and $\pi_{\V{Q}}$
of $\fW$ are unitarily inequivalent and, thus, 
disjoint as they are irreducible.
\end{cor}

\begin{proof}
We have
$\Omega_{\V{P}}^{\ren}\underset{w}{\approx}\Omega_{\V{Q}}^{\ren}$
if and only if $\sum_{j=0}^\infty||\alpha_j|-1|<\infty$, where
$$
\alpha_j:=\SPn{U_j^*(\V{P})\,\Omega_j^{j+1}}{U_j^*(\V{Q})\,\Omega_j^{j+1}}
=\SPn{\Omega_j^{j+1}}{e^{i\ee\vo(h_{\V{P},\V{Q}}^{(j)})}\,\Omega_j^{j+1}}
=e^{-\frac{\ee^2}{2}\|h_{\V{P},\V{Q}}^{(j)}\|^2},
$$
with 
$h_{\V{P},\V{Q}}^{(j)}(\V{k},\lambda)
:=f_j(\V{P};\V{k},\lambda)-f_j(\V{Q};\V{k},\lambda)$.
By strict convexity of $E_\infty$ on $\cB_\pmax$ we have
$|\nabla E_j(\V{P})-\nabla E_j(\V{Q})|\ge c>0$, for all sufficiently
large $j\in\NN$. Hence, it is elementary to check that
$\sum_{j=0}^\infty\|h_{\V{P},\V{Q}}^{(j)}\|^2=\infty$.
Therefore,
$\sum_{j=0}^\infty||\alpha_j|-1|=\infty$, as we easily see.
Now, it is known that
two representations of the Weyl algebra in 
different IDPS are unitarily equivalent
if and only if the IDPS belong to weakly equivalent
reference vectors \cite{KMcKW1966,Streit1967}.
Let us
recall the brief argument for the ``only if'' part given
in \cite{Streit1967} in our special situation:
Assume there is a unitary
$T\in\cB(\HR^{\ren}_{\V{P}},\HR^{\ren}_{\V{Q}})$ such that
$$
\SPn{T\,\Psi}{\pi(W)\,T\,\Psi}=\SPn{T\,\Psi}{\pi_{\V{Q}}(W)\,T\,\Psi}
=\SPn{\Psi}{\pi_{\V{P}}(W)\,\Psi}
=\SPn{\Psi}{\pi(W)\,\Psi},
$$ 
for all $W\in\fW$ and $\Psi\in\HR^{\ren}_{\V{P}}$.
Then $\SPn{T\,\Psi}{Q\,T\,\Psi}=\SPn{\Psi}{Q\,\Psi}$, for all
$Q\in\pi(\fW^\circ)''=B^\#$ and $\Psi\in\HR^{\ren}_{\V{P}}$.
Since $\Omega^{\ren}_{\V{P}}\not\underset{w}{\approx}\Omega^{\ren}_{\V{Q}}$,
there exists, however, some
$Q\in B^\#$ which is reduced by both $\HR^{\ren}_{\V{P}}$
and $\HR^{\ren}_{\V{Q}}$ such that $Q_{\V{P}}=\id$ and $Q_{\V{Q}}=0$
\cite[Theorem~X(II)]{vNeumann1939}; a contradiction!
\end{proof}

\begin{rem}\label{rem-Omega-ren}
One might think it is more natural to determine the spaces
$\HR_{\V{P}}^\ren$ by slightly different coherent states, namely
$\wt{\Omega}^{\ren}_{\V{P}}:=
\eta\otimes\prod\otimes_{j\in\NN_0}\wt{\Omega}_{j,\V{P}}^{j+1}$
with
$\wt{\Omega}_{j,\V{P}}^{j+1}:=\wt{U}_j^*(\V{P})\,\Omega_j^{j+1}$,
where $\wt{U}_j^*(\V{P})$ is defined by \eqref{def-U-intro} and \eqref{def-ff}
with $E_j$ in \eqref{def-ff} replaced by $E=E_\infty$.
Using $|\nabla E_j-\nabla E|\le\const\,\ee\,(1+\const\,\ee)^j\rho_j$,
it is, however, straightforward to verify that
${\Omega}^{\ren}_{\V{P}}$ and $\wt{\Omega}^{\ren}_{\V{P}}$
are equivalent and, hence, give rise to the same IDPS.
\hfill$\Diamond$
\end{rem}



\begin{appendix}

\section{Some properties of the fiber Dirac operator}\label{app-esa-D}

\noindent
Below we prove two lemmata
on the Dirac operators defined as the closures
of the operators \eqref{def-Dkj} which have
been used several times in the main text.
We recall that the dense subspace $\sC_j\subset\HR_j$
has been defined below \eqref{def-Dkj}.
In the whole Appendix~\ref{app-esa-D} the vector potential
$\V{A}_k^j$ satisfies the conditions in Assumption~\ref{ariel}.
The values of $\ee,\kappa>0$ are arbitrary.

\begin{lem}\label{le-esa-wD}
$\D_k^j(\V{P})$ and $\D_k^j(\V{P})^2$
are essentially self-adjoint on $\sC_j$,
for all $\V{P}\in\RR^3$ and $k,j\in\NN_0\cup\{\infty\}$, $k\le j$.
\end{lem}

\begin{proof}
Assume that 
$\psi=(\psi^{(0)},\ldots,\psi^{(m)},0,0,\ldots\,)\in\sC_j$
and that 
$\psi^{(n)}\in\CC^4\otimes[\sF_j]^{(n)}\cong L^2_{\mathrm{sym}}((\RR^3\times\ZZ_2)^n,\CC^4)$ 
is supported in
$([-R,R]^{3}\times\ZZ_2)^n$, for every $n=1,\ldots,m$, and some
$m\in\NN$ and $R>\kappa$.
Let $N:=d\Gamma(1)$ denote the number operator.
A trivial estimate using
$\gamma_j:=\|\valpha\cdot\V{A}_j\,(N+1)^{-1/2}\|<\infty$
shows that
\begin{equation}\label{vera1}
\|\D_k^j(\V{P})\,\psi\|\le C(j,R,\V{P})(m+1)\,\|\psi\|\,,
\quad C(j,R,\V{P}):=|\V{P}|+R+\gamma_j\,.
\end{equation}
Furthermore, we observe that $(\D_k^j(\V{P})\,\psi)^{(n)}$
is still supported in $([-R,R]^{3}\times\ZZ_2)^n$, for every $n\in\NN$,
and $(\D_k^j(\V{P})\,\psi)^{(n)}=0$,
for $n=m+2,m+3,\ldots\;$. 
Applying \eqref{vera1} repeatedly we conclude,
for $\vp\in\sC_j$, such that
$\vp^{(n)}$ is supported in $([-R,R]^{3}\times\ZZ_2)^n$, for every $n\in\NN$,
and $\vp^{(n)}=0$, for $n=m+1,m+2,\ldots\;$,
$$
\big\|\D_k^j(\V{P})^{\nu\ell}\,\vp\big\|\le
C(j,R,\V{P})^{\nu\ell}\,\frac{(m+\nu\ell)!}{m!}\,\|\vp\|\,,\quad\nu\in\{1,2\}\,.
$$
Since $\frac{(m+\nu\ell)!}{(\nu\ell)!\,m!}={m+\nu\ell\choose m}<2^{m+\nu\ell}$
this implies
$$
\sum_{\ell=0}^\infty \frac{t^{\nu\ell}}{(\nu\ell)!}\,
\big\|\D_k^j(\V{P})^{\nu\ell}\,\vp\big\|
<\infty\,,\quad\nu\in\{1,2\}\,,
$$
for all $0<t<2^{-\nu}/C(j,R,\V{P})$.
Consequently, every $\vp\in\sC_j$
is an analytic vector for $\D_k^j(\V{P})$ and a
semi-analytic vector for $\D_k^j(\V{P})^2$, which
implies the assertion;
see \cite[Theorems~X.39 and~X.40]{ReedSimonII}.
\end{proof}

\smallskip

\noindent
The following lemma 
on the resolvent $\R_k^j(iy)\equiv\R_k^j(\V{P},iy)$
and the sign function $S_k^j\equiv S_k^j(\V{P})$
of the fiber Dirac operator defined in \eqref{def-Rkj}
and \eqref{for-sgnT}, respectively,
is a variant of \cite[Corollary~3.1 \& Lemma~3.3]{MatteStockmeyer2009a}.
We present its proof for the sake of completeness.

\begin{lem}\label{le-dominique}
For $\nu>0$, we find $\rho\in(0,\infty)$
and $Y_{\nu}(y),\Upsilon_\nu\in\sL(\HR_j)$
with
\begin{align}\label{zita111}
\R_k^j(iy)\,(\Hf^{(j)}+\rho)^{-\nu}&=
(\Hf^{(j)}+\rho)^{-\nu}\,\R_k^j(iy)\,Y_{\nu}(y)\,,\quad
\|{Y}_{\nu}(y)\|\le2\,,
\\\label{zita555}
S_k^j\,(\Hf^{(j)}+\rho)^{-\nu}&=
(\Hf^{(j)}+\rho)^{-\nu}\,(S_k^j+\Upsilon_{\nu})\,,\qquad
\|\Upsilon_{\nu}\|\le1\,,
\end{align}
for $y\in\RR$.
In particular,
$\R_k^j(iy)$ and $S_k^j$ map $\dom((\Hf^{(j)})^\nu)$ into itself.
Furthermore, the following resolvent identity
is valid, for all $\lambda>0$,
\begin{align*}
\big(\R_k^j(iy)-\R_j(iy)\big)
(\Hf^{(k,j)}\!+\lambda)^{-\nf{1}{2}}
&=\R_j(iy)\,\ee\,
\valpha\cdot\V{A}_k^j\,(\Hf^{(k,j)}\!+\lambda)^{-\nf{1}{2}}
\R_k^j(iy)\,,
\end{align*}
and the bound \eqref{roswita} holds true
(with the same constant as in \eqref{akj-neu}).
\end{lem}

\begin{proof}
We put $\Theta:=\Hf^{(j)}+\rho$ and recall from
\cite[Lemma~3.2]{Matte2009}\&\cite[Lemma~3.1]{MatteStockmeyer2009a} 
that, for every $\nu\in\RR$,
$[\valpha\cdot\V{A}_k^j,\Theta^{-\nu}]\,\Theta^\nu$,
defined a priori on $\sC_j$,
extends to a bounded operator on $\sH_j$,
henceforth denoted by $T_\nu$, and
$\|T_\nu\|\le C(\nu,\kappa)/\rho^{\nf{1}{2}}$, for all $\nu>0$.
We choose $\rho$ so large that $\|T_\nu\|\le1/2$.
For $\vp\in\sC_j$,
\begin{align*}
\R_k^j(iy)\,\Theta^{-\nu}(\D_k^j-iy)\,\vp
&=
\R_k^j(\imath y)\,(\D_k^j-\imath y)\,\Theta^{-\nu}\vp
+\R_k^j(\imath y)\,[\Theta^{-\nu},\valpha\cdot\V{A}_k^j]\,\vp
\\
&=
\big(\id-\R_k^j(\imath y)\,T_\nu\big)
\Theta^{-\nu}\R_k^j(\imath y)\,(\D_k^j-\imath y)\,\vp\,.
\end{align*}
Since $\D_k^j$ is essentially self-adjoint
on $\sC_j$ we know that
$(\D_k^j-\imath y)\,\sC_j$ is dense in $\sH_j$,
whence
$$
\R_k^j(\imath y)\,\Theta^{-\nu}=\big(\id-\R_k^j(\imath y)\,T_\nu\big)
\Theta^{-\nu}\R_k^j(\imath y)\,.
$$
Since $\|\R_k^j(\imath y)\,T_\nu\|\le1/2\Jy$ we may define
$X_\nu(y):=(\id-\R_k^j(\imath y)\,T_\nu)^{-1}$, so that
$X_\nu(y)\,\R_k^j(\imath y)\,\Theta^{-\nu}=\Theta^{-\nu}\R_k^j(\imath y)$.
Thus, we obtain \eqref{zita111} with
$Y_\nu(y):=X_\nu(y)^*=\id+Z_\nu(y)$, where
$\|Z_\nu(y)\|\le1/\Jy$. Computing the strongly convergent 
principal value of \eqref{zita111} 
and using \eqref{for-sgnT} we also obtain \eqref{zita555}. 

Next, we verify the asserted resolvent identity first on the 
domain of $\Hf^{(j)}$ by applying $\D_j-iy$ on both sides
and using that $\R_k^j(iy)$ and $(\Hf^{(k,j)}\!+\lambda)^{-\nf{1}{2}}$
commute.
After that it extends by continuity to an identity in $\cB(\HR_j)$,
which together with \eqref{akj-neu} implies \eqref{roswita}.
\end{proof}


\section{Infra-red behavior of ground states}\label{app-a(k)}

\noindent
In the next formula we employ the notation $\wh{\cR}_j(\V{P},\V{k})$
introduced in \eqref{def-wtR}.

\begin{lem}\label{le-a(k)}
Let $\V{P}\in\RR^3$, $j\in\NN\cup\{\infty\}$, and $\ee>0$,
and assume that $\phi_j\equiv \phi_j(\V{P})\in\dom(\HAM_j)$ is a ground state
eigenvector of $\HAM_j\equiv\HAM_j(\V{P})$. 
Then the following representation is valid,
for $\lambda\in\ZZ_2$ and almost every $\V{k}\in\cA_j$,
\begin{align}\label{for-a(k)}
&a_\lambda(\V{k})\,\phi_j
+\ee\,{\V{G}_j(\V{k},\lambda)}\cdot
\wh{\cR}_j(\V{P},\V{k})\,\nabla\HAM_j\,\phi_j
\\\nonumber
&=-\ee\!\int_\RR
\wh{\cR}_j(\V{P},\V{k})\R_{j}(\V{P}-\V{k},iy)\,\valpha\cdot\V{k}\,\R_j(\V{P},iy)\,
\valpha\cdot\V{G}(\V{k},\lambda)\R_j(\V{P},iy)\,\phi_j\,\frac{y\,dy}{i\pi}.
\end{align}
\end{lem}

\begin{proof}
Let $f\in C_0^\infty(\cA_j)$, $\V{p}\in\RR^3\setminus\{\V{0}\}$, 
and $\lambda\in\ZZ_2$.
We fix $\V{P}\in\RR^3$ and the scale parameter $j\in\NN\cup\{\infty\}$,
and abbreviate $\D:=\D_j(\V{P})$, $\D_{\V{p}}:=\D_j(\V{P}-\V{p})$,
$\R(iy):=(\D-iy)^{-1}$, $\R_{\V{p}}(iy):=(\D_{\V{p}}-iy)^{-1}$,
and $\V{G}_\lambda:=\V{G}_j(\cdot,\lambda)$.
We also write $a_\lambda^\sharp(f):=a^\sharp(\tilde{f})$,
where $\tilde{f}(\V{k},\mu):=f(\V{k})\,\delta_{\lambda,\mu}$,
for $(\V{k},\mu)\in\RR^3\times\ZZ_2$.
Finally, we recall that $\V{m}$ denotes multiplication with $\V{k}$.

On the dense domain $\dom((\Hf^{(j)})^{\nf{3}{2}})$ 
we then have the operator identities
\begin{align}\label{norah1}
[\Hf^{(j)},\ad_\lambda(f),]&=\ad_\lambda(\omega\,f)\,,
\qquad[\Pf^{(j)},\ad_\lambda(f)]=\ad_\lambda(\V{m}\,f)\,,
\end{align}
and, hence,
\begin{align*}
\D\,\ad_\lambda(f)-\ad_\lambda(f)\,\D_\V{p}
&=\valpha\cdot\big\{\ee\SPn{\V{G}_\lambda}{f}
-\ad_\lambda((\V{m}-\V{p})\,f)\big\}\,.
\end{align*}
Thanks to Lemma~\ref{le-dominique} we know that
$\R_{\V{p}}(iy)$ maps $\dom((\Hf^{(j)})^{\nf{3}{2}})$ into itself.
On that domain we thus have
\begin{align*}
\ad_\lambda(f)\R_\V{p}(iy)-\R(iy)\ad_\lambda(f)&
=\R(iy)\valpha\cdot\big\{\ee\SPn{\V{G}_\lambda}{f}
-\ad((\V{m}-\V{p})f)\big\}\R_\V{p}(iy)\,.
\end{align*}
Since $\phi_j\in\dom(\Hf^{(j)})$, the previous relation implies, 
for all $\eta'\in\sC_j$,
\begin{align*}
\Delta(\V{p},f,\eta')&:=\SPn{|\D_\V{p}|\,\eta'}{a_\lambda(f)\,\phi_j}
-\SPn{\ad_\lambda(f)\,\eta'}{|\D|\,\phi_j}
\\
&=
\lim_{\tau\to\infty}\int_{-\tau}^\tau\!
\SPb{\R(iy)\,\valpha\cdot\ad((\V{m}-\V{p})f)\,\R_{\V{p}}(iy)
\,\eta'}{\phi_j}\,\frac{iy\,dy}{\pi}
\\
&\quad
-\ee\SPn{\V{G}_\lambda}{f}\cdot\lim_{\tau\to\infty}\int_{-\tau}^\tau\!
\SPb{\R(iy)\,\valpha\,\R_{\V{p}}(iy)
\,\eta'}{\phi_j}\,\frac{iy\,dy}{\pi}\,,
\end{align*}
where we used the representation \eqref{for-absvT} of the absolute value.
In the integral appearing in the second line
we now write $\R_{\V{p}}(iy)\,y/i=\id-\D_{\V{p}}\,\R_{\V{p}}(iy)$
and apply the formula \eqref{for-sgnT} for the sign function $S:=\sgn(\D)$.
In the third line
we expand the right resolvent as $\R_\V{p}=\R+\R\valpha\cdot\V{p}\R_{\V{p}}$,
apply the formula \eqref{primus} for the derivative of the Hamiltonian.
Proceeding in the way we arrive at
\begin{align*}
\Delta(\V{p},f,\eta')&=
\SPb{S\,\valpha\cdot\ad((\V{m}-\V{p})\,f)\,\eta'}{\phi_j}-J(\V{p},f,\eta')
\\
&\quad
-\ee\SPn{f}{\V{G}_\lambda}\cdot\big\{
\SPn{\eta'}{\nabla\HAM_j\,\phi_j}+\V{I}(\V{p},\eta')\big\}\,.
\end{align*}
Here and below $\nabla\HAM_j$ is evaluated at $\V{P}$ and
\begin{align*}
\V{I}(\V{p},\eta')
&:=\int_\RR\SPb{\R(iy)\,\valpha\,\R(iy)\,(\valpha\cdot\V{p})\,
\R_{\V{p}}(iy)
\,\eta'}{\phi_j}\,\frac{iy\,dy}{\pi}\,,
\\
J(\V{p},f,\eta')&:=
\int_\RR\SPb{\R_{\V{p}}(iy)\,\D_{\V{p}}
\,\eta'}{\valpha\cdot a((\V{m}-\V{p})f)\,\R(-iy)\,\phi_j}\,\frac{dy}{\pi}\,.
\end{align*}
By means of \eqref{zita111} it is easy to see that the latter
integral is absolutely convergent.
In fact, let $\rho>1$ be sufficiently large and set
$\Theta:=\Hf^{(j)}+\rho$ and
$B_{\V{p}}:=\valpha\cdot a((\V{m}-\V{p})\,f)\,\Theta^{-\nf{1}{2}}$.
Then \eqref{zita111} shows that
the composition
$F(y):=\Theta^{\nf{1}{2}}\R(-iy)\Theta^{-\nf{1}{2}}$
is well-defined on $\HR_j$ and $\|F(y)\|\le\const\,(1+y^2)^{-\nf{1}{2}}$.
Moreover, 
$\|B_{\V{p}}\|
\le\const\|(1+1/\omega)^{\nf{1}{2}}(\V{m}-\V{p})\,f\|$ by a standard estimate.
From the representation
\begin{align*}
J(\V{p},f,\eta')&=\int_\RR
\SPb{\{\D_{\V{p}}\,\R_{\V{p}}(iy)\,|\D_{\V{p}}|^{-\nf{1}{2}}\}\,|\D_{\V{p}}|^{\nf{1}{2}}
\,\eta'}{B_{\V{p}}\,F(y)\,\Theta^{\nf{1}{2}}\phi_j}\,\frac{dy}{\pi}\,,
\end{align*}
where the operator in curly brackets satisfies 
$\|\{\cdots\}\|\le\const\,(1+y^2)^{-\nf{1}{4}}$,
we may thus read off that the map
$\sC_j\ni\eta'\mapsto J(\V{p},f,\eta')$
is continuous when $\sC_j$ is equipped with the form
norm of $\HAM_{\V{p}}:=\HAM_j(\V{P}-\V{p})$. Since $\sC_j$ is a form core
for $\HAM_{\V{p}}$ we may hence extend the definition of
$J(\V{p},f,\eta')$ to all $\eta'\in\fdom(\HAM_{\V{p}})$.

We further find by means of \eqref{norah1} and a virial type argument 
\begin{align}\nonumber
\SPb{&(\HAM_{\V{p}}-E_j(\V{P})+|\V{p}|)\,\eta'}{a_\lambda(f)\,\phi_j}
\\
&=\nonumber
\SPb{\HAM_{\V{p}}\eta'}{a_\lambda(f)\,\phi_j}
-\SPb{\ad_\lambda(f)\,\eta'}{\HAM_j(\V{P})\,\phi_j}
+|\V{p}|\,\SPn{\eta'}{a_\lambda(f)\,\phi_j}
\\\label{norah4}
&=\Delta(\V{p},f,\eta')+\SPb{\eta'}{a((|\V{p}|-\omega)f)\,\phi_j}\,.
\end{align}
Since $\phi_j\in\dom(\Hf^{(j)})$ we know that
$a_\lambda(f)\,\phi_j\in\fdom(\Hf^{(j)})$
and, by Lemma~\ref{le-lisa}(iv), the form
domain of $\HAM_{\V{p}}$ is $\fdom(\Hf^{(j)})$.
Taking this and the above remarks on
$J(\V{p},f,\eta')$ into account we conclude
that the first and the last line of \eqref{norah4}
are continuous in $\eta'$ w.r.t. the 
form norm of $\HAM_{\V{p}}$.
Approximating 
$\wh{\RES}_{\V{p}}\,\eta:=\wh{\RES}_j(\V{P},\V{p})\,\eta$,
where $\eta\in\HR_j$ is arbitrary, 
by some sequence in $\sC_j$, which is convergent
w.r.t. the form norm of $\HAM_{\V{p}}$, we thus obtain
\begin{align}\label{arabella}
\SPb{\eta}{a_\lambda(f)\,\phi_j}
&=\Delta(\V{p},f,\wh{\RES}_{\V{p}}\,\eta)
+\SPb{\eta}{\wh{\RES}_{\V{p}}\,a((|\V{p}|-\omega)f)\,\phi_j}\,.
\end{align}
Now, let
$f_{\V{p},\ve}(\V{k}):=h_\ve(\V{k}-\V{p})$, 
$h_\ve(\V{k}):=h(\V{k}/\ve)/\ve^{3}$, for tiny $\ve>0$,
where $h\in C_0^\infty(\{|\V{k}|\le1\},[0,\infty))$
satisfies $\int_{\RR^3}h=1$.
In the next step we insert
the peak function $f_{\V{p},\ve}$ into \eqref{arabella}, 
multiply the resulting expressions with
$g\in C_0^\infty(\RR^3\setminus\{0\})$, and integrate with respect to $\V{p}$.
Proceeding in this way we arrive at
\begin{align}\label{norah5}
\int_{\RR^3}g(\V{p})\,\SPb{\eta}{a_\lambda(f_{\V{p},\ve})\,\phi_j}\,d^3\V{p}
&=\sum_{\ell=1}^5C_{\ell}(\ve)\,,
\end{align}
with (in $C_3(\ve)$ we have $S\,\phi_j\in\dom(a((\V{m}-\V{p})f_{\V{p},\ve}))$
because of \eqref{zita555})
\begin{align*}
C_1(\ve)&:=-\ee\int_{\RR^3}g(\V{p})\,\SPn{f_{\V{p},\ve}}{\V{G}_\lambda}
\SPn{\eta}{\wh{\RES}_{\V{p}}\,\nabla\HAM_j\,\phi_j}\,d^3\V{p}\,,
\\
C_2(\ve)&:=\int_{\RR^3}g(\V{p})\,
\SPb{\eta}{\wh{\RES}_{\V{p}}\,
a((|\V{p}|-\omega)f_{\V{p},\ve})\,\phi_j}\,d^3\V{p}\,,
\\
C_3(\ve)&:=\int_{\RR^3}g(\V{p})\,
\SPb{\eta}{\wh{\RES}_{\V{p}}\,\valpha\cdot a((\V{m}-\V{p})f_{\V{p},\ve})\,S\,\phi_j}
\,d^3\V{p}\,,
\\
C_{4}(\ve)&:=-\ee\int_{\RR^3}g(\V{p})\,\SPn{f_{\V{p},\ve}}{\V{G}_\lambda}
\cdot\V{I}\big(\V{p},\wh{\RES}_{\V{p}}\,\eta\big)\,d^3\V{p}\,,
\\
C_{5}(\ve)&:=\int_{\RR^3}g(\V{p})\,
J\big(\V{p},f_{\V{p},\ve},\wh{\RES}_{\V{p}}\,\eta\big)
\,d^3\V{p}\,.
\end{align*}
It is straightforward to see that
the LHS of \eqref{norah5} converges to $\SPn{\eta}{a_\lambda(\ol{g})\,\phi_j}$,
as $\ve>0$ tends to zero.
Furthermore, 
$|\SPn{\eta}{\wh{\RES}_{\V{p}}\,\nabla\HAM_j\,\phi_j}|
\le\const\,\|\eta\|/|\V{p}|\le\const'\|\eta\|$ on the support of $g$
by \eqref{sergio7}, and 
$|\V{I}(\V{p},\wh{\RES}_{\V{p}}\,\eta)|\le\const|\V{p}|\,
\|\wh{\RES}_{\V{p}}\|\,\|\eta\|\le\const\,\|\eta\|$.
Hence, 
$\V{p}\mapsto g(\V{p})\,\SPn{\eta}{\wh{\RES}_{\V{p}}\,\nabla\HAM_j\,\phi_j}$
and 
$\V{p}\mapsto g(\V{p})\,\V{I}(\V{p},\wh{\RES}_{\V{p}}\,\eta)$
belong to $L^2(\RR^3,\CC^3)$.
Since also $\SPn{f_{\V{p},\ve}}{\V{G}_\lambda}=(h_\ve*\V{G}_\lambda)(\V{p})$
and $h_\ve*\V{G}_\lambda\to\V{G}_\lambda$
in $L^2(\RR^3,\CC^3)$ we conclude that
\begin{align}\label{li8}
\lim_{\ve\searrow0}C_1(\ve)
&=-\ee\int_{\RR^3}g(\V{p})\,{\V{G}_j(\V{p},\lambda)}\cdot
\SPn{\eta}{\wh{\RES}_{\V{p}}\,\nabla\HAM_j\,\phi_j}\,d^3\V{p}\,,
\\\nonumber
\lim_{\ve\searrow0}C_{4}(\ve)
&=\ee\int_\RR\int_{\RR^3} g(\V{k})\,
\SPb{\R(iy)\,(\valpha\cdot\V{G}_j(\V{k},\lambda))\,
\R(iy)\,(\valpha\cdot\V{k})\,\R_{\V{k}}(iy)\times
\\\label{li9}
&\qquad\qquad\times
\wh{\RES}_{\V{k}}\,\eta}{\phi_j}\,d^3\V{k}\,\frac{y\,dy}{i\pi}\,.
\end{align}
Next, we show that $C_{5}(\ve)$ goes to zero:
We have
$\supp(g)\subset\{r\le|\V{p}|\le1/r\}$, for some $r>0$, 
and we shall always assume that
$0<\ve\le r/2$.
Then $\V{p}\in\supp(g)$ and $h_\ve(\V{k}-\V{p})\not=0$
implies $1/|\V{p}|\le1/r$ and $1\le2|\V{k}|/r$. 
By Fubini's theorem 
we thus have
\begin{align*}
C_{5}(\ve)&\le\int_\RR\int_{\RR^3}\int_{\RR^3}|g(\V{p})|\,h_\ve(\V{k}-\V{p})
\,|\V{k}-\V{p}|\,\Big|\SPB{\valpha\,\{\D_{\V{p}}\,\R_{\V{p}}(iy)\,
\wh{\RES}_{\V{p}}\}\,\eta}{\times
\\
&\qquad
\times a_\lambda(\V{k})\,\Theta^{-\nf{1}{2}}\,F(y)\,\Theta^{\nf{1}{2}}\phi_j}
\Big|\,d^3\V{k}\,d^3\V{p}\,\frac{dy}{\pi}\,.
\end{align*}
In the previous expression we have $|\V{k}-\V{p}|\le\ve$,
if $h_\ve(\V{k}-\V{p})\not=0$.
Furthermore,
$$
\|\{\cdots\}\|\le\|\D_{\V{p}}\,\R_{\V{p}}(iy)\,|\D_{\V{p}}|^{-\nf{1}{2}}\|\,
\||\D_{\V{p}}|^{\nf{1}{2}}\wh{\RES}_{\V{p}}\|
\le\const\,(1+y^2)^{-\nf{1}{4}}/r\,.
$$
We estimate the remaining factors of the integrand
by Young's inequality,
$|\SPn{u}{v*w}|\le\const\,\|u\|_2\|v\|_1\|w\|_2$,
applied to the $d^3\V{k}d^3\V{p}$-integration.
In this way we obtain
\begin{align*}
|C_{5}(\ve)|\le&\frac{\const\,\ve}{r}\,\|h_\ve\|_1\,\|g\|_2\,\|\eta\|
\\
&\cdot\int_\RR
\bigg\{(\nf{2}{r})\!\int\limits_{|\V{k}|\ge \nf{r}{2}}\!\!|\V{k}|\,
\big\|a_\lambda(\V{k})\,\Theta^{-\nf{1}{2}}
F(y)\,\Theta^{\nf{1}{2}}\phi_j\big\|^2d^3\V{k}\bigg\}^{\nf{1}{2}}\!
\frac{dy}{(1+y^2)^{\nf{1}{4}}}\,.
\end{align*}
Here, the integral $\int_{|\V{k}|\ge \nf{r}{2}}\ldots d^3\V{k}$
is not greater than
$$
\|(\Hf^{(j)})^{\nf{1}{2}}\Theta^{-\nf{1}{2}}F(y)\,\Theta^{\nf{1}{2}}\phi_j\|^2
\le\|F(y)\|^2\,\|\HAM_j^{\nf{1}{2}}\phi_j\|^2
\le\const\,(1+y^2)^{-1}.
$$
Since also $\|h_\ve\|_1=1$
we conclude that $C_{5}(\ve)\to0$, as $\ve\searrow0$.

Obviously, $C_2$ and $C_3$ can also be treated by means of
Young's inequality and we easily verify that
$\lim_{\ve\searrow0}C_2(\ve)=\lim_{\ve\searrow0}C_3(\ve)=0$,
again using that $||\V{p}|-|\V{k}||\le|\V{k}-\V{p}|\le\ve$,
when $h_\ve(\V{k}-\V{p})\not=0$. For $C_3$ we actually find
$$
C_3(\ve)\le\const\,(\nf{\ve}{r})\,\|g\|_2\,\|h_\ve\|_1\,\|\eta\|
\,(\nf{2}{r})^{\nf{1}{2}}\,\|(\Hf^{(j)})^{\nf{1}{2}}\,S\,\phi_j\|\,,
\quad0<\ve\le r/2\,,
$$
where last norm 
$\|(\Hf^{(j)})^{\nf{1}{2}}\,S\,\phi_j\|$ is well-defined and bounded 
because of \eqref{zita555} and $\phi_j\in\dom(\Hf^{(j)})$.
 
Putting everything together we see that $\SPn{\eta}{a_\lambda(\ol{g})\,\phi_j}$
is equal to the sum of terms on the RHS of
\eqref{li8} and \eqref{li9}.
As this holds true, for every
$g\in C_0^\infty(\RR^3\setminus\{0\})$, we conclude that
\begin{align*}
\SPb{\eta&}{a_\lambda(\V{k})\,\phi_j}
=
-\ee\,{\V{G}_j(\V{k},\lambda)}\cdot
\SPn{\eta}{\RES_{\V{k}}(|\V{k}|)\,\nabla\HAM_j\,\phi_j}
\\
&+\ee\int_\RR
\SPb{\R(iy)\,(\valpha\cdot\V{G}_j(\V{k},\lambda))\,
\R(iy)\,(\valpha\cdot\V{k})\,\R_{\V{k}}(iy)\,
\wh{\RES}_{\V{k}}\,\eta}{\phi_j}\,\frac{y\,dy}{i\pi}\,,
\end{align*}
for all $\V{k}\in\RR^3\setminus N_\eta$, where $N_\eta$ is some
$\eta$-dependent zero set.
Applying this results to all $\eta$ in some
countable dense domain in $\HR$ we obtain \eqref{for-a(k)}.
\end{proof}


\section{Operators on Fock spaces}\label{app-Fock}

\noindent
{\bf Bosonic Fock spaces}

\smallskip

\noindent
Recall the definition of the annuli $\cA_k^j$ and $\cA_j$ in \eqref{def-Akj}.
The bosonic Fock space over $\HP_k^j=L^2(\cA_k^j\times\ZZ_2)$
is a direct sum of
$n$-particle subspaces,
\begin{equation*}
\sF_k^j:=\bigoplus_{n=0}^\infty[\sF_k^j]^{(n)}\ni
\psi=(\psi^{(n)})_{n=0}^\infty\,,
\end{equation*}
which are given as
$[\sF_k^j]^{(0)}:=\CC$ and 
$[\sF_k^j]^{(n)}:=\cS_n L^2((\cA_k^j\times\ZZ_2)^n)$,
$n\in\NN$, where $\cS_n$ is the orthogonal projection onto the permutation
symmetric functions,
\begin{align*}
\cS_n\psi^{(n)}(\V{k}_1,\lambda_1,\ldots,\V{k}_n,\lambda_n)
:=\frac{1}{n!}\sum_{\pi\in\mathfrak{S}_n}
\psi^{(n)}(\V{k}_{\pi(1)},\lambda_{\pi(1)},\ldots,
\V{k}_{\pi(n)},\lambda_{\pi(n)})\,,
\end{align*}
almost everywhere,
for $\psi^{(n)}\in L^2((\cA_k^j\times\ZZ_2)^n)$,
$\mathfrak{S}_n$ denoting the group of permutations of $\{1,\ldots,n\}$.
$\sF_j$ is defined in the same way but with $\cA_k^j$ replaced by $\cA_j$.
Finally, we write $\sF=\sF_\infty$, i.e. in the definition of
$\Fock{}{}$ the set $\cA_k^j$ is replaced by $\RR^3$. 

\smallskip

\noindent
{\bf Second quantized multiplication operators}

\smallskip

\noindent
The second quantization
of some measurable function $\vo:\cA_k^j\times\ZZ_2\to\RR$
is the direct sum
$d\Gamma(\vo):=\bigoplus_{n=0}^\infty d\Gamma^{(n)}(\vo)$,
where $d\Gamma^{(0)}(\vo):=0$, and 
$d\Gamma^{(n)}(\vo)$ is the maximal operator in $[\sF_k^j]^{(n)}$
of multiplication with the permutation symmetric
function $(\V{k}_1,\lambda_1,\ldots,\V{k}_n,\lambda_n)
\mapsto\vo(\V{k}_1,\lambda_1)+\dots+\vo(\V{k}_n,\lambda_n)$.
For instance, $N:=d\Gamma(1)$ is the number operator.

\smallskip

\noindent
{\bf Creation and annihilation operators}

\smallskip

\noindent
For $f\in\HP_k^j$, the bosonic
creation and annihilation operators,
$\ad(f)$ and $a(f)$, are given as follows:
Defining
$\ad(f)^{(n)}:[\sF_k^j]^{(n)}\to[\sF_k^j]^{(n+1)}$ 
by $\ad(f)^{(n)}\,\psi^{(n)}:=\sqrt{n+1}\,\cS_{n+1}(f\otimes\psi^{(n)})$,
we set
$\ad(f)\,\psi:=
(0,\ad(f)^{(0)}\,\psi^{(0)},\ad(f)^{(1)}\,\psi^{(1)},\ldots\;)$,
for all $\psi=(\psi^{(0)},\psi^{(1)},\ldots\;)\in\sF_k^j$
such that the right side again belongs to $\sF_k^j$.
Then $a(f)$ is the adjoint operator,
$a(f):=\ad(f)^*$.
Straightforward computations yield the CCR \eqref{CCR}
on, e.g., the domain of $N^2$.

\smallskip

\noindent
{\bf Mapping properties of Weyl operators}

\begin{lem}\label{le-will}
Let $j\in\NN_0\cup\{\infty\}$, $\nu\ge1/2$, and let $f\in\HP_j$
satisfy $\omega^{-\nf{1}{2}}f\in\HP_j$ and $\omega^\nu\,f\in\HP_j$.
Then both $e^{i\vp(f)}$ and $e^{i\vo(f)}$ map 
$\dom((\Hf^{(j)})^\nu)$ into itself.
\end{lem}
 
\begin{proof}
Let $A$ be either $\vp(f)$ or $\vo(f)$ and set
$\Theta:=(\Hf^{(j)}+\lambda)^\nu$, for some $\lambda\ge1$.
According to Lemma~2 of \cite{GeorgescuGerard1999}
and the condition (ABG$'$) recalled on Page~276
of \cite{GeorgescuGerard1999} it suffices to show that
\begin{enumerate}
\item[(i)] 
$|\SPn{\Theta^{-1}\,\psi}{A\,\psi}-\SPn{A\,\psi}{\Theta^{-1}\,\psi}|
\le\const\,\|\psi\|^2$, for all $\psi\in\dom(A)$.
\item[(ii)]
The quadratic form 
$\SPn{\Theta\,\phi}{A\,\psi}-\SPn{A\,\phi}{\Theta\,\psi}$,
$\phi,\psi\in\dom(A)\cap\dom(\Theta)=\dom(\Theta)$,
is bounded w.r.t. the topology on $\dom(\Theta)$
and the unique operator $[\Theta,A]_0\in\cB(\dom(\Theta),\dom(\Theta)^*)$
representing it obeys $\Ran([\Theta,A]_0)\subset\Fock{j}{}$.
\end{enumerate}
Since $\Theta^{-1}$ maps into the domain of $A$ and
$A\,\Theta^{-1}$ is a bounded operator condition (i)
is trivially fulfilled.
(It implies $\Theta\in C^1(A)$.)
Furthermore, under the above assumptions on $f$
Lemma~3.1 in \cite{MatteStockmeyer2009a} shows
that $[a^\sharp(f),\Theta^{-1}]\,\Theta$ defines a bounded
operator on the dense domain $\sC_j$, if $\lambda$ is large enough.
Therefore,
$|\SPn{\Theta\,\phi}{A\,\psi}-\SPn{A\,\phi}{\Theta\,\psi}|
\le\const\,\|\phi\|\,\|\Theta\,\psi\|$,
for all $\phi,\psi\in\sC_j$.
Since $\Theta$ is essentially self-adjoint on
$\sC_j$ and $\|A\,\phi\|\le\const\,\|\Theta\,\phi\|$, this implies (ii). 
\end{proof}

\end{appendix}

\subsection*{Acknowledgment}
We thank the Erwin Schr\"odinger Institute for Mathematical Physics in Vienna
for their kind hospitality. 
O.M. also thanks Li Chen and the Department of Mathematical
Sciences at the Tsinghua University in Beijing for
their genereous hospitality.
Parts of this work have been prepared while
M.K. was working at the FernUniversit\"at Hagen
and while O.M. was working at the Technical University of Clausthal 
and at the Technical University of Munich.
This work has been partially
supported by the DFG (SFB/TR12).



\vspace{0.5cm}

\begin{multicols}{2}

\noindent
{\sc Martin K\"onenberg}\\ 
Fakult\"at f\"ur Physik\\
Universit\"at Wien\\
Boltzmanngasse 5\\
1090 Vienna, Austria.\\
{\tt martin.koenenberg@univie.ac.at}
 
\bigskip

\noindent
{\sc Oliver Matte}\\ 
Mathematisches Institut\\
Ludwig-Maximilians-Universit\"at\\
Theresienstra{\ss}e 39\\
80333 M\"unchen, Germany.\\
{\tt matte@math.lmu.de}

\end{multicols}

\end{document}